\documentclass[11pt,letter]{article}

\usepackage[margin=1in]{geometry}
\usepackage{amsmath,amssymb,amsthm}
\usepackage{cite}

\newcommand{\notf}{\lnot f}
\newcommand{\notP}{\lnot P}
\newcommand{\cut}{cut}
\newcommand{\calP}{\mathcal{P}}

\DeclareMathOperator{\wdist}{wdist}

\newcommand{\plssize}{\mathsf{pls}\text{-}\mathsf{size}}
\DeclareMathOperator{\false}{{\scriptstyle{FALSE}}}
\DeclareMathOperator{\true}{{\scriptstyle{TRUE}}}
\DeclareMathOperator{\disj}{\mathrm{DISJ}}
\DeclareMathOperator{\eq}{\mathrm{EQ}}
\newcommand{\cgst}{\textsc{congest}}
\newcommand{\local}{\textsc{local}}

\newcommand{\term}{\texttt{Term}}
\DeclareFontShape{OT1}{cmr}{bx}{sc}{<-> cmbcsc10}{}

\newcommand{\remove}[1]{}

\usepackage{xcolor}

\usepackage[utf8]{inputenc}
\usepackage[english]{babel}
\usepackage{comment}

\newcommand{\C}{{\cal{C}}}

\newcommand{\ip}[1]{\left}
\usepackage{xspace}

\usepackage{graphicx}
\usepackage{hyperref}

\makeatletter
\newtheorem*{rep@theorem}{\rep@title}
\newcommand{\newreptheorem}[2]{%
	\newenvironment{rep#1}[1]{%
		\def\rep@title{#2 \ref{##1}}%
		\begin{rep@theorem}}%
		{\end{rep@theorem}}}
\makeatother

\newenvironment{definition-repeat}[1]{\begin{trivlist}
		\item[\hspace{\labelsep}{\bf\noindent Definition \ref{#1} }]\em }%
	{\end{trivlist}}

\newenvironment{lemma-repeat}[1]{\begin{trivlist}
		\item[\hspace{\labelsep}{\bf\noindent Lemma \ref{#1} }]\em }%
	{\end{trivlist}}
\newenvironment{theorem-repeat}[1]{\begin{trivlist}
		\item[\hspace{\labelsep}{\bf\noindent Theorem \ref{#1} }]\em }%
	{\end{trivlist}}

\newcommand{\qedsymb}{\qed}
\newenvironment{proofof}[1]{\begin{trivlist}
		\item[\hspace{\labelsep}{\bf\noindent Proof of #1: }]
	}{\qedsymb\end{trivlist}}

\newtheorem{theorem}{Theorem}[section]
\newtheorem{claim}{Claim}[section]
\newtheorem{lemma}{Lemma}[section]

\newtheorem{definition}{Definition}[section]
\newtheorem{corollary}{Corollary}[section]

\DeclareMathOperator{\bin}{bin}
\newcommand{\size}[1]{\ensuremath{\left|#1\right|}}
\newcommand{\set}[1]{\left\{ #1 \right\}}
\usepackage{lipsum}

\usepackage{amsmath}
\usepackage{amsfonts}
\usepackage{comment}
\usepackage{amsthm}
\usepackage[inline]{enumitem}
\graphicspath{ {./images/} }

\usepackage{latexsym}
\usepackage{amssymb}
\usepackage{comment}
\usepackage{mathtools}
\usepackage{tikz}

\theoremstyle{remark}

\theoremstyle{definition}

\usepackage{titling}
\usepackage{authblk}

\title{Hardness of Distributed Optimization}

\author[1]{Nir Bachrach} 
\author[1]{Keren Censor-Hillel} 
\author[1]{Michal Dory} 
\author[1]{Yuval Efron} 
\author[1]{Dean Leitersdorf} 
\author[2]{Ami Paz}
\affil[1]{Technion-Israel Institute of Technology}
\affil[2]{IRIF, CNRS, Paris Diderot University}

\date{}
\begin{document}

\begin{titlepage}

\maketitle

\begin{abstract}
 This paper studies lower bounds for fundamental optimization problems in the \cgst{} model.
We show that solving problems \emph{exactly} in this model can be a hard task, by providing $\tilde{\Omega}(n^2)$ lower bounds for cornerstone problems, such as minimum dominating set (MDS), Hamiltonian path, Steiner tree and max-cut. These are almost tight, since all of these problems can be solved optimally in $O(n^2)$ rounds. Moreover, we show that even in bounded-degree graphs and even in simple graphs with maximum degree 5 and logarithmic diameter, it holds that various tasks, such as finding a maximum independent set (MaxIS) or a minimum vertex cover, are still difficult, requiring a near-tight number of $\tilde{\Omega}(n)$ rounds.

Furthermore, we show that in some cases even \emph{approximations} are difficult, by providing an $\tilde{\Omega}(n^2)$ lower bound for a $(7/8+\epsilon)$-approximation for MaxIS, and a nearly-linear lower bound for an $O(\log{n})$-approximation for the $k$-MDS problem for any constant $k \geq 2$, as well as for several variants of the Steiner tree problem.

Our lower bounds are based on a rich variety of constructions that leverage novel observations, and reductions among problems that are specialized for the \cgst{} model.
However, for several additional approximation problems, as well as for exact computation of some central problems in $P$, such as maximum matching and max flow, we show that such constructions \emph{cannot} be designed, by which we exemplify some limitations of this framework.
\end{abstract}

\thispagestyle{empty}
\end{titlepage}

\newpage 
\tableofcontents

\newpage
\section{Introduction}

Optimization problems are cornerstone problems in computer science, for which finding exact and approximate solutions is extensively studied in various computational settings. Since optimization problems are fundamental for a variety of computational tasks, mapping their trade-offs between time complexity and approximation ratio is a holy-grail, especially for those that are NP-hard.

Distributed settings share this necessity of resolving the complexity of exact and approximate solutions for optimization problems, and a rich landscape of complexities is constantly being explored. However, distributed settings exhibit very different behavior, compared with their sequential counterpart, in terms of what is efficient and what is hard. Here, we focus on the \cgst{} model in which $n$ vertices communicate synchronously over the underlying network graph, using an $O(\log{n})$-bit bandwidth~\cite{Peleg00}.
Since the local computation of vertices is not polynomially bounded, hardness results in the sequential setting do not translate to hardness results in the distributed one. In particular, any natural graph problem can be solved in the \cgst{} model in $O(m)$ rounds, $m$ being the number of edges, by letting the vertices learn the whole graph. For some problems, such as MaxIS and finding the chromatic number~\cite{DBLP:conf/wdag/Censor-HillelKP17}, this na\"{\i}ve solution is known to be nearly-optimal, whereas for other problems more efficient solutions exist. On the other hand, there do exist problems with a polynomial sequential complexity, which require $\tilde\Omega(m)$ rounds in the \cgst{} model, such as deciding whether the graph contains a cycle of a certain length and weight~\cite{DBLP:conf/wdag/Censor-HillelKP17}.

In the sequential setting, finding an exact solution for some problems, such as minimum dominating set (MDS) or a maximum independent set (MaxIS), is known to be NP-hard \cite{DBLP:books/daglib/p/Karp10}.
In such cases, it is sometimes possible to obtain efficient approximations, such as an $O(\log{\Delta})$-approximation for MDS, where $\Delta$ is the maximum degree in the graph~\cite{DBLP:books/daglib/0004338}. However, in some cases, even obtaining an approximation is hard.
For MDS, any approximation better than logarithmic is hard to obtain~\cite{lund1994hardness}. MaxIS does not admit even an $O(n^{1-\epsilon})$ approximation \cite{Hastad96}.
In the \cgst{} model, there are polylogarithmic $O(\log \Delta)$-approximations for MDS~\cite{jia2002efficient, kuhn2005constant, DBLP:journals/jacm/KuhnMW16}, and it is also known that obtaining such an approximation requires at least a polylogarithmic time. More specifically, there are lower bounds of $\Omega(\sqrt{\log n/\log\log n})$ and $\Omega(\log\Delta/\log\log\Delta)$ \cite{DBLP:journals/jacm/KuhnMW16}. However, currently nothing else is known with respect to better approximations or exact solutions. For the MaxIS problem, there are efficent $(\Delta +1)/2$ and $0.529\Delta$ approximations for the unweighted and weighted cases, respectively \cite{DBLP:conf/sirocco/BoppanaHR18}. However, the complexity of achieving any better approximations is not known. Solving it \emph{exactly} requires $\tilde{\Omega}(n^2)$ rounds \cite{DBLP:conf/wdag/Censor-HillelKP17}. In the closely related {\local}  model, where the size of messages is not bounded, $(1+\epsilon)$-approximations for both problems can be obtained in polylogarithmic time~\cite{ghaffari2017complexity}.

The curious aspect of the huge gaps that are present in our current understanding of various optimization and approximation complexities in the \cgst{} model is that \emph{we do not have any hardness conjectures to blame these gaps on}. This raises the natural question: can we obtain better approximations efficiently in the \cgst{} model?

For problems in P, in many cornerstone cases, such as min cut, max flow and maximum matching, we have efficient $(1+\epsilon)$-approximations \cite{LotkerPP15,DBLP:journals/siamcomp/GhaffariKKLP18,DBLP:conf/wdag/NanongkaiS14}, but the complexity of \emph{exact} computation is still open. Many additional questions are open with respect to various optimization problems.

The contributions of this paper are three-fold, providing (i) novel techniques for nearly-tight lower bounds for exact optimizations, (ii) advanced approaches for nearly-tight lower bounds for approximations, and (iii) new methods for showing limitations of the main lower-bound framework.

\subsection{Our contributions, the challenges, and our techniques}

\paragraph{Lower bounds for exact computation.}
We show that in many cases, solving problems \emph{exactly} in the \cgst{} model is hard, by providing many new $\tilde{\Omega}(n^2)$ lower bounds for fundamental optimization problems, such as MDS, max-cut, Hamiltonian path, Steiner tree and minimum 2-edge-connected spanning subgraph (2-ECSS).
Such results were previously known only for the minimum vertex cover (MVC), MaxIS and minimum chromatic number problems \cite{DBLP:conf/wdag/Censor-HillelKP17}. Our results are inspired by \cite{DBLP:conf/wdag/Censor-HillelKP17}, but combine many new technical ingredients. In particular, one of the key components in our lower bounds are reductions between problems. After having a lower bound for MDS, a cleverly designed reduction allows us to build a new lower bound construction for Hamiltonian path. These constructions serve as a basis for our constructions for the Steiner tree and minimum 2-ECSS. We emphasize that we \emph{cannot} use directly known reductions from the sequential setting, but rather we must create reductions that can be applied efficiently on lower bound constructions.

To demonstrate the challenge, we now give more details about the general framework.
We use the well-known framework of reductions from 2-party communication complexity, as originated in~\cite{DBLP:journals/siamcomp/PelegR00} and used in many additional works, e.g.,~\cite{Dassarmaetal12, FrischknechtHW12, DBLP:conf/spaa/FischerGKO18, DBLP:conf/podc/Censor-HillelD18, AbboudCHK16, DBLP:conf/wdag/CzumajK18}.
In communication complexity, two players, Alice and Bob, receive private input strings and their goal is to solve some problem related to their inputs, for example, decide whether their inputs are disjoint, by communicating the minimum number of bits possible. To show a lower bound for the \cgst{} model, the high-level idea is to create a graph that satisfies some required property, for example have an MDS of a certain size, iff the input strings satisfy some property.
If a fast algorithm in the \cgst{} exists, Alice and Bob can simulate it and solve the communication problem.
Then, lower bounds from communication complexity translate to lower bounds in the \cgst{} model. The exact lower bound we can show depends on certain parameters such as the size of the graph, the size of the inputs and the size of the cut between the parts of the graph that the players simulate. An attempt to use the known reduction from MVC to MDS together with the $\tilde{\Omega}(n^2)$ lower bound for MVC from \cite{DBLP:conf/wdag/Censor-HillelKP17} faces a complication: This reduction requires adding a new vertex for each edge in the original graph, which blows up the size of the graph with respect to the inputs, and allows showing only a nearly-linear lower bound. For similar reasons, known reductions from MVC to Hamiltonian cycle and Steiner tree cannot show any super-linear lower bound. Nevertheless, we show that in some cases reductions can be a powerful tool in providing new $\tilde{\Omega}(n^2)$ bounds, but they need to be designed carefully in order to preserve certain parameters of the graph, such as the number of vertices and the size of the cut between the two players.

A succinct summary of our results in this section is the following.
\begin{itemize}
    \item We show an $\tilde{\Omega} (n^2)$ lower bound for solving MDS, weighted max-cut, Hamiltonian path and cycle, minimum Steiner tree, and unweighted 2-ECSS on general graphs. For unweighted max-cut, we show an $\tilde{O} (n)$ algorithm for computing a $(1-\epsilon)$-approximation for max-cut in general graphs, for all $\epsilon>0$.
\end{itemize}

\paragraph{Lower bounds in bounded-degree graphs.} 
In many cases, the graph over which one needs to solve a certain problem is not a worst-case instance, but rather is drawn from a specific graph family that does allow efficient solutions. For example, while we show that finding a Hamiltonian path requires $\tilde{\Omega}(n^2)$ rounds in the worst case, in random graphs $G_{n,p}$, there exist fast algorithms, with the exact complexity depending on the probability $p$~\cite{Turau18, ChatterjeeFPP18, DBLP:conf/wdag/GhaffariL18}.
When we focus on bounded-degree graphs, there exist efficient \emph{constant} approximations for many optimization problems, such as MaxIS and MDS, whereas in general graphs such results are currently not known in the \cgst{} model.  We show that when it comes to \emph{exact} computation, even in bounded-degree graphs many problems are still difficult. Specifically, to solve MaxIS or MVC in bounded-degree graphs one would need $\tilde{\Omega}(n)$ rounds, and this holds even in graphs with logarithmic diameter and maximum degree 5. A similar result is shown for MDS. This is nearly-optimal, since all these problems can be solved in $O(m)=O(n)$ rounds in these graphs. Our lower bound here is again based on reductions, but this time not necessarily between graph problems. To show a lower bound for MaxIS we use a sequence of reductions between MaxIS and max 2SAT instances. Replacing a graph by a CNF formula $\phi$ is useful since it allows us to use the power of expander graphs and replace $\phi$ by a new equivalent CNF formula $\phi'$ where each variable appears only a constant number of times. This reduction is inspired by \cite{papadimitriou1991optimization, lecture} and is the main ingredient that allows us eventually to convert our graph to a bounded-degree graph.
Once we have a lower bound for MaxIS, a lower bound for MVC and MDS is obtained using standard reductions between the problems.

A succinct summary of our results in this section is the following.
\begin{itemize}
    \item We show an $\tilde{\Omega} (n)$ lower bound for solving MVC, MDS, MaxIS, and weighted 2-spanner on bounded degree graphs. 
\end{itemize}

\paragraph{Hardness of approximation.}
While solving problems \emph{exactly} seems to be a difficult task, one can hope to find fast and efficient approximation algorithms. In the \cgst{} model, currently the best efficient approximation algorithms known for many problems achieve the same approximation factors as the best approximations known for polynomial sequential algorithms. An intriguing question is whether better approximations can be obtained efficiently. As a first step towards answering this question, we show that in some cases even just \emph{approximating} the optimal solution is hard.

The challenge in showing such a lower bound is that we need to create a \emph{gap}. It is no longer enough that the graph satisfies some predicate iff the inputs are, for example,  disjoint, but rather we want that the size of the optimal solution would \emph{change dramatically} according to the inputs. Creating gaps is also crucial in showing inapproxiambility results in the sequential setting, a prime example for this is the PCP theorem which is a key tool for creating such gaps. In the distributed setting, we may need a more direct approach. Several approaches to create such gaps are shown in previous work. In weighted problems, sometimes we can use the weights to create a gap, as done in the constructions of Das Sarma et al. \cite{Dassarmaetal12}.
In some cases, the construction itself allows showing a gap. For example, if the chromatic number of a graph is either at most $3c$ or at least $4c$ depending on the inputs, it shows a lower bound for a $(4/3-\epsilon)$-approximation \cite{DBLP:conf/wdag/Censor-HillelKP17}. Similar ideas are used in lower bounds for approximating the diameter \cite{AbboudCHK16, HolzerP14, FrischknechtHW12, holzer2012optimal}. Another option is to reduce from a problem in communication complexity that already embeds a gap in it, such as the \emph{gap disjointness} problem \cite{DBLP:conf/podc/Censor-HillelD18}, or to design a specific construction that produces a gap, as in the lower bound for directed $k$-spanners \cite{DBLP:conf/podc/Censor-HillelD18}.

We contribute two new techniques for this toolbox. Namely, we show that \emph{error-correcting codes} and \emph{probabilistic methods} are useful for creating gaps also in the \cgst{} model.
Based on these, we show that obtaining a $(7/8+\epsilon)$-approximation for MaxIS requires $\tilde{\Omega}(n^2)$ rounds. As we explain in Section \ref{section:limitations(b)}, although MaxIS may be very difficult to approximate, we cannot use the Alice-Bob framework to show any lower bound for approximation better than $1/2$. For other problems, however, we are able to show a lower bound for a stronger approximation. Specifically, we show a near-linear lower bound for obtaining an $O(\log{n})$-approximation for the $k$-MDS problem for $k \geq 2$, and for several variants of the Steiner tree problem. Here, $k$-MDS is the problem of finding in a given vertex weighted graph $G=(V,E,w)$, a minimum weight set $S\subseteq V$ such that for all $v\in V$, either $v\in S$, or $d(v,S)\leq k$.  We also show such a result for MDS, but with some restrictions on the algorithm. For general algorithms, we show in Section \ref{section:limitations(b)} that we cannot use the Alice-Bob framework to show any hardness result for approximation above 2. This simply follows from the fact that if each one of Alice and Bob solves the problem optimally on its part, the union of the solutions gives a 2-approximation.

Our results demonstrate a clear separation between the \cgst{} and \local{} models, since in the latter there are efficient $(1+\epsilon)$-approximations for MaxIS and $k$-MDS~\cite{ghaffari2017complexity}.\footnote{The algorithm for $k$-MDS follows from algorithm for MDS, since in the \local{} model we can simulate an MDS algorithm in the graph $G^k$.} Such a separation for a \emph{local approximation problem}, a problem whose approximate solution does not require diameter many rounds when the message size is unbounded, was previously known only for approximating spanners \cite{DBLP:conf/podc/Censor-HillelD18}.

A succinct summary of our results in this section is the following.
\begin{itemize}
    \item  We show an $\tilde{\Omega} (n^2)$ and an $\tilde{\Omega} (n)$ lower bounds for computing a $(\frac{7}{8}+\epsilon )$-approximation for MaxIS and a  $(\frac{5}{6}+\epsilon )$-approximation for MaxIS, respectively, for any constant $\epsilon>0$.  We also show a nearly-linear lower bound for computing an $O(\log n)$-approximation for weighted $k$-MDS for all $k\geq 2$. We show similar results also for several variants of the Steiner tree problem. In addition, we show such results for weighted MDS, assuming some restrictions on the algorithm. 
\end{itemize}

\paragraph{Limitations.} Finally, we study the limitations of this general lower bound framework. While it is capable of providing many near-quadratic lower bounds for \emph{exact and approximate} computations, we show that sometimes it is limited in showing hardness of approximation. In addition, we prove impossibility of using this framework for providing lower bounds for exact computation for several central problems in P, such as maximum matching, max flow, min $s$-$t$ cut and weighted $s$-$t$ distance. Interestingly, we also show it cannot provide strong lower bounds for several \emph{verification} problems, which stands in sharp contrast to known lower bounds for these problems~\cite{Dassarmaetal12}. This implies that using a fixed cut as in our paper, is provably weaker than allowing a changing cut as in~\cite{Dassarmaetal12}.

One tool for showing such results is providing a protocol that allows Alice and Bob solve the problem by communicating only a small number of bits. Such ideas are used in showing the limitation of this framework for obtaining any lower bound for triangle detection~\cite{DruckerKO13}, any super-linear lower bound for weighted APSP~\cite{DBLP:conf/wdag/Censor-HillelKP17} (recently proven to have a linear solution~\cite{BernsteinN18}), and any lower bound larger than $\Omega(\sqrt{n})$ for detecting 4-cliques~\cite{DBLP:conf/wdag/CzumajK18}.

We push this idea further, showing that a \emph{non-deterministic} protocol for Alice and Bob, which may be much easier to establish, can imply the limitations of the technique. We also show how to obtain such protocols using a connection to \emph{proof labeling schemes} (PLS).

\subsection{Preliminaries}
\label{sec:prelim(b)}

We denote by $[n]$ the set $\set{0,\ldots,n-1}$. To prove lower bounds on the number of rounds necessary in order to solve a distributed problem in the \cgst{} model, we use reductions from two-party communication complexity problems. In what follows, we give the required definitions and the main reduction mechanism.
\subsection{Communication Complexity}\label{cc}

In the two-party communication complexity setting~\cite{KushilevitzN:book96}, there is a function  $f:\{0,1\}^K\times\{0,1\}^K\to\{\true ,\false \}$, and two players, Alice and Bob, who are given two input strings, $x,y\in\{0,1\}^K$, respectively, and need to compute $f(x,y)$ by exchanging bits according to some protocol $\pi$. The \emph{communication complexity} $CC(\pi)$ of $\pi$ is the maximal number of bits Alice and Bob exchange in $\pi$, taken over all possible pairs of $K$-bit strings $(x,y)$.

For our lower bounds, we consider deterministic and randomized protocols. To show limitations of obtaining lower bounds we also consider nondeterministic protocols, whose discussion we defer to Section~\ref{section:limitations(b)}. In a randomized protocol, Alice and Bob can generate truly random bits of their own, and the final output of Alice and Bob needs to be correct (according to $f$) with probability at least $2/3$ over the random bits generated by both players.

The \emph{deterministic communication complexity} $CC(f)$ of $f$ is the minimum $CC(\pi)$, taken over all deterministic protocols $\pi$ for computing $f$. The \emph{randomized communication complexity} $CC^R(f)$ is defined analogously.

The main communication complexity problem that we use for our lower bounds is \emph{set disjointness}, $\disj_K$, which is defined as $\disj_K(x,y)=\false$ if and only if there is an index $i\in\{0,\ldots,K-1\}$ such that $x_i=y_i=1$. It is known that $CC(\disj_K)=\Omega(K)$ and $CC^R(\disj_K)=\Theta(K)$~\cite[Example 3.22]{KushilevitzN:book96}. The latter holds even if Alice and Bob are allowed to generate shared truly random bits.

\subsection{Family of Lower Bound Graphs}\label{lowerboundsgraphs}
To formalize the reductions, we use the following definition which is taken from \cite{DBLP:conf/wdag/Censor-HillelKP17}.

\begin{definition}
\label{def: family of lb graphs}
	\textsf{Family of Lower Bound Graphs}\\
	Given integers $K$ and $n$, a Boolean function $f:\{0,1\}^{K} \times \{0,1\}^{K} \to \{\true ,\false \}$ and some Boolean graph property or predicate denoted $P$,
	a set of graphs $\set{G_{x,y}=(V,E_{x,y}) \mid x,y\in \{0,1\}^K}$ is called a \emph{family of lower bound graphs} with respect to $f$ and $P$ if the following hold:
	\begin{enumerate}
		\item The set of vertices $V$ is the same for all the graphs in the family, and we denote by $V_A,V_B$ a fixed partition of the vertices.
		\item Given $x,y\in \{0,1\}^{K}$, the only part of the graph which is allowed to be dependent on $x$ (by adding edges or weights, no adding vertices) is $G[V_A]$.
		\item Given $x,y\in \{0,1\}^{K}$, the only part of the graph which is allowed to be dependent on $y$ (by adding edges or weights, no adding vertices) is $G[V_B]$.
		\item $G_{x,y}$ satisfies $P$ if and only if $f(x,y)=\true$.
	\end{enumerate}
    The set of edges $E(V_A,V_B)$ is denoted by $E_{\cut}$, and is the same for all graphs in the family.
\end{definition}
We use the following theorem whose proof can be found in \cite{DBLP:conf/wdag/Censor-HillelKP17}.

\begin{theorem}
\label{generallowerboundtheorem}
	Fix a function $f:\{0,1\}^{K} \times \{0,1\}^{K} \to \{\true ,\false \}$ and a predicate $P$. If there exists a family of lower bound graphs $\{G_{x,y} \}$ w.r.t $f$ and $P$ then any deterministic algorithm for deciding $P$ in the \cgst{} model takes $\Omega(CC(f)/(\size{E_{\cut}}\log n))$  rounds, and every randomized algorithm for deciding $P$ in the \cgst{} model takes $\Omega(CC^R (f)/(\size{E_{\cut}} \log n))$ rounds.
\end{theorem}	

\section{Near Quadratic Exact Lower Bounds}
Here we show near-quadratic lower bounds for Minimum Dominating Set, max-cut, Minimum Steiner Tree, Directed and Undirected Hamiltonian Path or Cycle, and Minimum 2-edge-connected spanning subgraph.


\subsection{Minimum Dominating Set}
\label{subsec:mds}
In the minimum dominating set (MDS) problem we are given a graph $G$,
and our goal is to find a minimum cardinality set of vertices $D$ such that each vertex is \emph{dominated} by a vertex in $D$: it is either in $D$ (thus dominates itself), or has a neighbor in $D$.
The MDS problem is a central problem, with many efficient $O(\log{\Delta})$-approximation algorithms in the \cgst{} model~\cite{jia2002efficient, kuhn2005constant, DBLP:journals/jacm/KuhnMW16}, where $\Delta$ is the maximum degree in the graph.
In this section, we show that solving the problem \emph{exactly} requires nearly quadratic number of rounds, proving the following.

\begin{theorem}
	\label{thm: mds lb}
	Any distributed algorithm in the \cgst{} model for computing a minimum dominating set or for deciding whether there is a dominating set of a given size $M$ requires $\Omega (\frac{n^2}{\log ^2 n} )$ rounds.
\end{theorem}

Note that our bound immediately applies to the vertex-weighted version of the problem.
Also, note that a super-linear lower bound for \emph{deciding} whether there is a dominating set of a given size $M$
also implies the same lower bound for \emph{computing} a minimum dominating set,
since computing the size of a given set of vertices takes $O(D)$ rounds in the \cgst{} model.
Thus, it suffices to prove the second part of the theorem,
which we do by presenting a family of lower bound graphs.

Our construction is inspired by the lower bound graph construction for vertex cover from \cite{DBLP:conf/wdag/Censor-HillelKP17}.
A first attempt to obtain this lower bound could be by using the standard
NP-hardness reduction from vertex cover to MDS~\cite{papadimitriou1991optimization}
(see also Section \ref{sec:implications} for more details).
However, this would require adding a vertex on each edge in the original graph, blowing up the size of the graph, and consequently showing only a near-linear lower bound.
Instead, we show how to extend the construction from~\cite{DBLP:conf/wdag/Censor-HillelKP17} to obtain a family of lower bound graphs for the MDS problem.
Since MDS is a very basic problem, showing a lower bound for it allows us to later show lower bounds for additional problems such as Steiner Tree and Hamiltonian cycle. While there are standard reductions to both problems also from MVC, using them together with the $\tilde{\Omega}(n^2)$ lower bound from \cite{DBLP:conf/wdag/Censor-HillelKP17} can only show a near-linear lower bound. Roughly speaking, this follows since for each edge in the lower bound graph for MVC we must add at least one vertex, which blows up the number of vertices with respect to the inputs. In MDS, we cover \emph{vertices} and not \emph{edges} which allows showing an $\tilde{\Omega}(n^2)$ lower bound.
We next describe our graph construction for MDS.

\textbf{The family of lower bound graphs:} Let $k$ be a power of $2$,
and build a family of graphs $G_{x,y}$ with respect to $f=\disj_{k^2}$ and the following predicate $P$:
the graph $G_{x,y}$ contains a dominating set of size $4\log k +2$.

\textbf{The fixed graph construction:}
Start with a fixed graph $G$ (see Figure~\ref{fig: mds})
consisting of four sets of $k$ vertices each, denoted
$A_1=\{a_1^i\mid 0\leq i\leq k-1\}$,
$A_2=\{a_2^i\mid 0\leq i\leq k-1\}$,
$B_1=\{b_1^i\mid 0\leq i\leq k-1\}$,
$B_2=\{b_2^i\mid 0\leq i\leq k-1\}$;
we refer to each set as a \emph{row}, and to their vertices as \emph{row vertices}.
For each set $S\in \{A_1,A_2,B_1,B_2\},\ell \in \set{1,2}$
add three additional sets of vertices
$T_S=\{t_S^h\mid 0\leq h\leq \log k-1\},
F_S=\{f_S^h\mid 0\leq h\leq \log k-1\}$ and
$U_S=\{u_S^h\mid 0\leq h\leq \log k-1\}$;
we refer to these vertices as \emph{bit-gadget vertices}.
For each $0\leq h\leq \log k-1$ and each $\ell\in \{1,2\}$,
connect the $6$-cycle
$(f^h_{A_\ell},t^h_{A_\ell},u^h_{A_\ell},f^h_{B_\ell},t^h_{B_\ell},u^h_{B_\ell})$.
The cycles are connected to the sets $A_1,A_2,B_1,B_2$
by binary representation in the following way.
Given a row vertex $s_\ell^i\in S$,
i.e., $s\in \{a,b\},\ell \in \{1,2\},i\in [n]$,
let $i_h$ denote the $h$-th bit in the binary representation of $i$.
Connect $s_\ell^i$ to the set $\bin(s_\ell^i)\subset (F_S\cup T_S)$, defined as
$\bin(s_\ell^i)=\{f_S^h \mid i_h=0\} \cup \{t_S^h \mid i_h =1\}$.
Similarly, define $\overline\bin(s_\ell^i)=\{f_S^h\mid i_h=1\}\cup \{t_S^j\mid i_h=0\}$.
\begin{figure}[t]
	\begin{center}
		\includegraphics[scale=1,
		trim=2.5cm 16cm 3cm 1.5cm,clip]{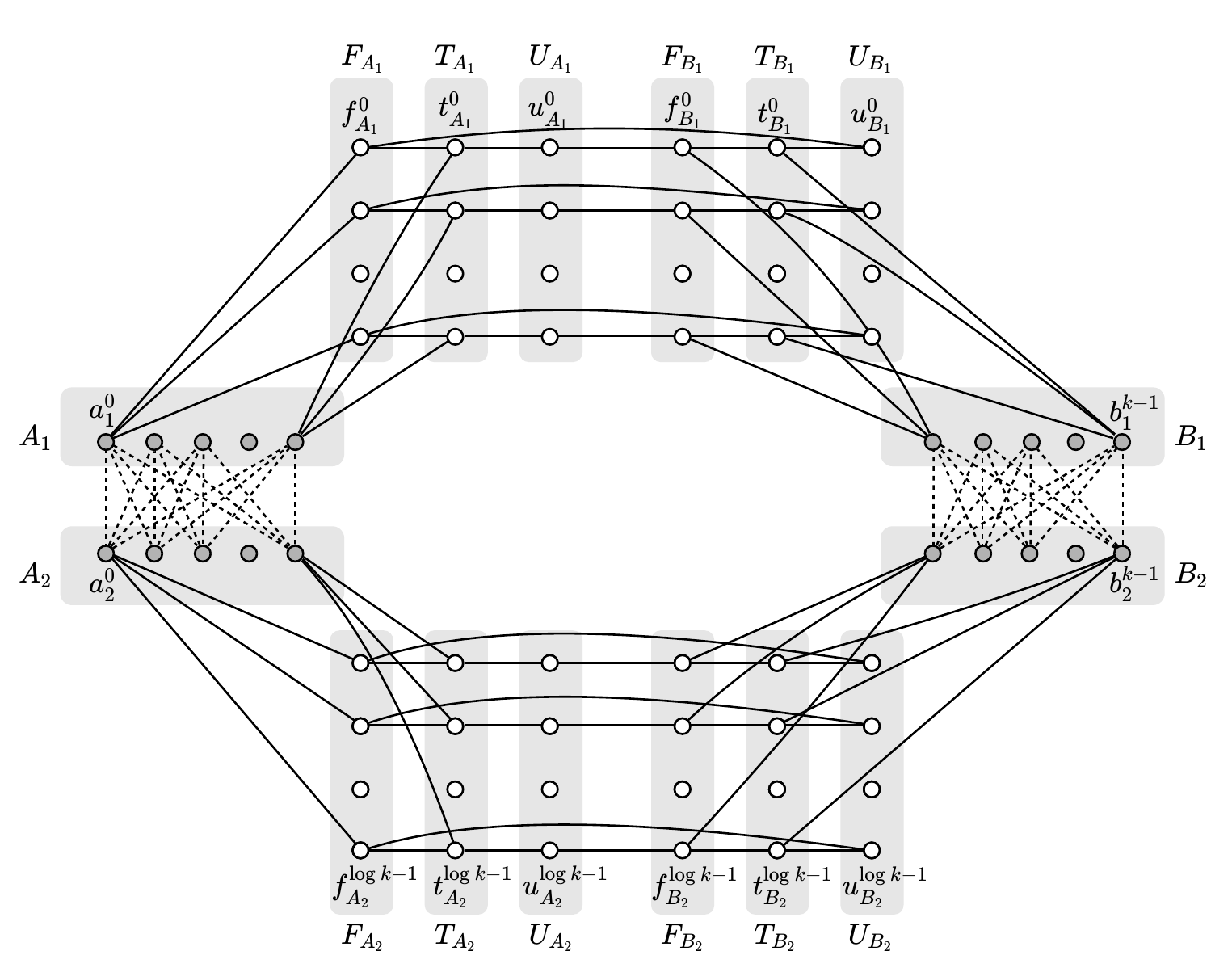}
		\caption{
			The family of lower bound graphs for MDS}
		\label{fig: mds}
	\end{center}
\end{figure}


\paragraph{Constructing $G_{x,y}$ from $G$ given $x, y \in \{0, 1\}^{k^2}$:} We index the strings $x, y \in \{0, 1\}^{k^2}$ by pairs of the form $(i, j)$ such that $0 \leq i, j \leq k-1$. Now we augment $G$ in the following way: For all pairs $(i, j)$, we add the edge $(a^i_1, a^j_2)$ if and only if $x_{i,j} = 1$, and we add the edge $(b^i_1, b^j_2)$ if and only if $y_{i,j} = 1$.

\begin{lemma} \label{lemma: mds disj}
	Given $x,y\in \set{0,1}^{k^2}$, the graph $G_{x,y}$ has a dominating set of size $4\log k+2$ if and only if $\disj_{k^2} (x,y)=\false$.
\end{lemma}


\begin{proof}
	Assume first that $\disj_{k^2} (x,y)=\false$,
	and let $(i,j)$ be the index s.t. $x_{(i,j)}=y_{(i,j)}=1$, i.e., the edges $\{a_1^i,a_2^j\},\{b_1^i,b_2^j\}$ are both in the graph $G_{x,y}$.
	Construct a dominating set $D$ containing
	$a_1^i,b^i_1$, and all the vertices in $\overline\bin(a_1^i),\overline\bin(a_2^j),
	\overline\bin(b_1^i),\overline\bin(b_2^j)$.
	This set is of size $4\log k +2$.
	By taking all of $\overline\bin(a_1^i)$ and $\overline\bin(b_1^i)$,
	we made sure all the bit-gadget vertices of
	$F_{A_1},T_{A_1},U_{A_1},F_{B_1},T_{B_1},U_{B_1}$
	are dominated,
	and similarly, all the other bit-gadget vertices are dominated as well.
	
	For the row vertices,
	first observe that $a_1^i,b_1^i,a_2^j,b_2^j$ are dominated by $D$
	using the edges $\{a_1^i,a_2^j\}$ and $\{b_1^i,b_2^j\}$.
	For the other row vertices,
	consider w.l.o.g.\ some vertex $a_1^{i'}$ where $i'\neq i$,
	so there exists $0\leq h\leq \log k-1$ s.t. $i'_h\neq i_h$.
	If $f_{A_1}^h \in \bin(a_1^i)$,
	then $t_{A_1}^h \in \overline\bin(a_1^i)\subset D$,
	and $a_1^{i'}$ is dominated;
	similarly, if $t_{A_1}^h \in \bin(a_1^i)$ then $a_1^{i'}$ is dominated as well.
	Thus, $D$ is a dominating set of size $4\log k+2$.
	
	For the other direction, assume $G_{x,y}$ has a dominating set $D$ of size $4\log k+2$.
	For each $0\leq h\leq\log k-1$ and $\ell\in\set{1,2}$,
	the only way to dominate two vertices $u_{A_\ell}^h$ and $u_{B_\ell}^h$
	is by two different bit-gadget vertices,
	so $D$ must contain at least $4\log k$ bit-gadget vertices.
	We show that this is the exact number of bit-gadget vertices in $D$
	by eliminating all other options.
	
	If $D$ has $4\log k+2$ bit-gadget vertices
	and no row vertices,
	then there is a set $S\in\set{A_1,A_2,B_1,B_2}$ such that
	$F_S\cup T_S\cup U_S$ contains at most $\log k$ vertices of $D$.
	In this set, $D$ contains at most one vertex out of any consecutive triple
	$\set{f_S^h,t_S^h,u_S^h}$:
	If it contains more than one vertex in such a triplet,
	it contains no vertices of some other triple $\{f_S^{h'},t_S^{h'},u_S^{h'}\}$,
	and as $t_S^{h'}$ cannot be dominated by row vertices,
	it is not dominated at all.
	So, for some $s^i_\ell\in S$,
	there is a set of the form $\bin(s^i_\ell)$ in $F_S\cup T_S\cup U_S$,
	which does not intersect $D$.
	As $D$ does not contain any row vertices,
	it does not contain $s^i_\ell$ or any of its neighbors,
	so this vertices is not dominated.
	
	If $D$ has $4\log k+1$ bit-gadget vertices,
	and exactly one row-vertex,
	assume w.l.o.g. that this row vertex is in $A_1$.
	For each other set $S\in\set{A_2,B_1,B_2}$,
	$D$ must contain at least one vertex from each triplet
	$\set{f_S^h,t_S^h,u_S^h}$ in order to dominate $t_S^h$,
	and thus $F_S\cup T_S\cup U_S$ contains at least $\log k$ vertices of $D$.
	Hence, for at least one of the set $S\in\set{B_1,B_2}$,
	$F_S\cup T_S\cup U_S$ contains exactly $\log k$ vertices of $D$.
	The argument now follows the same lines as the previous:
	There is a set of the form $\bin(b^i_\ell)$ that does not intersect $D$,
	and as $D$ does not contain vertices from $B_1$ and $B_2$,
	$b^i_\ell$ is not dominated.
	
	We are left with the case where
	$D$ has exactly $4\log k$ bit-gadget vertices and $2$ row vertices.
	This requires a sequence of claims, as follows.
	
	\emph{The two row vertices are in different sets.}
	If the two row vertices are in the same set, w.l.o.g.~$A_1$,
	then $D$ must contain at least one vertex from each triplet $f^h_{A_2},t^h_{A_2},u^h_{A_2}$,
	in order to dominate $t^h_{A_2}$.
	Since $D$ contains exactly $2$ vetrices of every $6$-cycle,
	it can contain most one vertex of each triplet $f^h_{B_2},t^h_{B_2},u^h_{B_2}$.
	Hence, there is a vertex $b_2^i\in B_2$ such that $\bin(b_2^i)$
	does not intersect $D$.
	Since no vertex of $B_1$ or $B_2$ is in $D$,
	the vertex $b_2^i$ is not dominated, a contradiction.
	
	\emph{$D$ contains exactly one of every pair of neighbors $f_S^h,t_S^h$.}
	Assume w.l.o.g. that $f_{A_1}^h,t_{A_1}^h$ are both in $D$,
	so in the same $6$-cycle,
	$f_{B_1}^h,t_{B_1}^h$ must both be dominated by vertices of $B_1$.
	But no single vertex in $b_1^i\in B_1$ has both $f_{B_1}^h,t_{B_1}^h$
	in $\bin(b_1^i)$,
	and $D$ cannot contain two vertices of $B_1$.
	
	\emph{For each pair $\{f^h_{A_\ell},f^h_{B_\ell}\}$,
		and for each pair $\{t^h_{A_\ell},t^h_{B_\ell}\}$,
		either both vertices are in $D$ or both are not in $D$.}
	If $f^h_{A_\ell}$ is in $D$ then by the previous claim $t^h_{A_\ell}$ is not,
	so $f^h_{B_\ell}$ must be in $D$ in order to dominate $u^h_{A_\ell}$;
	similarly, if $f^h_{B_\ell}$ is in $D$ then $f^h_{A_\ell}$ must be in $D$ as well.
	If $t^h_{A_\ell}$ is in $D$ then $f^h_{A_\ell}$ is not,
	then $t^h_{B_\ell}$ must be in $D$ in order to dominate $u^h_{B_\ell}$;
	similarly,
	if $t^h_{B_\ell}$ is in $D$ then $t^h_{A_\ell}$ must be in $D$.
	
	As exactly $2$ of every $6$-cycle vertices are in $D$,
	from each pair $\{f_S^h,t_S^h\}$ at most one vertex is in $D$.
	Choose a vertex not in $D$ from each pair $\{f_{A_1}^h,t_{A_1}^h\}$,
	and let $a_1^i$ be the vertex such that $\bin(a_1^i)$ is the set of these vertices.
	Hence, $\bin(a_1^i)$ does not intersect $D$,
	and $a_1^i$ is not dominated by bit-gadget vertices.
	For each vertex $f^h_{A_1}$ that is not in $D$,
	the corresponding vertex $f^h_{B_1}$ is also not in $D$,
	and similarly for vertices of the form $t^h_S$.
	Thus, $\bin(b_1^i)$ also does not intersect $D$,
	and $b_1^i$ is not dominated by bit-gadget vertices.
	Similarly, let $a_2^j$ be a vertex such that $\bin(a_2^j)$ does not intersect $D$,
	and conclude $\bin(b_2^j)$ also does not intersect $D$.
	Thus, we have four vertices, $a_1^i,b_1^i,a_2^j,b_2^j$
	that are not dominated by bit-gadget vertices,
	and instead, must be dominated by row vertices.
	It is impossible that both $A_1$ and $A_2$ have vertices from $D$,
	as then there are no vertices from $D$ in $B_1$ and $B_2$
	and $b_1^i$ and $b_2^j$ are not dominated;
	similarly, it is impossible that both $B_1$ and $B_2$ have vertices from $D$.
	So, one of $a_1^i$ and $a_2^j$ must be in $D$, and dominate the other,
	and either way the edge $\{a_1^i,a_2^j\}$ must exist in $G_{x,y}$;
	similarly, the edge $\{b_1^i,b_2^j\}$ must exist in $G_{x,y}$.
	We thus conclude that $x_{(i,j)}=y_{(i,j)}=1$, as claimed.
\end{proof}

Now all that is left in order to prove Theorem \ref{thm: mds lb} is to use
Theorem~\ref{generallowerboundtheorem}.

\begin{proof}
	We use Theorem~\ref{generallowerboundtheorem}.
	For that purpose we divide the graph $G_{x,y}$ into $V_A$ and $V_B$ in the following way. Define $V_A=A_1 \cup A_2 \cup F_{A_1} \cup T_{A_1} \cup U_{A_1} \cup F_{A_2} \cup T_{A_2}\cup U_{A_2}$, and define $V_B=V\setminus V_A$.
	Since $|E_{\cut}|=O(\log k)$ we get by using
	Theorem~\ref{generallowerboundtheorem} a lower bound of $\Omega \left( \frac{k^2}{\log k\log n} \right)$ rounds for the problem of deciding  whether there is a dominating set of size $4\log k+2$.
	As $n=\Theta(k)$, the proof is completed.
\end{proof}

\subsection{Hamiltonian Path and Cycle Lower Bounds and applications}
\label{subsec:Hamiltonian}

In this Section, we show near-quadratic lower bounds for Hamiltonian path and cycle in directed and undirected graphs, and for the minimum 2-edge-connected spanning subgraph problem (2-ECSS).

While the Hamiltonian path problem is not an optimization problem, it is known to be NP-hard,
e.g., through a reduction from minimum vertex cover~\cite{DBLP:books/daglib/p/Karp10}. A lower bound of $\Tilde{\Omega} (\sqrt{n}+D)$ is known for the \emph{verification} version of the problem \cite{Dassarmaetal12}, and Hamiltonian paths can be found efficiently in random graphs \cite{DBLP:conf/wdag/GhaffariL18, ChatterjeeFPP18, Turau18}.
However, in general graphs we show an $\Tilde{\Omega}(n^2)$ lower bound.

\remove{
\begin{theorem}
\label{thrm:Hamiltonian}
Any distributed algorithm in the \cgst{} model for finding a directed Hamiltonian path or deciding whether there is a directed Hamiltonian path in the input graph requires $\Omega ({n^2}/{\log ^4 n} )$ rounds.
\end{theorem}

We extend Theorem~\ref{thrm:Hamiltonian} to the case of a cycle and to the undirected variants, as follows.
\begin{theorem}
\label{thrm:Hamiltoniancycle}
Any distributed algorithm in the \cgst{} model for finding a directed Hamiltonian cycle or for deciding whether these such in the input graph requires $\Omega ({n^2}/{\log ^4 n} )$ rounds.
\end{theorem}

\begin{theorem}
\label{thrm:Hamiltonianundirected}
Any distributed algorithm in the \cgst{} model for finding an undirected Hamiltonian path or cycle or deciding whether there is such in the input graph requires $\Omega ({n^2}/{\log ^4 n} )$ rounds.
\end{theorem}
}


We also show that this directly shows hardness of \emph{unweighted} 2-ECSS. While approximation algorithms to the problem include an $O(n)$-round $3/2$-approximation \cite{krumke2007distributed}, an $O(D)$-round 2-approximation \cite{DBLP:conf/opodis/Censor-HillelD17} and an $O(\log^{1+o(1)}{n})$-round $O(1)$-approximation \cite{un_kECSS}, we show that solving the 2-ECSS problem \emph{exactly} requires near-quadratic number of rounds.

\remove{  
\begin{theorem}
\label{thrm:twoedgeconnected}
Any distributed algorithm in the \cgst{} model for finding a minimum size 2-edge-connected spanning subgraph or deciding whether there is a 2-edge-connected sub-graph with a given number of edges requires $\Omega ({n^2}/{\log ^4 n} )$ rounds.
\end{theorem}
}

\subsubsection{Directed Hamiltonian Path}

\paragraph{Intuition for the construction:}
Imagine traversing our lower bound graph for MDS in search for a Hamiltonian path, to which we add $start$ and $end$ vertices, where the path begins and ends, respectively. We need the existence of such a path to determine if the input strings $x,y$ are disjoint, so our approach is to traverse the bit-gadget vertices and the row vertices such that after a certain prefix, it remains to use a single edge from $A_1$ to $A_2$ followed by a single edge from $B_1$ to $B_2$ in order to each $end$. The crux is to guarantee that these two edges have corresponding indexes. Since row vertices in different sides are not connected, we add a special vertex that is reachable from all $A_2$ vertices, leading to another special vertex that is connected to all $B_1$ vertices, which adds a single edge to $E_{\cut}$.

The key challenge is guaranteeing that we can use such two edges iff they have the same indexes, implying that we need the prefix of the path to exclude exactly such 4 vertices.
The second challenge is that there are more row vertices than bit-gadget vertices, so we cannot simply walk back and forth between these two types of vertices without reaching the bit-gadget vertices multiple times.

Our high-level approach for addressing these two issues is as follows (see Figure~\ref{fig:hamiltonian}).
For each row vertex in $S\in\{A_1,A_2,B_1,B_2\}$ and each of its corresponding $2\log k$ bit-gadget vertices, we plug in a \emph{gadget} of constant size that replaces the bit-gadget vertices altogether. Further, we traverse the gadgets in an order that corresponds to some choice of $\true$ and $\false$ for each index of row vertices. The latter is what promises that we reach $end$ iff we use a single edge from $A_1$ to $A_2$ and a single edge from $B_1$ to $B_2$ with corresponding indexes.
However, since the number of gadget-vertices is proportional to $2\log k$ times the number of row vertices $k$, we now face an opposite problem: there are more gadget vertices than row vertices, and the latter play a role in multiple gadgets. Moreover, if we chose to traverse the gadgets that correspond to, say, a $\true$ choice for some index, we still need to visit the gadget vertices of the $\false$ choice and we need to do so without visiting its respective row vertex. The power of the gadget is that it is designed to allow us to choose whether the path \emph{visits} the row vertex of that gadget, or \emph{skips} it and continues to the next gadget.
This has two main strengths: First, it nullifies the issue that is caused by having row vertices appear in multiple gadgets. Second, it allows us to traverse all the gadget vertices that do not correspond to the $\true/\false$ choices that we make.
This eventually allows us to obtain our desired claim, that a directed Hamiltonian path exists if and only if $\disj_{K}(x,y)=\false$ for an appropriate $K$ that allows a near-quadratic lower bound. \\[-7pt]

\paragraph{The family of lower bound graphs:} Given $k$ (we assume for simplicity that $k$ is a power of 2) we build our family of graphs $G_{x,y}$ with respect to $f = \disj_{k^2}$ and $P$ being the predicate that says that there exists in $G_{x,y}$ a directed Hamiltonian Path.

\paragraph{The fixed graph construction:}
See Figure~\ref{fig:hamiltonian} for an illustration.
\begin{figure}[t]

	\begin{center}
		\includegraphics[scale=1.1,
		trim=2.5cm 16cm 3.5cm 3cm,clip]{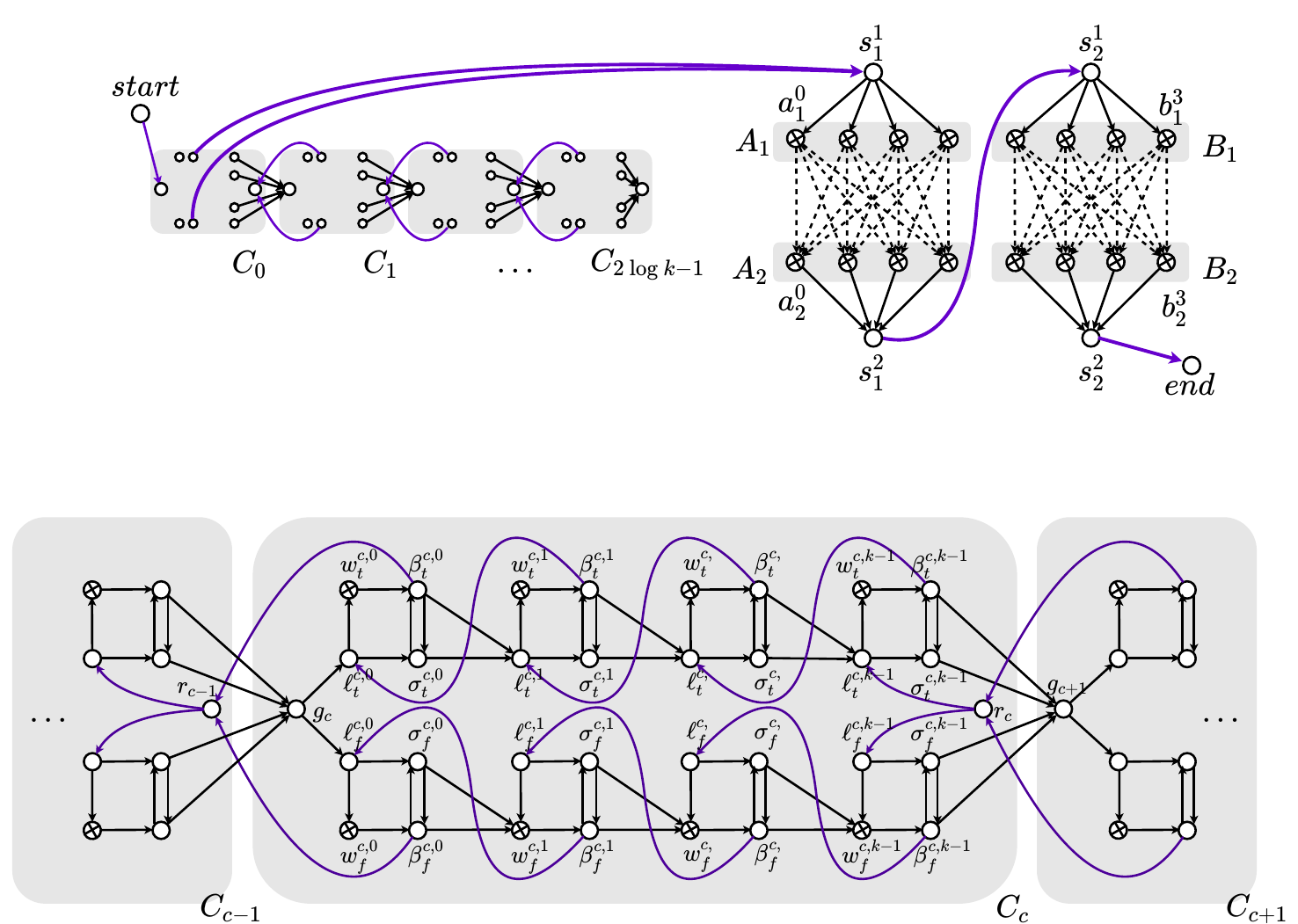}
		\caption{
			Hamiltonian path graph construction for $k=4$.}
		\label{fig:hamiltonian}
	\end{center}

\end{figure}
We define 2 vertices, $start$ and $end$, and 4 vertices $s^1_1, s^2_1, s^1_2, s^2_2$. Further, there are $4k$ vertices $a^i_1, a^i_2, b^i_1, b^i_2$ for every $0 \leq i \leq k-1$ (row vertices).

The vertex $s^1_1$ has outgoing edges to all vertices $a^i_1$ for $0 \leq i \leq k-1$.
The vertex $s^2_1$ has incoming edges from all vertices $a^i_2$ for $0 \leq i \leq k-1$, and has an outgoing edge to $s^1_2$.
The vertex $s^1_2$ has outgoing edges to all vertices $b^i_1$ for $0 \leq i \leq k-1$.
The vertex $s^2_2$ has incoming edges from all vertices $b^i_2$ for $0 \leq i \leq k-1$, and has an outgoing edge to $end$.

For every $0 \leq c \leq 2\log(k)-1$ we define a \emph{box} $C_c$. The box contains vertices $g_c$ and $r_c$, and for every $q \in \{t,f\}$ and every $0 \leq d\leq k-1$ it contains 4 vertices: $\ell^{c,d}_{q}$ (launch vertices), $wheel^{c,d}_{q}$ (wheel vertices), $\sigma^{c,d}_{q}$ (skip vertices), and $\beta^{c,d}_{q}$ (burn vertices).

A crucial detail here is that the wheel vertices are not additional vertices, but they are simply reoccurrences of the row vertices, as follows. For every $0 \leq i \leq k-1$, it holds that
\begin{itemize}
	\item $a^{i}_1=wheel^{c,d}_{t}$, for all $c,d$ such that $i$ is the $d$-th index whose $c$-th bit is $1$, and $a^{i}_1=wheel^{c,d}_{f}$, for all $c,d$ such that $i$ is the $d$-th index whose $c$-th bit is $0$.
	
	\item $b^{i}_1=wheel^{c,d}_{t}$, for all $c,d$ such that $i$ is the $(d-(k/2))$-th index whose $c$-th bit is $1$, and $b^{i}_1=wheel^{c,d}_{f}$, for all $c,d$ such that $i$ is the $(d-(k/2))$-th index whose $c$-th bit is $0$.
	
	\item $a^{i}_2=wheel^{c,d}_{t}$, for all $c,d$ such that $i$ is the $d$-th index whose $(c-\log(k))$-th bit is $1$, and $a^{i}_2=wheel^{c,d}_{f}$, for all $c,d$ such that $i$ is the $d$-th index whose $(c-\log(k))$-th bit is $0$.
	
	\item $b^{i}_2=wheel^{c,d}_{t}$, for all $c,d$ such that $i$ is the $(d-(k/2))$-th index whose $(c-\log(k))$-th bit is $1$, and $b^{i}_2=wheel^{c,d}_{f}$, for all $c,d$ such that $i$ is the $(d-(k/2))$-th index whose $(c-\log(k))$-th bit is $0$.
\end{itemize}

Within each box $C_c$, the vertices $g_c$ has outgoing edges to $\ell^{c,0}_{q}$ for both possible values of $q$. For each $0 \leq d \leq k-1$, the vertex $\ell^{c,d}_{q}$ has outgoing edges to $\sigma^{c,d}_{q}$ and to $wheel^{c,d}_{q}$. The vertex $wheel^{c,d}_{q}$ has an outgoing edge to $\beta^{c,d}_{q}$. The vertices $\sigma^{c,d}_{q}$ and $\beta^{c,d}_{q}$ have outgoing edges between each other in both directions. In addition, $\sigma^{c,d}_{q}$ and $\beta^{c,d}_{q}$ have outgoing edges to a vertex $v$, such that $v=\ell^{c,d+1}_{q}$ if $d \neq k-1$, $v=g_{c+1}$ if $d = k-1$ and $c\neq 2\log(k)-1$, and $v=r_{2\log(k)-1}$ if $d = k-1$ and $c = 2\log(k)-1$. Finally, the vertex $\beta^{c,d}_{q}$ has an outgoing edge to a vertex $u$, such that $u=\ell^{c,d-1}_{q}$ if $d \neq 0$, and $u=r_{c-1}$ if $d=0, c \neq 0$.

Finally, the vertex $start$ a single outgoing edge into $g_0$, and the vertex $s^1_1$ has two incoming edges, from $\beta^{0,0}_{q}$ for $q \in \{t,f\}$.

\paragraph{Constructing $G_{x,y}$ from $G$ given $x,y \in \{0,1\}^{k^2}$:}
We index the strings $x,y \in \{0,1\}^{k^2}$ by pairs of the form $(i,j)$ such that $0 \leq i,j \leq k-1$. Now we augment $G$ in the following way: For all pairs $(i, j)$, we add the edge
$(a^i_1,a^j_2)$ if and only if $x_{i,j} = 1$, and we add the edge $(b^i_1,b^j_2)$ if and only if $y_{i,j} = 1$.

\paragraph{Proving that $G_{x,y}$ is a family of lower bound graphs:}
For every $q \in \{t,f\}$, $0\leq c\leq 2\log(k)-1$, and $0\leq d\leq k-1$, we define the following. We first define three types of \emph{$(q,c,d)$-forward-steps}:
A \emph{$(q,c,d)$-wheel-forward-step} is a subpath $(\ell^{c,d}_{q},wheel^{c,d}_{q},\beta^{c,d}_{q},\sigma^{c,d}_{q}, v)$, such that if $d \neq k-1$ then $v=\ell^{c,d+1}_{q}$, if $d=k-1$ and $c \neq 2\log(k)-1$ then $v=g_{c+1}$, and otherwise, $v=r_{2\log(k)-1}$.
A \emph{$(q,c,d)$-sigma-forward-step} is a subpath $(\ell^{c,d}_{q},\sigma^{c,d}_{q}, v)$, such that if $d \neq k-1$ then $v=\ell^{c,d+1}_{q}$, if $d=k-1$ and $c \neq 2\log(k)-1$ then $v=g_{c+1}$, and otherwise, $v=r_{2\log(k)-1}$ (we will show that a Hamiltonian path cannot contain sigma-forward-steps).
A \emph{$(q,c,d)$-beta-forward-step} is a subpath $(\ell^{c,d}_{q},\sigma^{c,d}_{q},\beta^{c,d}_{q}, v)$, such that if $d \neq k-1$ then $v=\ell^{c,d+1}_{q}$, if $d=k-1$ and $c\neq 2\log(k)-1$ then $v=g_{c+1}$, and otherwise, $v=r_{2\log(k)-1}$.

Finally, we define:
A \emph{$(q,c,d)$-backward-step} is a subpath $(\beta^{c,d}_{q}, v)$, such that if $d\neq 0$ then $v=\ell^{c,d-1}_{q}$, if $d=0,c \neq 0$ then $v=r_{c-1}$, and otherwise, $v=s^1_1$.

\begin{claim}
	\label{claim:HP}
	If $\disj(x,y)=\false$ then $G_{x,y}$ has a Hamiltonian path.
\end{claim}

\begin{proof}
	If $\disj(x,y)=\false$ then there are row vertices $a_1^i, a_2^j, b_1^i, b_2^j$ such that the edges $(a_1^i, a_2^j), (b_1^i, b_2^j)$ are in $G_{x,y}$. Denote by $bin(i)$ the set of indexes $c$ in which the $c-th$ bit in the binary representation of $i$ is $1$.
	
	The Hamiltonian path is as follows. From $start$ it goes into $g_0$ which is in the first box $C_0$. From $g_0$, if $0\in bin(i)$ then we go into the launch vertex $\ell _{f}^{0,0}$ and otherwise we go into the launch vertex $\ell _t ^{0,0}$. In general, when the path enters $g_s,\ 0\leq s\leq \log(k)-1$, if $s\in bin(i)$ then it goes to the launch vertex $\ell _f^{s,0}$ and otherwise to the launch vertex $\ell _t^{s,0}$). When the path enters $g_s,\ \log(k)\leq s\leq \log(k)-1$, if $s\in bin(j)$ then it goes to the launch vertex $\ell _f^{s,0}$ and otherwise to the launch vertex $\ell _t^{s,0}$). We refer to these choices of edges as \emph{$choose _s$}.
	
	For every $0 \leq s \leq 2\log(k)-1$, from the launch vertex $\ell^{s,0}_{q}$ that was reached by $choose_s$, the path proceeds by $(q,s,d)$-forward-steps until it reaches $g_{s+1}$ if $s\neq 2\log(k)-1$, or $r_{\log(k)-1}$ otherwise. These forward steps are either $(q,s,d)$-wheel-forward-steps or $(q,s,d)$-beta-forward-steps, such that for every $0 \leq d \leq k-1$, the step is a $(q,s,d)$-wheel-forward-step if and only if $wheel^{s,d}_{q}$ is not yet visited.
	
	At the end of the above, the path reaches $r_{\log(k)-1}$. For every $0 \leq s \leq 2\log(k)-1$, from $r_s$ the path goes into $\ell^{s,k-1}_{q'}$ such that $\ell^{s,0}_{q}$ is the edge $choose_s$ and $q'\neq q$. For every $0 \leq d \leq k-1$, the path then takes $(q',s,d)$-backward-steps until it reaches $r_{s-1}$ if $s \neq 0$, or $s^1_1$ otherwise.
	
	The above traversal has the property that it visits all row vertices through one of their wheel copies, except for exactly 4 row vertices: $a_1^i, a_2^j, b_1^i$, and $b_2^j$. This is because of our choice of $choose_s$ according to the binary representations of $i$ and $j$.
	
	Now, from $s^1_1$ the path continues to $a_1^i$, from which it continues to $a_2^j$ through the edge that exists given the input $x$. From there the path continues to $s_1^2$ and to $s_2^1$, from which it goes into $b_1^i$. From there it continues to $b_2^j$ through the edge that exists given the input $y$. Finally, the path goes into $s_2^2$ and ends in $end$.
\end{proof}

We now claim that the opposite also holds.

\begin{claim}
	\label{claim:noHP}
	If $\disj(x,y)=\true$ then $G_{x,y}$ does not have a Hamiltonian path.
\end{claim}

To prove Clai~\ref{claim:noHP}, we show that the following hold for any Hamiltonian path $\cal{P}$ in $G_{x,y}$.

\begin{claim}\label{claim:Pwheel}
	For any Hamiltonian path $\cal{P}$ in $G_{x,y}$, for every $q \in \{t,f\}$, $0\leq c\leq 2\log(k)-1$, and $0\leq d\leq k-1$, if $\cal{P}$ goes from a launch vertex $\ell^{c,d}_{q}$ to $wheel^{c,d}_{q}$, then it contains the $(q,c,d)$-wheel-forward-step.
\end{claim}

\begin{proof}
	Fix $q \in \{t,f\}$, $0\leq c\leq 2\log(k)-1$, and $0\leq d \leq k-1$. Assume that $\cal{P}$ goes into $\ell^{c,d}_{q}$ and continues to $wheel^{c,d}_q$. Assume towards a contradiction that the next vertex that $\cal{P}$ goes into is not $\beta^{c,d}_q$. Then $\cal{P}$ is not Hamiltonian because it never reaches either of $\sigma^{c,d}_q, \beta^{c,d}_q$. This is because the only incoming edges to these two vertices are from each other, or from $\ell^{c,d}_{q}$ and $wheel^{c,d}_q$, which have already been visited by $\cal{P}$. This means that from $wheel^{c,d}_q$, the path goes to $\beta^{c,d}_q$. The same argument gives that $\cal{P}$ must then continue to $\sigma^{c,d}_{q}$, as otherwise it never reaches it, because it already visited all the vertices at its incoming edges. From $\sigma^{c,d}_{q}$ the path must then continue to $\ell^{c,d+1}_{q}$ if $d \neq k-1$, or to $g_{c+1}$ otherwise, unless $c=2\log(k)-1$ in which case the path must continue to $r_{2\log(k)-1}$.
\end{proof}

\begin{claim}\label{claim:P}
	For any Hamiltonian path $\cal{P}$ in $G_{x,y}$, for every $q \in \{t,f\}$, $0\leq c\leq 2\log(k)-1$, the following hold:
	\begin{enumerate}
		\item \label{case:forward} If $\cal{P}$ goes from $g_c$ to $\ell^{c,0}_{q}$, then it must continue by a sequence of $(q,c,d)$-forward-steps until it reaches $g_{c+1}$ if $c\neq 2\log(k)-1$, or $r_{2\log(k)-1}$, otherwise.
		\item \label{case:backward} If $\cal{P}$ goes from $r_c$ to $\ell^{c,k-1}_{q}$, then it must continue by a sequence of $(q,c,d)$-backward-steps until it reaches $r_{c-1}$ if $c \neq 0$, or $s^1_1$, otherwise.
	\end{enumerate}
\end{claim}

\begin{proof}
	The proof is by induction on $c$. Let $c=0$, and assume that $\cal{P}$ goes from $g_0$ to $\ell^{0,0}_{q}$. If $\cal{P}$ does not continue with a $(q,0,0)$-forward-step, it means that it goes to $\beta^{0,0}_{q}$, either through $wheel^{0,0}_{q}$ or through $\sigma^{0,0}_{q}$, and from there to $s^1_1$. Since other vertices, e.g., $\ell^{0,1}_{q}$ must be visited, at some point after visiting $s^1_1$, the path must move into some $wheel^{c',d'}_{q'}$ from its row copy. But combining the fact that $s^1_1$ is already visited with Claim~\ref{claim:Pwheel}, gives that $\cal{P}$ can never go back into a row copy, and hence can never reach the end node. This means that $\cal{P}$ continues from $\ell^{0,0}_{q}$ with a $(q,0,0)$-forward-step. Now, we proceed with an induction over $d$, and let $0 < d\leq k-1$ be the minimal such that $\cal{P}$ does not take a $(q,0,d)$-forward-step. This means that from $\ell^{0,d}_{q}$, the path $\cal{P}$ goes to $\beta^{0,d}_{q}$, either through $wheel^{0,d}_{q}$ or through $\sigma^{0,d}_{q}$, and then continues to $\ell^{0,d-1}_{q}$. But by the minimality of $d$, $\ell^{0,d-1}_{q}$ is already visited, which is a contradiction. This gives Part~\ref{case:forward}.
	
	Next, assume that $\cal{P}$ goes from $r_0$ to $\ell^{0,k-1}_{q}$. By Part~\ref{case:forward} for $c=0$, it must be the case that from $g_0$ the path continued by a sequence of $(q',0,d)$-forward-steps until it reached $g_1$, for $q' \neq q$, as otherwise, $\ell^{0,k-1}_{q}$ would have already been visited. But then from $\ell^{0,k-1}_{q}$, the path $\cal{P}$ cannot continue with a $(q,0,k-1)$-forward-step because $g_1$ is already visited, and hence it continues with a $(q,0,k-1)$-backward-step. Now, we proceed with a backward induction over $d$, and let $0 \leq d < k-1$ be the maximal such that $\cal{P}$ does not take a $(q,0,d)$-backward-step. This means that from $\beta^{0,d}_{q}$, the path $\cal{P}$ goes to $\ell^{0,d+1}_{q}$, either directly or through $\sigma^{0,d}_{q}$. But by the maximality of $d$, $\ell^{0,d+1}_{q}$ is already visited, which is a contradiction. This gives Part~\ref{case:backward}.
	
	This proves the claim for $c=0$, and we now continue by assuming towards a contradiction that $0 <  c \leq 2\log(k)-1$ is the minimal for which the claim does not hold. First, assume that $\cal{P}$ goes from $g_c$ to $\ell^{c,0}_{q}$. If $\cal{P}$ does not continue with a $(q,c,0)$-forward-step, it means that it goes to $\beta^{c,0}_{q}$, either through $wheel^{c,0}_{q}$ or through $\sigma^{c,0}_{q}$, and from there to $r_{c-1}$. By the minimality of $c$, the path $\cal{P}$ then continues by backward-steps all the way to $s^1_1$. Since other vertices, e.g., $\ell^{c,1}_{q}$ must be visited, at some point after visiting $s^1_1$, the path must move into some $wheel^{c',d'}_{q'}$ from its row copy. But combining the fact that $s^1_1$ is already visited with Claim~\ref{claim:Pwheel}, gives that $\cal{P}$ can never go back into a row copy, and hence can never reach the end node. This means that $\cal{P}$ continues from $\ell^{c,0}_{q}$ with a $(q,c,0)$-forward-step. Now, we proceed with an induction over $d$, and let $0 < d\leq k-1$ be the minimal such that $\cal{P}$ does not take a $(q,c,d)$-forward-step. This means that from $\ell^{c,d}_{q}$, the path $\cal{P}$ goes to $\beta^{c,d}_{q}$, either through $wheel^{c,d}_{q}$ or through $\sigma^{c,d}_{q}$, and then continues to $\ell^{c,d-1}_{q}$. But by the minimality of $d$, $\ell^{i,d-1}_{q}$ is already visited, which is a contradiction. This gives Part~\ref{case:forward}.
	
	Next, assume that $\cal{P}$ goes from $r_c$ to $\ell^{c,k-1}_{q}$. By Part~\ref{case:forward} for $c$, it must be the case that from $g_c$ the path continued by a sequence of $(q',c,d)$-forward-steps until it reached $g_{c+1}$, for $q' \neq q$, as otherwise, $\ell^{c,k-1}_{q}$ would have already been visited. But then from $\ell^{c,k-1}_{q}$, the path $\cal{P}$ cannot continue with a $(q,c,k-1)$-forward-step because $g_{c+1}$ is already visited, and hence it continues with a $(q,c,k-1)$-backward-step. Now, we proceed with a backward induction over $d$, and let $0 \leq d < k-1$ be the maximal such that $\cal{P}$ does not take a $(q,c,d)$-backward-step. This means that from $\beta^{c,d}_{q}$, the path $\cal{P}$ goes to $\ell^{c,d+1}_{q}$, either directly or through $\sigma^{c,d}_{q}$. But by the maximality of $d$, $\ell^{c,d+1}_{q}$ is already visited, which is a contradiction. This gives Part~\ref{case:backward}, and completes the proof.
\end{proof}

Finally, we show that sigma-forward-steps cannot occur.

\begin{claim}\label{claim:Psigma}
	For any Hamiltonian path $\cal{P}$ in $G_{x,y}$, for every $q \in \{t,f\}$, $0\leq c \leq 2\log(k)-1$, and $0\leq d \leq k-1$, the path $\cal{P}$ does not contain a $(q,c,d)$-sigma-forward-step.
\end{claim}

\begin{proof}
	The proof is by induction on $c$. For the base case $c=0$, assume that $\cal{P}$ contains a $(q,0,0)$-sigma-forward-step. Then in order to visit $\beta^{0,0}_{q}$, the path must enter it though $wheel^{0,0}_{q}$. By Claim~\ref{claim:Pwheel}, this must happen after $s^1_1$ is visited, and therefore from $\beta^{0,0}_{q}$ the path must continue to $\sigma^{0,0}_{q}$ or to $\ell^{0,1}_{q}$. But they are both visited from $\ell^{0,0}_{q}$ by the $(q,0,0)$-sigma-forward-step, a contradiction. We now proceed by induction on $0 < d \leq k-1$. Assume that $d$ is the minimal such that $\cal{P}$ contains a $(q,0,d)$-sigma-forward-step. Then in order to visit $\beta^{0,d}_{q}$, the path must enter it though $wheel^{0,d}_{q}$. By Claim~\ref{claim:Pwheel}, this must happen after $s^1_1$ is visited, and therefore from $\beta^{0,d}_{q}$ the path must continue to $\sigma^{0,d}_{q}$ or to the $v$ that is reached by the $(q,0,d)$-sigma-forward-step, which are both impossible.
	
	We now proceed with the induction step, and let $0 < c < 2\log(k)-1$ be the smallest such that $\cal{P}$ contains a $(q,c,d)$-sigma-forward-step. If $d=0$, then in order to visit $\beta^{c,0}_{q}$, the path must enter it though $wheel^{c,0}_{q}$. By Claim~\ref{claim:Pwheel}, this must happen after $s^1_1$ is visited, and therefore by Claim~\ref{claim:P} we have that $r_{c-1}$ is visited by $\beta^{c,0}_{q'}$ such that $q' \neq q$. Therefore from $\beta^{c,0}_{q}$, the path must continue to $\sigma^{c,0}_{q}$ or to the $v$ that is reached by the $(q,c,0)$-sigma-forward-step, which are both impossible. We now proceed by induction on $0 < d \leq k-1$. Assume that $d$ is the minimal such that $\cal{P}$ contains a $(q,c,d)$-sigma-forward-step. Then in order to visit $\beta^{c,d}_{q}$, the path must enter it though $wheel^{c,d}_{q}$. By Claim~\ref{claim:Pwheel}, this must happen after $s^1_1$ is visited, and therefore from $\beta^{c,d}_{q}$ the path must continue to $\sigma^{c,d}_{q}$ or to the $v$ that is reached by the $(q,c,d)$-sigma-forward-step, which are both impossible.
\end{proof}

\begin{proofof}{Claim~\ref{claim:noHP}}
	We assume that $\cal{P}$ is a Hamiltonian path in $G_{x,y}$ and show that $\disj(x,y)=\false$.
	
	By Claims~\ref{claim:Pwheel},~\ref{claim:P}, and~\ref{claim:Psigma} it must be that $\cal{P}$ goes from $start$ to $g_0$, from there it goes to $\ell^{0,0}_{q}$ for some choice of $q \in \{t,f\}$, and continues to $g_1$ by $(q,0,d)$-forward-steps. In general, it must go from every $g_c$ to some $\ell^{c,0}_{q}$ for some choice of $q \in \{t,f\}$. Denote by $choose_c$ these choices, according to the path. This reaches $r_{2\log(k)-1}$, from which $\cal{P}$ must continue by backward-steps all the way to $s^1_1$.
	
	No matter what choices $\cal{P}$ took for the $choose_c$ values, there are at least 4 row vertices that are not yet visited by it. These are the vertices $a_1^i, a_2^j, b_1^i$, and $b_2^j$, such that the $i$ is the index whose binary representation appears as a concatenation of the bits that represent the $choose_c$ decisions for $0 \leq c \leq \log(k)-1$, and $j$ is the index whose binary representation appears as a concatenation of the bits that represent the $choose_c$ decisions for $\log(k) \leq c \leq 2\log(k)-1$. If from $s^1_1$ the path does not continue to $a_1^i$, then it can never reach this vertex because all vertices with edges that are incoming to it are already visited by $\cal{P}$. So from $s^1_1$ the path goes to $a_1^i$. A similar argument gives that then $\cal{P}$ must go into $a_2^j$. If this happens, then $x_{i,j}=1$. From $a_2^j$ the path must continue to $s_1^2$ and from there to $s_2^1$, because all of the outgoing edges from the wheel copies of $a_2^j$ are already visited. A similar argument gives that now $\cal{P}$ must go to $b_1^i$ and to $b_2^j$ in order to continue to $s^2_2$ and to $end$. But if this edge exists, then $y_{i,j}=1$, which implies that $\disj(x,y)=\false$, as needed.
\end{proofof}

Based on the above, we now show the near-quadratic lower bound for directed Hamiltonian path.

\begin{theorem}
\label{thrm:Hamiltonian}
Any distributed algorithm in the \cgst{} model for finding a directed Hamiltonian path or deciding whether there is a directed Hamiltonian path in the input graph requires $\Omega ({n^2}/{\log ^4 n} )$ rounds.
\end{theorem}

\begin{proof}
Claims~\ref{claim:HP} and~\ref{claim:noHP} give that $G_{x,y}$ as constructed is indeed a family of lower bound graphs for the problem of finding a directed Hamiltonian path. The number of vertices in the graphs is $n=\Theta(k\log k)$, which implies that $k=\Theta(n/\log n)$, and the size of $E_{cut}$ is $O(\log k)=O(\log n)$. The proof of Theorem~\ref{thrm:Hamiltonian} is concluded by applying Theorem~\ref{generallowerboundtheorem}.
\end{proof}

\subsubsection{Hamiltonian Cycle and Undirected variants}

We next extend Theorem~\ref{thrm:Hamiltonian} to the case of a cycle and to the undirected variants, as follows.

\begin{theorem}
\label{thrm:Hamiltoniancycle}
Any distributed algorithm in the \cgst{} model for finding a directed Hamiltonian cycle or for deciding whether these such in the input graph requires $\Omega ({n^2}/{\log ^4 n} )$ rounds.
\end{theorem}

\begin{theorem}
\label{thrm:Hamiltonianundirected}
Any distributed algorithm in the \cgst{} model for finding an undirected Hamiltonian path or cycle or deciding whether there is such in the input graph requires $\Omega ({n^2}/{\log ^4 n} )$ rounds.
\end{theorem}

In order to prove the undirected cases in Theorem~\ref{thrm:Hamiltonianundirected}, we employ folklore reductions from the sequential model, and show how can we implement them efficiently in the \cgst{} model. But first, we prove the lower bound for the case of a directed Hamiltonian \emph{cycle}, by making a slight modification to the fixed graph construction of the family of lower bound graphs of the directed Hamiltonian path.

\begin{claim}\label{claim for directed cycle Hamiltonian}
Let $G_{x,y},\ x,y\in \set{0,1}^{k^2}$ be in the family of lower bound graphs we constructed for the directed Hamiltonian path lower bound, and denote by $G'_{x,y}$ the graph obtained from $G_{x,y}$ by adding a single vertex $middle$ and connecting the edges $(middle,start)$ and $(end,middle)$.
Then $G'_{x,y}$ contains a directed Hamiltonian cycle iff $G_{x,y}$ contains a directed Hamiltonian path.
\end{claim}
\begin{proof}
Assume $G_{x,y}$ contains a directed Hamiltonian path, which implies that $\disj(x,y)=\false$ by Claim \ref{claim:noHP}. By the proof of Claim \ref{claim:HP}, there is a Hamiltonian directed path $P$ in $G_{x,y}$ in which the first vertex is $start$ and the last vertex is $end$.  Now in $G'_{x,y}$ we concatenate $P$ to the path $(end,middle,start)$, and denote this cycle by $C$. Every vertex in $G'_{x,y}$ is visited exactly once in $C$, which begins and ends in $start$. Thus $C$ is a directed Hamiltonian cycle in $G'_{x,y}$.

For the other direction, assume there is a directed Hamiltonian cycle in $G'_{x,y}$, denoted by $C$. Since $middle$ has to be visited in $C$, we can deduce that $C$ contains the sub-path $(end,middle,start)$, which implies that it also contains a sub-path $P$ from $start$ to $end$. Since $C$ is Hamiltonian, then so must be $P$, and all that is left to notice is that $P$ does not visit the vertex $middle$, and thus $P$ is a directed Hamiltonian path in $G_{x,y}$.
\end{proof}

\begin{proofof}{Theorem \ref{thrm:Hamiltoniancycle}}
From Claims \ref{claim:noHP},\ref{claim:HP} and ~\ref{claim for directed cycle Hamiltonian}, we immediately deduce that $\disj(x,y)=\false$ iff $G'_{x,y}$ contains a directed Hamiltonian cycle.

We  assign $middle$ to $V_A$, and so the size of the cut between Alice and Bob increases by exactly 1. The proof of Theorem \ref{thrm:Hamiltoniancycle} is concluded by applying Theorem \ref{generallowerboundtheorem}.
\end{proofof}

In order to prove Theorem  \ref{thrm:Hamiltonianundirected},  we show how to implement sequential reductions efficiently in the \cgst{} model. Given a directed graph  $G=(V,E)$, construct the undirected graph $G'=(V',E')$, in which
\[V'=\set{v_{in}|\ v\in V}\cup \set{v_{out}|\ v\in V}\cup \set{v_{middle}|\ v\in V},\]
and
\[E'=\set{(v_{in},v_{middle})|\ v\in V}\cup \set{(v_{middle},v_{out})|\ v\in V}\cup \set{(u_{out},v_{in})|\ (u,v)\in E}.\]

This is exactly as in the classic reduction from directed Hamiltonian cycle to Hamiltonian cycle\cite{DBLP:books/daglib/p/Karp10}. Now we prove the following lemma.
\begin{lemma}\label{Simulation: DHC to HC}
Given an algorithm $A$ with round complexity $T(n)$ that decides whether a given undirected graph has a Hamiltonian cycle, there is an algorithm $B$ that decides whether a given directed graph has a Hamiltonian cycle, with round complexity $O(T(n))$.
\end{lemma}

\begin{proof}
Given the directed graph $G$, algorithm $B$ simulates $G'$ and runs $A$ on it and answers accordingly.
Each vertex $v$ simulates $v_{in},v_{middle},v_{out}$, and each round of $A$ is simulated as follows for a given vertex $v$. For all messages sent on edges of the form $(u_{out},v_{in})$, $v$ sends the appropriate message to $u$ in $G$ in a single round. This is possible since by the reduction we know that $(u,v)\in E$, and in a single round $v$ sends at most one message on each edge. Messages on the edges $(v_{in},v_{middle}),(v_{middle},v_{out})$ require no communication and can be simulated locally by $v$. In the second round, for all messages sent on edges $(v_{out},u_{in})$, $v$ sends the message to $u$ in a single round. Thus, a single round in $G'$ is simulated in only $2$ rounds in $G$, and hence simulating $A$ on $G'$ takes $2\cdot T(3n)$ rounds in $G$, and we return that $G$ has a directed Hamiltonian cycle iff $A$ finds a Hamiltonian cycle in $G'$.

Since $T(n)=O(n^2)$, we deduce that $2\cdot T(3n)=O(T(n))$. Correctness is due to the correctness of the described reduction, and $B$ has round complexity $O(T(n))$, which concludes the proof.
\end{proof}

Now consider the following known reduction from Hamiltonian cycle to Hamiltonian path\cite{DBLP:books/daglib/p/Karp10}. Given an undirected graph $G=(V,E)$, let $v\in V$ be some arbitrary vertex, and construct an undirected graph $G'=(V',E')$, where
\[V'=V\backslash \set{v}\cup \set{v_1,v_2,s,t},\]
and
\[E'=\set{(u,w)|\ u,w\in V\backslash \set{v}} \cup \set{(v_i,u)|\ i\in \set{1,2},\ (v,u)\in E}\cup \set{(s,v_1),(v_2,t)}.\]

\begin{lemma}\label{Simulation: HC to HL}
Given an algorithm $A$ with round complexity $T(n)$ that decides whether a given undirected graph has a Hamiltonian path, there is an algorithm $B$ that decides whether a given undirected graph has a Hamiltonian cycle, with round complexity $O(T(n)+D)$.
\end{lemma}

\begin{proof}
Given $G$, our algorithm $B$ simulates $G'$, and runs $A$ on it and answers accordingly.
First, we find the vertex with the lowest $ID$ in the graph in $O(D)$ rounds. This can be done for example by letting each vertex broadcast the lowest $ID$ value it knows so far in each round. Denote this vertex by $v$. The vertex $v$ simulates $v_1,v_2,s,t$. The rest of the vertices simulate themselves. Now, given a round of $A$, each message on edges of the form $(u,w)$ where $u,w\neq v$ is simulated in a single round. The message for the edge $(u,v_1)$ is sent to $v$ first, and then the message for the edge $(u,v_2)$ is sent. Thus each $u\neq v$ can simulate each round of $A$ on $G'$ in $2$ rounds of $G$. For $v$, the messages for the edges $(s,v_1),(v_2,t)$ can be simulated locally. Messages for edges of the form $(v_1,u)$ are sent by $v$ in a single round, and then it sends all of the messages on edges of the form $(v_2,u)$ in the following round. We conclude that each round of $A$ on $G'$ can be simulated with 2 rounds in $G$. Thus, simulating $A$ on $G'$ takes $2\cdot T(n+3)$ rounds in $G$. We return that $G$ has a Hamiltonian cycle iff $A$ returned that $G'$ has a Hamiltonian path.

Since $T(n)=O(n^2)$, we deduce that $2\cdot T(n+3)=O(T(n))$. Correctness is due to the correctness of the reduction, and $B$ has round complexity $O(T(n))$, which concludes the proof.
\end{proof}

Since $D=O(n)$, Theorem \ref{thrm:Hamiltonianundirected}, is an immediate corollary of Lemmas \ref{Simulation: DHC to HC} and \ref{Simulation: HC to HL}.

\subsubsection{Minimum 2-edge-connected spanning subgraph}

Finally, we show our lower bound for 2-ECSS.

\begin{theorem}
\label{thrm:twoedgeconnected}
Any distributed algorithm in the \cgst{} model for finding a minimum size 2-edge-connected spanning subgraph or deciding whether there is a 2-edge-connected sub-graph with a given number of edges requires $\Omega ({n^2}/{\log ^4 n} )$ rounds.
\end{theorem}

To prove Theorem \ref{thrm:twoedgeconnected}, we prove the following claim.
\begin{claim}\label{claimtwoedgeconnect}
A graph $G=(V,E),\ |V|=n$ contains a 2-edge-connected subgraph with $n$ edges iff $G$ contains a Hamiltonian cycle.
\end{claim}
\begin{proof}
Assume $G$ has a Hamiltonian cycle $C$. The cycle $C$ contains exactly $n$ edges, it is spanning since it is Hamiltonian, and it is 2-edge-connected as it is a cycle.

For the other direction, denote by $C$ a 2-edge-connected subgraph spanning of $G$ with $n$ edges. By the 2-edge connectivity we know that $deg_C(v)\geq 2$ for all $v\in V$. Since $C$ has exactly $n$ edges we conclude that $deg_C(v)=2$ for all $v\in V$. This implies that $C$ is a union of simple cycles, but since $C$ is spanning, we can deduce that $C$ is a Hamiltonian cycle. This concludes the proof.
\end{proof}
Theorem \ref{thrm:twoedgeconnected} is an immediate corollary of Claim \ref{claimtwoedgeconnect} and Theorem \ref{thrm:Hamiltonianundirected}.

\subsection{Minimum Steiner Tree}

In this Section, we show our lower bound for the Steiner tree problem. In the Steiner tree problem, the input is a graph $G=(V,E)$ with weights on the edges, and a set of terminals $S \subseteq V$, and the goal is to find a tree of minimum cost that spans all the terminals. While several efficient approximation algorithm to the problem exist \cite{lenzen2014improved, khan2012efficient, saikia2019simple}, we show that solving the problem \emph{exactly} requires near-quadratic number of rounds. 
Our lower bound is obtained using a reduction from our MDS lower bound graph construction, and we first formally define the notion of reductions between families of lower bound graphs.

\label{subsec:steiner}
\subsubsection{Reductions Between Families of Lower Bound Graphs}

We extend the technique of utilizing lower bounds from communication complexity through families of lower bound graphs, by observing that once a certain such family is found, one may use \emph{sequential} reductions between graph-problems in order to more easily prove additional lower bounds. However, we can only use sequential reductions if they ensure that in addition to not blowing up the number of vertices in the graph, the size of $E_{cut}$ also does not increase by much. We use this technique to prove a near-quadratic lower bound for the \emph{Steiner Tree} problem. This approach is also what led us to obtain our lower bound for the Hamiltonian Path problem, but in that case a direct presentation of the construction eventually seemed more clear.

\begin{theorem}
\label{thrm:reductions}
\textsf{Reductions between Families of Lower Bound Graphs}\\
\noindent Fix a function $f:\{0,1\}^{K} \times \{0,1\}^{K} \to \{\true ,\false \}$ and a predicate $P_1$. Let $\mathcal{F}_{P_1, f}\{G_{x,y}=(V_A \cup V_B, E_{x,y}) |\ x,y\in \{0,1\}^{K}\}$ be a family of lower bound graphs with respect to $f$ and $P_1$. Let $P_2$ be another property, and let $\mathcal{F}_{P_2, f} = \{G_{x,y}'=(V_A' \cup V_B', E_{x,y}') | G_{x, y} \in \mathcal{F}_{P_1, f}  \}$ be a family of graphs. If the following conditions hold, then $\mathcal{F}_{P_2, f}$ is a family of lower bound graphs with respect to $f$ and $P_2$.
	
\begin{enumerate}
\item $V_A'$ and $V_B'$ are determined by $V_A$ and $V_B$, respectively. \label{reductions:listBegin}
\item $E_{x,y}'(V_A', V_A')$ and $E_{x,y}'(V_B', V_B')$ are determined by $E_{x,y}(V_A, V_A)$ and $E_{x,y}(V_B, V_B)$, respectively.
\item $E_{cut}' = E_{x,y}'(V_A', V_B')$ is determined by $E_{cut}$. \label{reductions:listEnd}
\item $G_{x,y}$ satisfies $P_1$ if and only if $G_{x, y}'$ satisfies $P_2$.
\end{enumerate}
\end{theorem}

The proof of Theorem \ref{thrm:reductions} trivially follows from the definition of families of lower bound graphs.

\subsubsection{Steiner Tree Lower Bound}

\begin{theorem}
\label{thrm:ST}
Any distributed algorithm in the \cgst{} model for computing a Steiner tree or for deciding whether there is a Steiner tree of a given size $M$ requires $\Omega ({n^2}/{\log ^2 n} )$ rounds.
\end{theorem}

We use a sequential reduction from MDS to Steiner tree, which is similar in nature to common reductions from vertex cover to Steiner tree. However, the typical sequential reductions transform a graph by adding a vertex for every edge, a characteristic which drastically increases the number of vertices in the graph and thus would only show a linear lower bound in the \cgst{} model. Therefore, we must modify these reductions to fit our model.

We denote by $\mathcal{F}_{P_{MDS}, \disj_{K}}$ the family of lower bound graphs constructed in Theorem~\ref{thm: mds lb}, and we define a family $\mathcal{F}_{P_{ST}, \disj_{K}}$. We then apply Theorem~\ref{thrm:reductions} to show that $\mathcal{F}_{P_{ST}, \disj_{K}}$ is a family of lower bound graphs w.r.t. $\disj_{K}$ and the predicate $P_{ST}$ that says that there exists a Steiner tree of size $4k + 16\log k + 1$ that spans the given set of terminals, where $k=\sqrt{K}$, and the number of vertices in the graph is $n=8k+24\log k$.

\paragraph{The family of graphs:} For any $G_{x, y}(V_A \cup V_B, E_{x, y}) \in \mathcal{F}_{P_{MDS}, \disj_{K}}$, define $G_{x, y}'(V_A' \cup V_B', E_{x, y}') \in \mathcal{F}_{P_{ST}, \disj_{K}}$ as follows. Let $V_A' = V_A \cup \widetilde{V_A}, V_B' = V_B \cup \widetilde{V_B}$, where $\widetilde{v}$ is a copy of a vertex $v\in V$, and for a set $U\subseteq V$ we denote $\widetilde{U}=\{\widetilde{v}~|~v \in U\}$.
Let $E_{x, y}'$ consist of four types of edge sets: (1) \emph{identity edges} which connect $\widetilde{v}$ to $v$: $\{  (\widetilde{v}, v) | v \in  V_A \cup V_B \}$, (2) \emph{original edges} which connect $\widetilde{u}$ to $v$ according to $E_{x,y}$: $\{ (\widetilde{u}, v)  | (u, v) \in E_{x, y}  \}$, (3) \emph{clique edges} which connect $\widetilde{V_A}, \widetilde{V_B}$ in two cliques: $\widetilde{V_A} \times \widetilde{V_A}, \widetilde{V_B} \times \widetilde{V_B}$, and (4) exactly two \emph{crossing edges} which connect six specific vertices $e_1 = (\widetilde{f}_{A_1} ^0, \widetilde{f}^0 _{B_1}), e_2 =(\widetilde{t}_{A_1} ^0, \widetilde{t}^0 _{B_1})$. Notice that these are added only for $h=0$ and only for $A_1,B_1$, and not for any other bit-gadget row.

\begin{claim}
\label{claim:HST_is_a_LB_family}
The family $\mathcal{F}_{P_{ST}, \disj_{K}}$ is a family of lower bound graphs with respect to $\disj_{K}$ and $P_{ST}$ given the set of terminals $\term = V_A \cup V_B$. It holds that $|E_{cut}|=O(\log n)$.
\end{claim}

Combining Claim~\ref{claim:HST_is_a_LB_family} and Theorem~\ref{generallowerboundtheorem} proves Theorem~\ref{thrm:ST}.

\begin{proofof}{Claim \ref{claim:HST_is_a_LB_family}}
First, the only edges crossing the cut between $V_A'$ and $V_B'$ are $O(\log n)$ edges that correspond to the \emph{original edges} and the three \emph{crossing edges}, which gives that $|E_{cut}| = O(\log n)$.

We apply Theorem~\ref{thrm:reductions} to $\mathcal{F}_{P_{MDS}, \disj_{K}}$ and $\mathcal{F}_{P_{ST}, \disj_{K}}$. Let $G_{x, y}(V_A \cup V_B, E_{x, y}) \in \mathcal{F}_{P_{MDS}, \disj_{K}}$ and let $G_{x, y}'(V_A' \cup V_B', E_{x, y}') \in \mathcal{F}_{P_{ST}, \disj_{K}}$ be its corresponding graph. Observe that conditions \ref{reductions:listBegin}-\ref{reductions:listEnd} of Theorem~\ref{thrm:reductions} are trivially met due to the definition of $G_{x, y}'$. It remains to show that $G'_{x, y}$ satisfies $P_{ST}$ if and only if $G_{x,y}$ satisfies $P_{MDS}$.
	
Assume $G_{x,y}$ satisfies $P_{MDS}$, and let $C \subseteq V_A \cup V_B$ be a dominating set of size $4\log k +2$. Denote $C_A = C \cap V_A, C_B = C \cap V_B$. Notice that in $G'_{x,y}$, the vertices in $\widetilde{C_A}$ and the vertices in $\widetilde{C_B}$ are connected as cliques, and so there exist trees $T'_A$ and $T'_B$ in $G_{x,y}'$ which span them and are of size $|\widetilde{C_A}| - 1$ and $|\widetilde{C_B}| - 1$, respectively. Further, because we know that $C$ contains exactly either $\widetilde{f}_{A_1} ^0, \widetilde{f}^0 _{B_1}$ or $\widetilde{t}_{A_1} ^0, \widetilde{t}^0 _{B_1}$, we have that one of the \emph{crossing edges} $e_1, e_2,$ may be used to connect $T'_A, T'_B$ to form a tree $T'$ in $G_{x,y}'$ which spans $\widetilde{C}$ and is of size $|\widetilde{C}| - 1 = 4\log k + 1$. Lastly, notice that since $C$ is a dominating set in $G_{x, y}$, then for every vertex $v \in \term = V_A \cup V_B$ there exists a vertex $\widetilde{u} \in \widetilde{C}$ such that in $G'_{x,y}$ there is an edge $(\widetilde{u}, v)$. Thus, for every vertex $v \in \term = V_A \cup V_B$, add exactly one such edge $(\widetilde{u}, v)$ to $T'$ in order to create a Steiner tree spanning $\term$ of size exactly $4k + 16\log k + 1$, since $|\term| = 4k + 12 \log k$.
	
For the other direction, assume that there exists a Steiner tree $T'$ in $G_{x,y}'$ that spans $\term$ and is of size $4k + 16\log k + 1$. We first show that the existence of $T'$ implies the existence of a tree $T''$ in $G_{x,y}'$ which spans $\term$, is of the same size as $T'$, and has one additional property: every $v \in \term$ is a leaf in $T''$. Let $v \in \term$ be any vertex such that at least $2$ edges in $T'$ are incident to it, and let $N_{T'}(v)$ be the neighbors of $v$ w.r.t. the edges in $T'$. Observe that $V_A\cup V_B$ is an independent set, and thus $N_{T'}(v) \subseteq \widetilde{V_A} \cup \widetilde{V_B}$ since all neighbors in $G_{x,y}'$ of $v$ are vertices in $\widetilde{V_A} \cup \widetilde{V_B}$. Therefore, it is possible to form $T''$ from $T'$ by replacing the edges between $v$ and all but one of the vertices in $N_{T'}(v)$ with clique edges. This ensures that $v$ is a leaf in $T''$, which remains spanning. Further, notice that this does not increase the number of edges in $T''$, and so $|T''| \leq 4k + 16 \log k + 1$. Let $V_{T'}$ be the set of vertices in $T'$, and observe that $V_{T'} \setminus \term$ is a dominating set in $G_{x,y}$. Finally, observe that $|V_{T'} \setminus \term| = |T''| + 1 - |\term| = 4 \log k + 2$, and so there is be a dominating set in $G_{x,y}$ of size $4 \log k + 2$.
\end{proofof}

\subsection{Max-Cut}
\label{subsec:MaxCut}

The max-cut problem requires each vertex to choose a \emph{side} such that if $S$ is the set of vertices that choose one side, then the cut $E(S, V \setminus S)$ is the largest among all cuts in the graph. A trivial random assignment is known to produce a 1/2-approximation, and requires no communication. In~\cite{DBLP:conf/algosensors/Censor-HillelLS17}, an algorithm is given for obtaining this approximation factor deterministically, within $\tilde{O}(\Delta + \log^* n)$ rounds. Recently,~\cite{KawarabayashiS18} showed that a $(1/2-\epsilon)$-approximation can be obtained in $O(\log^* n)$ rounds deterministically, and in $O(1/\epsilon)$ rounds using randomization, which also hold for the \emph{weighted} version of the problem, in which the cut $E(S, V \setminus S)$ has a maximum weight.

We show in Section~\ref{subsec:maxcutAlg} below that for the unweighted case, the algorithm of~\cite{DBLP:phd/de/Zelke2009} which obtains a $(1-\epsilon)$-approximation by computing an exact solution for a randomly sampled subgraph, can be implemented in the \cgst{} model within $\tilde{O}(n)$ rounds.
However, once we need a truly exact solution, we now show that for the weighted version, any algorithm requires $(n^2/ \log^2 n)$ rounds. 

\subsubsection{Lower bound for max cut}

The essence of our lower bound construction follows the framework of Theorem~\ref{generallowerboundtheorem}, and in particular shares some of the structure of the MDS construction, in the sense of having row vertices that are connected to bit-gadget vertices according to their binary representation.

For an edge-weighted graph $G=(V,E,w)$ and a set $S \subseteq V$, the cut defined by $S$ is the set of edges $C=E(S,V\setminus S) = \{\{v,u\}~|~ \{v,u\}\in E, v\in S, u\notin S\}$. The weight of a cut $C$ is $w(C)=\sum_{e\in C}{w(e)}$.

Let $M=k^4\cdot (8\log(k)+4) + k^3(12\log(k)-4) +4k^2+4k$.

\begin{theorem} \label{thm:max-cut}
Any distributed algorithm in the \cgst{} model for computing a maximum weighted cut or for deciding whether there is a cut of a given weight $M$ requires $\Omega\left(n^2/\log^2 n \right)$ rounds.
\end{theorem}

\textbf{Intuition for the construction:}
To prove Theorem \ref{thm:max-cut} we use a similar construction  as for MDS.
The key technical hurdle is that inserting an edge into a graph according to $x,y$ (imagine for now the MDS graph) \emph{may not change the value of the maximum cut}, e.g., if it connects two vertices on its same side. The novelty of our construction is in circumventing the above by adding weights to all edges, and adding some (small) constant number of new vertices, among them two vertices $N_A$ and $N_B$, connected to all $A_1\cup A_2$ and $B_1\cup B_2$ vertices, respectively. The role of $N_A,N_B$ is the following. As in the MDS construction, we add an edge of weight 1 between $a^i_1$ and $a^j_2$ iff $x_{i,j}=0$, and similarly for $y$. The crux of our construction is that for each row vertex, say $a^i_1$, the weight of its edge to $N_A$ is equal to the number of vertices in $A_2$ to which $a^i_1$ is \emph{not} connected. Thus, the sum of weights of edges going from $a^i_1$ to vertices in $A_2\cup\{N_A\}$ is exactly $k$. Intuitively, if we now add another edge from $a^i_1$ to some $a^j_2$ that is on the same side of a given cut, and if $a^i_1$ and $N_A$ are not on the same side, then the value of the cut now decreases by $1$. We give weights to remaining edges, such that some heavy-weighted edges force any max-cut to take both $t^h_S$ vertices or both $f^h_S$ vertices from every row $h$ to the same side of the cut. Combining the above ingredients gives that in order for the maximum cut to retain its weight, the inputs strings $x,y$ must not be disjoint, which is what allows us to obtain our lower bound.

\textbf{The family of lower bound graphs:} Let $k$ be a power of $2$,
and build a family of graphs $G_{x,y}$ with respect to $f=\disj_{k^2}$ and the following predicate $P$:
the graph $G_{x,y}$ contains a cut of weight $M$.

\textbf{The fixed graph construction:} See Figure~\ref{fig:max_cut} for an illustration. The set of vertices $V$ contains four sets of $k$ vertices each, denoted by $A_1 = \{a_1^j\mid 0\le j \le k-1\}, A_2 = \{a_2^j\mid 0\le j \le k-1\}, B_1 = \{b_1^j\mid 0\le j \le k-1\}, B_2 = \{b_2^j\mid 0\le j \le k-1\}$. For each $S\in\{A_1,A_2,B_1,B_2\}$ there are two sets of vertices of size $\log{k}$, denoted by $F_S = \{f_S^h\mid 0\le h \le \log{(k)} -1\} $ and $T_S = \{t_S^h\mid 0\le h \le \log{(k)} -1\}$. In addition, there are five vertices denoted by $C_A,\bar{C}_A,C_B,N_A,N_B$.

The vertices are connected as follows. The vertices $C_A$ and $N_A$ are connected with an edge of weight $k^4$. Similarly, the vertices $C_B$ and $N_B$ are connected with an edge of weight $k^4$. The vertices $C_A$ and $\bar{C}_A$ are connected with an edge of weight $k^4$, and the vertices $\bar{C}_A$ and $C_B$ are connected with an edge of weight $k^4$. For each $z\in\{1,2\}$ and $0\le h\le \log{(k)} -1 $, the vertices $(t_{A_z}^h,f_{A_z}^h,t_{B_z}^h,f_{B_z}^h)$ are connected as a 4-cycle, with edges of weight $k^4$. For each $S\in \{A_1,A_2,B_1,B_2\}$ and for each $0\le j \le k-1$ we connect the vertex $s^j$ as follows. Let $\{j_h\}_{h=0}^{\log{(k)} -1}$ be the binary representation of $j$, i.e., $j=\sum_{h=0}^{\log{(k)} -1} 2^h j_h$, and define $Bin(s^j)=\{t_S^h\mid j_h=1\}\cup\{f_S^h\mid j_h=0\}$. The vertex $s^j$ is connected to every vertex in $Bin(s^j)$ with an edge of weight $2k^2$. For every vertex $a_z^j\in A_1\cup A_2$, connect it with $C_A$ by an edge of weight $2k^2\log (k) - k^2$. Similarly for every vertex $b_z^j\in B_1\cup B_2$, connect it with $C_B$ by an edge of weight $2k^2\log (k) - k^2$. Finally, every vertex $a_z^j\in A_1 \cup A_2$ is connected to $N_a$, and every vertex $z_i^j\in B_1 \cup B_2$ is connected to $N_B$. The weight of these edges is determined by the input strings.

\begin{figure}
	\centering
	\includegraphics[width=\textwidth]{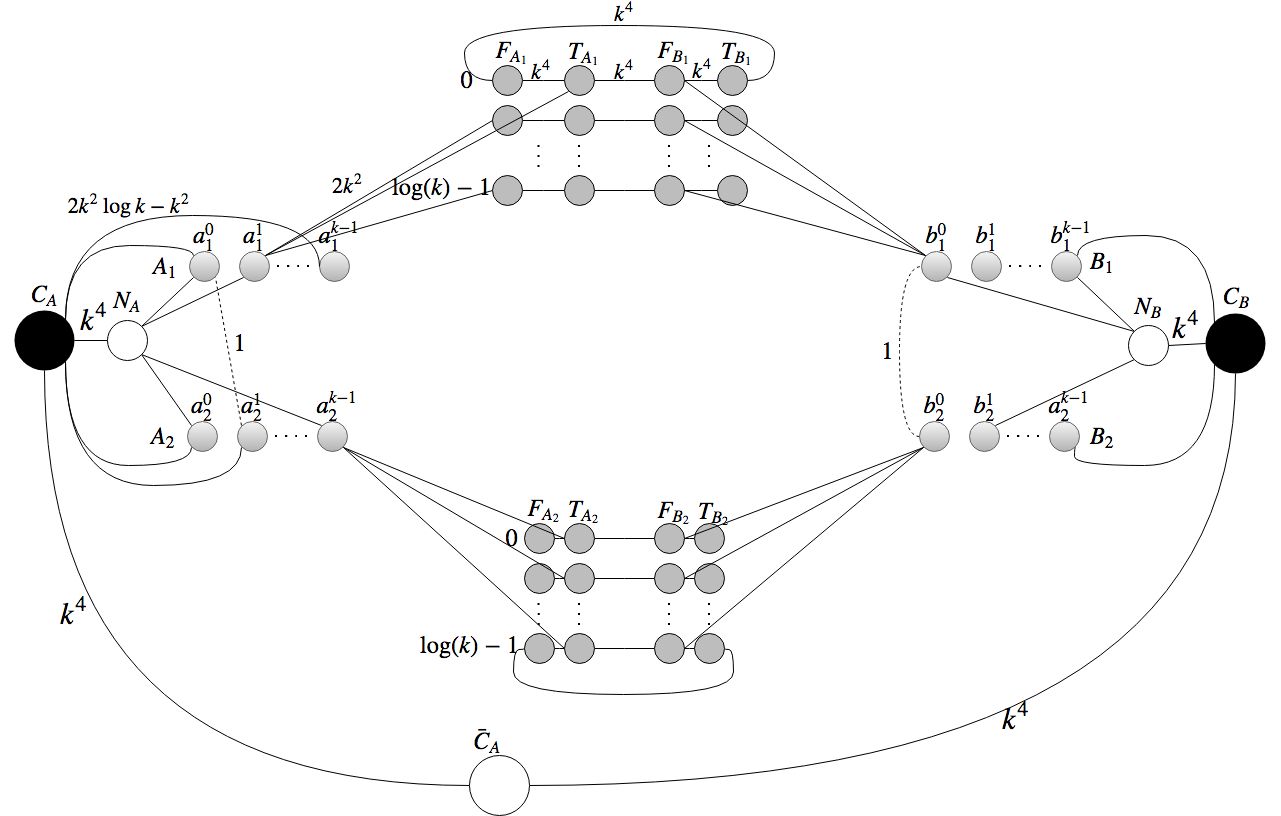}
	\caption{The family of lower bound graphs for deciding the weight of the max-cut, with many edges and weights omitted for clarity.}
	\label{fig:max_cut}
\end{figure}

\textbf{Constructing $G_{x,y}$ from $G$ given $x,y\in\{0,1\}^{k^2}$:}
We index the strings $x,y\in\{0,1\}^{k^2}$ by pairs of indices of the form $(i,j)\in\{0,1,\dots,k-1\}^2$. Now we augment $G$ in the following way: For all pairs $(i,j)$ such that $x_{i,j} = 0$ we connect the vertices $a_1^i$ and $a_2^j$ with an edge of weight $1$. Similarly for every $(i,j)$ such that $y_{i,j} = 0$ we connect the vertices $b_1^i$ and $b_2^j$ with an edge of weight $1$.

In addition, for every $0\le i\le k-1$, we set the weight of the edge $\{a_1^i,N_A\}$ to be $\sum_{j=0}^{k-1}x_{i,j}$ (intuitively, this promises that the total weight of the edges that are connected to $a_1^i$ is always exactly $k$). Similarly, we set the weights of the edges $\{N_A,a_2^i\},\{N_B,b_1^i\}$ and $ \{N_B,b_2^i\}$, to be $\sum_{j=0}^{k-1}x_{j,i}, \sum_{j=0}^{k-1}y_{i,j}$ and $\sum_{j=0}^{k-1}y_{j,i}$ respectively.

This completes the construction of $G_{x,y}$. Let $S\subseteq V$ be such that $C=E(S,V\setminus S)$ is a maximum weight cut of $G_{x,y}$. Assume without loss of generality that $C_A\in S$. We prove the following sequence of claims about $S$.

\begin{claim} \label{clm:max-cut:k^4 edges}
It holds that $C_B\in S$ and $N_A,N_B,\bar{C}_A\in V\setminus S$, and for every $z\in\{1,2\}$ and $0\le h \le \log (k) - 1$ it holds that $t_{A_z}^h\in S \iff t_{B_z}^h\in S \iff f_{A_z}^h\notin S \iff f_{B_z}^h\notin S$.
\end{claim}
\begin{proof}
Let $E'$ be the set of edges that are determined to be in the cut by the statement of the claim. Notice that $E'$ is exactly the set edges of weight $k^4$. The rest of the edges are $O(k)$ edges of weight $O(k^2\log k)$, $O(k\log k)$ edges of weight $O(k^2)$, and $O(k^2)$ edges of weight $1$. Therefore, the total weight of the edges that do not have weight $k^4$, is $O(k^3 \log k)$ which is strictly smaller than $k^4$. Thus, if we show that there is a cut that contains $E'$, then any maximum cut must be such. Taking $\tilde{S}=T_{A_1}\cup T_{A_2}\cup T_{B_1}\cup T_{B_2}\cup \{C_A,C_B\}$ produces a cut $E(\tilde{S}, V\setminus\tilde{S})$ which contains $E'$, completing the proof.
\end{proof}

\begin{claim} \label{clm:max-cut:bin}
For every $z\in\{1,2\}$ and $0\le j\le k-1$ it holds that $a_z^j\in S \iff Bin(a_z^j)\cap S = \emptyset$ and $b_z^j\in S \iff Bin(b_z^j)\cap S = \emptyset$.
\end{claim}

\begin{proof}
We prove the claim for $A_1$, and the proofs for $A_2,B_1,B_2$ are similar.

For the first direction, assume $a_1^j\in S$. If $Bin(a_1^j)\cap S\ne\emptyset$ then there exists $u\in Bin(a_1^j)$ such that $u\in S$. Let $S' = S\setminus\{a_1^j\}$, and $C'=E(S',V\setminus S')$. We show that $w(C')>w(C)$, which contradicts $C$ being a maximum cut. Notice that $w(C')-w(C)=w(C'\setminus C) - w(C\setminus C')$. Since $u\in S'$, it holds that $|Bin(a_1^j)\cap(V\setminus S')|\le \log (k) -1$. The edges that connect $a_1^j$ to $Bin(a_1^j)\cap(V\setminus S')$ are of weight $2k^2$, and are in $C\setminus C'$. By Claim~\ref{clm:max-cut:k^4 edges}, the rest of the edges in $C\setminus C'$, if they exist, are edges that connect $a_1^j$ to vertices in $A_2\cup\{N_A\}$, of which there are no more than $k$, and whose weights are $1$. Thus, $w(C\setminus C') \le 2k^2(\log (k) -1 ) + k = 2k^2\log k - 2k^2 + k$. The edge $\{a_1^j,C_A\}$ is in $C'\setminus C$ and its weight is $2k^2\log k -k^2$, which gives $w(C' \setminus C) \ge 2k^2\log k -k^2 > 2k^2\log k - 2k^2 + k \ge w(C\setminus C')$. Thus $w(C') > w(C)$, which is a contradiction, and therefore  $Bin(a_1^j)\cap S = \emptyset$.
	
For the other direction, assume $Bin(a_1^j)\cap S = \emptyset$. If $a_1^j\notin S$, define $S' = S \cup \{a_1^j\}$, and $C'=E(S',V\setminus S')$. As before, we show that $w(C')>w(C)$. Since $Bin(a_1^j) \cap S = \emptyset$ and $a_1^j\notin S$, and by Claim~\ref{clm:max-cut:k^4 edges}, we have that $C\setminus C'$ only contains edges that connect $a_1^j$ to vertices in $A_2 \cup \{C_A\}$. These have a total weight of no more than $2k^2 \log k -k^2 + k$. Since $a_1^j \in S'$ and $Bin(a_1^j)\cap S' = \emptyset$, we have that $C'$ contains all the edges that connect $a_1^j$ to vertex in $Bin(a_1^j)$. There are exactly $\log k$ such edges, and their weights are $2k^2$. Therefore, the total weight of $C'\setminus C$ is at least $2k^2 \log k$. Thus $w(C'\setminus C) \ge 2k^2 \log k > 2k^2 \log k -k^2 + k \ge w(C\setminus C')$. We get that indeed $w(C') > w(C)$, which contradicts $C$ being a maximum cut, and therefore $a_1^j \in S$.
\end{proof}

\begin{claim} \label{clm:max-cut:same-one}
For every $z\in\{1,2\}$, and $0\le j\le k-1$ it holds that $a_z^j\in S \iff b_z^j\in S$, and there exists exactly one $0\le j^*\le k-1$ such that $a_z^{j^*}\in S$.
\end{claim}

\begin{proof}
From Claim~\ref{clm:max-cut:k^4 edges}, we have that for every $z\in \{1,2\}$ and $0\le j \le k-1$, it holds that $Bin(a_z^j) \cap S = \emptyset \iff Bin(b_z^j) \cap S = \emptyset$, and also that there exists exactly one $0\le j^* \le k-1$ such that $Bin(a_z^{j^*}) \cap S = \emptyset$, which is the value whose binary representation is determined by the set of indexes $h$ for which $t^h_{A_z} \not\in S$. From Claim~\ref{clm:max-cut:bin} we conclude that $a_z^j \in S \iff b_z^j\in S$, and that there exists exactly one $0\le j^* \le k-1$ such that $a_z^{j^*}\in S$.
\end{proof}

For a set of vertices $U\subset V$, we denote $E[U] = \{\{u,v\}~|~ u,v\in U, \{u,v\}\in E\}$, which is the set of edges that connect two vertices in $U$.

\begin{claim} \label{clm:max-cut:constant}
Let $E'=C\setminus E[A_1\cup A_2 \cup B_1\cup B_2 \cup \{N_A,N_B\}]$. Then the  total weight of edges in $E'$ is $M'=k^4\cdot (8\log(k)+4) + k^3(12\log(k)-4) +4k^2$, regardless of $x$ and $y$.
\end{claim}

\begin{proof}
To prove the claim, we argue that for every $t\in \mathbb{N}$, the number of edges of weight $t$ in $E'$ does not depend on $x,y$.
		
By Claim~\ref{clm:max-cut:k^4 edges}, all of the edges of weight $k^4$ in $G_{x,y}$ are in $C$, and none of them are in $E[A_1\cup A_2 \cup B_1\cup B_2 \cup \{N_A,N_B\}]$. This implies that all of the edges of weight $k^4$ are in $E'$. The number of these edges is $8\log(k)+4$ regardless of $x,y$.
	
By Claim~\ref{clm:max-cut:same-one}, there are exactly $k-1$ vertices in each of $A_1,A_2,B_1,B_2$ which are not in $S$, and therefore there are exactly $4(k-1)$ edges in $E'$ of weight $2k^2\log k -k^2$, since $C_A,C_B\in S$.
	
Again by Claim~\ref{clm:max-cut:same-one}, we know that there exists $0\le j^* \le k-1$, such that for every $j\ne j^*$ it holds that $a_1^j\notin S$. By Claims~\ref{clm:max-cut:k^4 edges} and~\ref{clm:max-cut:bin}, we know that $Bin(a_1^{j^*}) \cap S = \emptyset$ and for every $0\le h \le \log (k) - 1$ it holds that $t_{A_1}^h\in S \iff f_{A_1}^h\notin S$. We conclude that $|Bin(a_1^j)\cap S|$ is equal to $|Bin(a_1^j)\setminus Bin(a_1^{j^*})|$ which is the number of bits that are different between the binary representations of $j$ and $j^*$. Therefore, for every $0\le h \le \log (k) - 1$, there are exactly $\binom{\log k}{h}$ indices $j_{\ell}$ for $1\leq \ell \leq {\binom{\log k}{h}}$ such that $|Bin(a_1^{j_{\ell}})\setminus Bin(a_1^{j^*})| = h$.
Since $a_1^{j^*}\in S$  and $Bin(a_1^{j^*}) \subseteq V\setminus S$, we get that all the edges that connect $a_1^{j^*}$ to $Bin(a_1^{j^*})$ are also in $E'$. Similar arguments hold for $A_2,B_1$ and $B_2$.
Thus, the number of edges of weight $2k^2$ in $E'$ is exactly $4(\log k + \sum_{h=1}^{\log k}h\binom{\log k}{h}) = 2(k + 2)\log k$.
	
The only edges in the graph we did not address above are edges in $E[A_1\cup A_2 \cup B_1\cup B_2 \cup \{N_A,N_B\}]$, which are not in $E'$. Therefore, $M'=w(E')= k^4\cdot (8\log(k)+4) + (2k^2\log k -k^2)\cdot 4(k-1) + (2k^2)\cdot 2(k + 2)\log k = k^4\cdot (8\log(k)+4) + 8k^3\log(k)-4k^3-8k^2\log(k)+ 4k^2 + 4k^3\log(k)+8k^2\log(k) = k^4\cdot (8\log(k)+4) + k^3(12\log(k)-4) +4k^2$, which completes the proof.
\end{proof}

\begin{lemma} \label{lem:max-cut}
Given $x,y\in \{0,1\}^{k^2}$, the graph $G_{x,y}$ contains a cut $C=E(S,V\setminus S)$ such that  $w(C) \ge M$ if and only if $\disj_{k^2}(x,y) = \false$.
\end{lemma}

\begin{proof}
From Claim~\ref{clm:max-cut:constant}, we have that for any $x,y\in\{0,1\}^{k^2}$ and any maximum cut $C=E(S,V\setminus S)$ of $G_{x,y}$, the total weight of $C\setminus E[A_1\cup A_2 \cup B_1\cup B_2 \cup \{N_A,N_B\}]$ is $M'$. 

	

For the first direction of the proof, assume $\disj_{k^2}(x,y) = \false$. Therefore, there exists $(j_1,j_2)$ such that $x_{j_1,j_2} = y_{j_1,j_2} = 1$. We define $S$ as follows. We take $a_1^{j_1}, b_1^{j_1}, a_2^{j_2}, b_2^{j_2}, C_A, C_B$ into $S$. In addition, for every $z\in \{1,2\}$, we add $(T_{A_z}\cup F_{A_z} )\setminus Bin(a_z^{j_z})$ and $(T_{B_z}\cup F_{B_z} )\setminus Bin(b_i^{j_z})$ to $S$.

By construction, the conditions that are claimed to hold for a maximum cut in Claims~\ref{clm:max-cut:k^4 edges},~\ref{clm:max-cut:bin}, and~\ref{clm:max-cut:same-one} hold for $S$ and $C$ as defined above.
Therefore, the proof of Claim~\ref{clm:max-cut:constant} carries over, which means that $w(C\setminus E[A_1\cup A_2 \cup B_1\cup B_2 \cup \{N_A,N_B\}]) = M'$.
	
Now, the total weight of edges from $a_1^{j_1}$ to vertices in $A_2 \cup \{N_A\}$ is exactly $k$. All of the vertices in $A_2 \cup \{N_A\}$ except for $a_2^{j_2}$ are not in $S$, and since $x_{j_1,j_2} = 1$, there is no edge between $a_1^{j_1}$ and $a_2^{j_2}$, and therefore all of the edges from $a_1^{j_1}$ to vertices in $A_2 \cup \{N_A\}$ are in $C$. Similar arguments holds for $a_2^{j_2}, b_1^{j_1}$ and $b_2^{j_2}$, and therefore $w(C) \ge M' + 4k = M$.

For the other direction, assume there exist a maximum cut $C=E(S,V\setminus S)$ such that $w(C)\ge M = M' +4k$.
By Claim~\ref{clm:max-cut:constant} we have that $w(C\setminus E[A_1\cup A_2 \cup B_1\cup B_2 \cup \{N_A,N_B\}]) = M'$ and thus, $w(C\cap E[A_1\cup A_2 \cup B_1\cup B_2 \cup \{N_A,N_B\}]) \ge 4k$.

By Claim~\ref{clm:max-cut:same-one}, there exist two indices $0\le j_1^*,j_2^*\le k-1$ such that $a_1^{j_1^*},b_1^{j_1^*},a_2^{j_2^*},b_2^{j_2^*}\in S$, and no other vertex in $A_1\cup A_2\cup B_1 \cup B_2$ is in $S$. By Claim~\ref{clm:max-cut:k^4 edges} it holds that $N_A,N_B\notin S$, and therefore all the edges in $C\cap E[A_1\cup A_2 \cup B_1\cup B_2 \cup \{N_A,N_B\}]$ are connected to one of the vertices $a_1^{j_1^*},b_1^{j_1^*},a_2^{j_2^*},b_2^{j_2^*}$.
For each vertex in $a_1^{j_1^*},b_1^{j_1^*},a_2^{j_2^*},b_2^{j_2^*}$ the total weight of edges that are connected to it and to a vertex in $A_1\cup A_2 \cup B_1\cup B_2 \cup \{N_A,N_B\}$ is $k$. Since $w(C\cap E[A_1\cup A_2 \cup B_1\cup B_2 \cup \{N_A,N_B\}]) \ge 4k$, we have that all of those edges are in $C$, are therefore there is no edge between $a_1^{j_1^*}$ and $a_2^{j_2^*}$, and no edge between $b_1^{j_1^*}$ and $b_2^{j_2^*}$. Since for every $(i,j)$ there is an edge between $a_1^i$ and $a_2^j$ if and only if $x_{i,j} = 0$, we get $x_{j_1^*,j_2^*} = 1$, and from the same reason $y_{j_1^*,j_2^*} = 1$, thus $\disj_{k^2}(x,y)=\false$.
\end{proof}

\begin{proofof}{Theorem~\ref{thm:max-cut}}
We partition $V$ into two sets $V_A = A_1\cup A_2 \cup T_{A_1} \cup F_{A_1} \cup T_{A_2} \cup F_{A_2} \cup \{C_A,N_A,\bar{C}_A\}$, and $V_B = V\setminus V_A$.
Note that $n=\Theta(k)$, and therefore the parameter $K$ we use for $\disj$ is $K=k^2=\Theta(n^2)$.
The set $E_{cut}$ is the set of edges between $T_{A_1}\cup F_{A_1} \cup T_{A_2} \cup F_{A_2}$ and $T_{B_1}\cup F_{B_1} \cup T_{B_2} \cup F_{B_2}$, hence its size is $\Theta(\log k) = \Theta(\log n)$.
Lemma~\ref{lem:max-cut} gives that $\{G_{x,y}\}$ is a family of lower bounds graphs, and hence applying Theorem~ \ref{generallowerboundtheorem} implies a lower bound of $\Omega(k^2/\log k\log n)=\Omega(n^2/\log^2 n)$ rounds.
\end{proofof}

\subsubsection{A (1-$\epsilon$) approximation for max-cut}
\label{subsec:maxcutAlg}
While the above shows that finding an exact solution for max-cut requires a near-quadratic number of rounds, here we show that an almost exact solution for the unweighted case can be obtained in a near-linear number of rounds. Specifically, we show a simple distributed algorithm for computing a $(1-\epsilon)$-approximation of the maximum cut, proving the following.
\begin{theorem}\label{theorem:max cut approx}
Given a constant $\epsilon>0$, there is a randomized distributed algorithm in the \cgst{} model that computes a $(1-\epsilon)$-approximation of the maximum cut in $G$ in $\tilde{O}(n)$ rounds, with high probability.
\end{theorem}

Our algorithm will be a rather straightforward adaptation to the \cgst{} model of the sampling technique presented in \cite{DBLP:phd/de/Zelke2009}. The idea is to sample each edge independently with some probability $p$. This results in a subgraph of $G$, denoted by $G_p$, that has $O(mp)$ edges in expectation, where $m$ is the number of edges in $G$. We then have a single vertex $w\in V$ learn the entire graph $G_p$ in $O(mp+D)$ rounds, where $D$ is the diameter of $G$. Now, $w$ locally computes the maximum cut in $G_p$, denoted by $C$, and denote the value of this cut in $G_p$ by $c_p^*$. Then we simply return $C$ and the value $c_p^*/p$ as our approximation.

For the correctness of this procedure we use the following result from~\cite[Theorem 21]{DBLP:phd/de/Zelke2009}.
\begin{lemma}\label{lemma:max cut}
Given a graph $G=(V,E)$ with $n$ vertices and $m$ edges and a constant $\epsilon>0$, denote by $c^*$ the size of its maximum cut.
Let $G_p$ be a subgraph of $G$ obtained by independently sampling each edge $e\in E$ into $G_p$ with probability $p=n\log ^s (n)/m$ for a constant $s$ that depends only on $\epsilon$, and denote a maximum cut in $G_p$ by $C$ and its value by $c_p^*$. With high probability, $c_p^*/p$ is a $(1-\epsilon)$-approximation of $c^*$.
\end{lemma}

To prove Theorem~\ref{theorem:max cut approx}, we show how to implement this sampling procedure in the \cgst{} model in near linear time.
\begin{proofof}{Theorem~\ref{theorem:max cut approx}}
Each vertex $v\in V$ checks for each of its neighbors $u\in N(v)$ whether $ID(v)<ID(u)$. If so, $v$ samples the edge $(u,v)$ into $G_p$ with probability $p=n\log ^s(n)/m$, where $s=s(\epsilon)$ is the constant needed in Lemma~\ref{lemma:max cut}. Now, the vertex $w$ with the smallest $ID(w)$ builds a BFS tree $T$ rooted at $w$, collects all the edges of $G_p$ over $T$, and computes $C$ and $c^*_p$, and sends this information back through $T$.

The correctness follows directly from Lemma~\ref{lemma:max cut}. For the round complexity, deciding on $w$ can be done in $O(n)$ rounds and building $T$ can be done in $O(D)$ rounds, where $D$ is the diameter of $G$.\footnote{One can replace this by building multiple BFS trees, one from each node, and dropping any procedure that encounters a BFS from a root with a smaller $ID$. This would require only $D$ rounds, but we will need to pay $O(n)$ rounds anyway in the rest of the algorithm.} Sending all edges of $G_p$ to $w$ can be done in $O(m_p+D)$ where $m_p$ is the number of edges in $G_p$. The expected value of $m_p$ is $O(mp)$ and by a standard Chernoff bound we can deduce that $G_p$ has $O(n\log ^s(n))$ edges, with high probability. The local computation done by $w$ requires no communication. Finally, downcasting the cut $C$ and its value can be completed in $O(n)$ rounds. Therefore, the total number of rounds required for the algorithm is $O(n\log ^s(n))$, with high probability.
\end{proofof}

\section{Lower bounds for bounded degree graphs} \label{sec:exact-bounded(b)}

In this section we show that finding exact solution for MaxIS, MVC, MDS and minimum 2-spanner remains difficult even if the graph has a bounded degree. In bounded degree graphs we can solve all these problems in $O(n)$ rounds by learning the whole graph. We show a nearly tight lower bound of $\tilde{\Omega}(n)$ rounds. It is easy to show a linear lower bound for MaxIS or MVC in a cycle that has bounded-degree 2. However, a cycle has a linear diameter $D$, and the lower bound follows from the fact that these problems are global problems that require $\Omega(D)$ rounds. Here, we show that $\tilde{\Omega}(n)$ rounds are required even if the graph has logarithmic diameter.
We mention that all the above problems admit efficient constant approximations in bounded degree graphs, where in general graphs currently there are no efficient constant approximations for MaxIS, MDS and minimum 2-spanner in the \cgst{} model. However, when it comes to exact solutions all these problems are still difficult. 

\subsection{Converting a graph to a bounded-degree graph} \label{bounded_reductions}

To show a lower bound for MaxIS, we start by describing a series of (non-distributed) reductions. Then, we explain how we apply these reductions on a family of lower bound graphs for MaxIS to obtain a new family of lower bound graphs for MaxIS where now all the graphs have a bounded degree.
We need the following definitions.
For a graph $G$, we denote by $\alpha(G)$ the size of maximum independent set in $G$. For a CNF formula $\phi$, we denote by $f(\phi)$ the maximum number of clauses that can be satisfied in $\phi$. We say that an assignment $\pi$ to the variables of $\phi$ is \emph{maximal} if it satisfies $f(\phi)$ clauses.

Our construction is based on a series of (non-distributed) reductions between MaxIS and max 2SAT instances. These reductions are applied on a family of lower bound graphs for MaxIS to obtain a new family of lower bound graphs for MaxIS which have bounded degrees.
First, we replace our graph $G$ with a CNF formula $\phi$, where $f(\phi)$ is determined by $\alpha(G)$. Working with $\phi$ instead of $G$ allows us to use the power of expander graphs and replace $\phi$ by a new equivalent CNF formula $\phi'$ where each variable appears only a constant number of times. Finally, $\phi'$ is replaced by a bounded-degree graph $G'$ such that $\alpha(G')$ is determined by $\alpha(G)$.
The first and last reductions are based on standard ideas, whereas the second one is inspired by \cite{papadimitriou1991optimization, lecture} and is the main ingredient that allows us eventually to convert our graph to a bounded-degree graph. 


\subsubsection*{From $G$ to $\phi$}

Given a graph $G=(V,E)$, we construct $\phi$ as follows. For each vertex $v \in V$, we have a variable $x_v$, and a clause $x_v$. For each edge $e=\{u,v\} \in E$, we have the clause $(\neg x_u \vee \neg x_v)$. Intuitively, this guarantees that we do not take both $u$ and $v$ to the independent set. The formula $\phi$ is the conjunction of all the clauses. We next show the following.

\begin{claim} \label{G_phi}
$f(\phi)=\alpha(G)+|E|$.
\end{claim}

\begin{proof}
We first show that $f(\phi) \geq \alpha(G)+|E|$. Let $U$ be an independent set of size $\alpha(G)$. It is easy to see that the assignment that gives $T$ to all the variables $\{x_v|v \in U\}$ satisfies $\alpha(G)+|E|$ clauses, which shows $f(\phi) \geq \alpha(G)+|E|$.

We next prove that $f(\phi) \leq \alpha(G)+|E|$.
First, there is a maximal assignment $\pi$ that satisfies all the edge clauses. This holds since we can convert any maximal assignment to a maximal assignment that satisfies all the edge clauses as follows. While there is an edge clause that is not satisfied $C=(\neg x_u \vee \neg x_v)$ we can set the variable $x_u$ to $F$. Now $C$ and all other edge clauses that have $x_u$ are satisfied, but the variable clause $x_u$ is not satisfied. The number of satisfied clauses in this process can only increase. We can continue in the same manner and get an assignment with the same maximal number of satisfied clauses, where all the edge clauses are satisfied. We denote this assignment by $\pi$.

Let $U=\{v| \pi(x_v)=T\}$. Since $\pi$ satisfies all the edge clauses, for each edge $\{u,v\} \in E$, at least one of $u$ and $v$ is not in $U$, which shows that $U$ is an independent set. The number of satisfied clauses is exactly $f(\phi)=|U|+|E|$.
Since $|U| \leq \alpha (G)$, we get $f(\phi) \leq \alpha(G) + |E|$.
\end{proof}

\subsubsection*{Expander graphs}

In order to convert $\phi$ to $\phi'$ we use the following graphs from \cite{papadimitriou1991optimization}.

\begin{claim} \label{expander_graph}
For every $d$, there is a graph $G_d=(V_d,E_d)$ with $\Theta(d)$ vertices,  maximum degree 4 and diameter $O(\log{d})$, with the following property. $G_d$ has a set $D$ of $d$ \emph{distinguished} vertices of degree 2, and for every cut $(S, \overline{S})$ in $G_d$, the number of edges crossing the cut is at least $\min{\{|D \cap S|, |D \cap \overline{S}|\}}$.
\end{claim}

\begin{proof}
We follow the construction from \cite{papadimitriou1991optimization}, and show that it has a small diameter. A graph $G$ is a $c$-expander, for a constant $c$, if for every $|S| \leq \frac{n}{2}$ in $G$, the set $S$ has at least $c|S|$ neighbours outside $S$.
In \cite{ajtai1994recursive}, there is a construction that for every $n$ produces an expander of size $n$ with maximum degree 3 and a constant expansion rate $c$.

The graph $G_d$ is constructed as follows. It has a set $D$ of $d$ distinguished vertices. For each $v \in D$ we construct a full binary tree of constant size that has at least $\frac{1}{c}$ leaves, $v$ is the root of the tree. Next, all the leaves of all the binary trees are connected by the expander graph from \cite{ajtai1994recursive}. $G_d$ clearly has $\Theta(d)$ vertices and maximum degree 4, and all the vertices in $D$ have degree 2. We next bound its diameter.

We look at the expander graph $G_E$ induced on the leaves of the binary trees. Let $n = \Theta(d)$ be the number of vertices in $G_E$. From the expander properties, for each vertex $v \in G_E$, the neighbourhood of radius $i$ of $v$ in $G_E$ is either of size greater than $\frac{n}{2}$ or satisfies $N^i(v) \geq (1+c)N^{i-1}(v)$. This means that for a certain $i=O(\log{d})$, $N^i(v) > \frac{n}{2}$. Since this holds for any vertex in $G_E$, every two leaves of the binary trees in $G_d$ have a vertex in the intersection of their $O(\log{d})$-neighbourhoods, which gives a path between them of length $O(\log{d})$. Since the diameter of the binary trees in constant, the whole graph $G_d$ has diameter $O(\log{d})$.

We next show that $G_d$ has the last property. Let $(S, \overline{S})$ be a cut in $G_d$, and let $d_1,d_2$ be the number of binary trees that are contained entirely in $S$ and $\overline{S}$, respectively. Let $d_3 = d - d_1 - d_2$ be the rest of the trees, each of them has at least one edge in the cut. Let $n_1,n_2$ be the number of leaves in $S$ and $\overline{S}$, respectively, and assume without loss of generality that $n_1 \leq n_2$. From the expander properties, there are at least $c \cdot n_1$ edges in the cut. Now, each of the $d_1$ trees in $S$ has at least $\frac{1}{c}$ leaves, which gives $n_1 \geq \frac{1}{c} \cdot d_1$, which implies that the number of edges in the cut is at least $d_1 + d_3 \geq \min{\{|D \cap S|, |D \cap \overline{S}|\}}$, as needed.
\end{proof}

\subsubsection*{From $\phi$ to $\phi'$}

We next explain how to convert $\phi$ to a new formula $\phi'$ where each variable appears a constant number of times and $f(\phi')$ is determined by $f(\phi)$. This is inspired by \cite{papadimitriou1991optimization, lecture}.

For each variable $x_v$ in $\phi$, let $d_v$ be the number of appearances of $x_v$ in $\phi$ (from the construction, $d_v = deg_G(v) + 1$). In $\phi'$, for each variable $x_v \in \phi$, we have $d_v$ different variables $D =\{x_v^i| 1 \leq i \leq d_v\}$. Each appearance of $x_v$ in $\phi$ is replaced by appearance of one of the new variables in $\phi'$, such that each one of the new variables appears in exactly one clause. For example, the clause $(\neg x_u \vee \neg x_v) \in \phi$ is replaced by a corresponding clause $(\neg x^i_u \vee \neg x^j_v) \in \phi'$ for some $i,j$. In addition, to guarantee that all the variables $\{x_v^i| 1 \leq i \leq d_v\}$ have the same value, we add new variables and clauses, as follows.

Let $G_{d_v}$ be the graph from Claim \ref{expander_graph} with parameter $d_v$. We identify the variables $\{x_v^i| 1 \leq i \leq d_v\}$ with the distinguished vertices in $G_{d_v}$, and we add additional $\Theta(d_v)$ variables of the form $\{x_v^i| d_v < i \leq \Theta(d_v)\}$, each of them is identified with one non-distinguished vertex in $G_{d_v}$.
For each edge $e=\{x_v^i,x_v^j\}$ in $G_{d_v}$, we add to $\phi'$ the clauses $(\neg x_v^i \vee x_v^j)$ and $(\neg x_v^j \vee x_v^i)$. Note that these clauses are equivalent to the conditions $x_v^i \rightarrow x_v^j$ and $x_v^j \rightarrow x_v^i$ or equivalently $x_v^i=x_v^j$.
We do the above for each vertex $v$. We call these new clauses added the \emph{expander} clauses, and the rest of the clauses the \emph{original} clauses.
Each variable appears only a constant number of times in $\phi'$. If $x_v^i \in D$ it appears once in an original clause $x_v$ appeared in. In addition, each vertex $x_v^i$ appears in two clauses for each edge adjacent to $x_v^i$ in $G_{d_v}$. Since $G_{d_v}$ has maximum degree 4, and the degree of a vertex in $D$ is 2, each variable appears in at most 8 clauses. In addition, since in the two clauses added for an edge in $G_{d_v}$, a variable appears once in a positive form and once in a negative form, each literal (a variable $x$ or its negation $\neg x$) appears at most 4 times in $\phi'$.
We next show the following.

\begin{claim} \label{expander}
There is always a maximal assignment for $\phi'$ where all the expander clauses are satisfied.
\end{claim}

\begin{proof}
Let $\pi$ be a maximal assignment for $\phi'$, we show how to change $\pi$ to a new maximal assignment where all the expander clauses are satisfied. While there is a vertex $v$ where not all the variables $X = \{x_v^i| 1 \leq i \leq \Theta(d_v)\}$ receive the same value in $\pi$, we do the following.

Let $D =  \{x_v^i| 1 \leq i \leq d_v\}$.
First, note that variables in $X \setminus D$ appear only in expander clauses, and that if all the variables in $X$ receive the same value, all the expander clauses are satisfied. Let $S= \{x_v^i | \pi(x_v^i)=T\}$, let $Y = D \cap S$, and let $Z = D \cap \overline{S}$. We define a new assignment $\pi'$ that is equal to $\pi$ on all variables not in $X$. We set $\pi'$ on variables of $X$ as follows. If $|Y| \geq |Z|$, we set $\pi'(x)=T$ for all $x \in X$, and otherwise we set $\pi'(x)=F$ for all $x \in X$. This guarantees that all the expander clauses of variables in $X$ are satisfied.

We next show that the number of clauses satisfied by $\pi'$ is at least as the number of clauses satisfied by $\pi$. For this, we look at all the edges in the cut $(S,\overline{S})$ in $G_{d_v}$.
Each of these edges is an edge of the form $(x_v^i,x_v^j)$ where exactly one of $\pi(x_v^i),\pi(x_v^j)$ is equal to $T$. This means that one of the clauses  $(\neg x_v^i \vee x_v^j)$ and $(\neg x_v^j \vee x_v^i)$ is not satisfied by $\pi$. In $\pi'$ all the expander clauses are satisfied, which means that for each edge in the cut $(S,\overline{S})$ there is one clause satisfied by $\pi'$ and not by $\pi$.
On the other hand, let $W$ be the smaller set between $Y$ and $Z$. Each of the $|W|$ vertices in $W$ appears in one original clause that may be satisfied by $\pi$ but not by $\pi'$. These are the only clauses in $\phi'$ that may be satisfied by $\pi$ and not by $\pi'$ since variables in $X \setminus D$ appear only in expander clauses, and variables not in $X$ have the same assignment in $\pi$ and $\pi'$.
From Claim \ref{expander_graph}, the number of edges in $G_{d_v}$ that cross the cut $(S,\overline{S})$ is at least $\min\{|Y|,|Z|\}$, which proves that the number of clauses satisfied by $\pi'$ can only increase with respect to $\pi$.

We continue in the same manner: move to the next variable $u$ such that not all the variables $\{x_u^i| 1 \leq i \leq \Theta(d_u)\}$ receive the same value in $\pi$, until we get an assignment $\pi'$ where all the expander clauses are satisfied. As explained, in the process the number of satisfied clauses can only increase, which proves the claim.
\end{proof}

Let $m_{exp}$ be the number of expander clauses in $\phi'$.
From Claim \ref{expander}, we get the following.

\begin{corollary} \label{phi'_phi}
$f(\phi')=f(\phi)+m_{exp}$.
\end{corollary}

\begin{proof}
Let $\pi$ be a maximal assignment for $\phi$. For all $v,i$ we define $\pi'(x^i_v)=\pi(x_v)$. This assignment satisfies all the expander clauses since $\pi'(x^i_v)=\pi'(x^j_v)$ for all $i,j$, and exactly $f(\phi)$ of the original clauses, proving $f(\phi') \geq f(\phi)+m_{exp}$.

On the other hand, from Claim \ref{expander}, there is a maximal assignment $\pi'$ that satisfies all the expander clauses. This means that for all $v$, $\pi'(x^i_v)=\pi'(x^j_v)$ for all $i,j$. Hence, this assignment corresponds to an assignment $\pi$ for $\phi$ where $\pi(x_v)=\pi'(x_v^1)$. The number of clauses satisfied by $\pi$ is at most $f(\phi)$, which gives that the number of original clauses satisfied in $\phi'$ by $\pi'$ is at most $f(\phi)$. Since $\pi'$ is a maximal assignment, we get $f(\phi') \leq f(\phi)+m_{exp}$, which completes the proof.
\end{proof}

\subsubsection*{From $\phi'$ to $G'$}

We next explain how to convert $\phi'$ to a bounded-degree graph $G'$ such that $\alpha(G')=f(\phi')$.

In $\phi'$ all the clauses have one or two literals, where a literal is a variable or its negation. For each clause $C$ of the form $\ell$, we add to $G'$ a new vertex $v_{\ell}^C$, and for each clause $C=(\ell_1 \vee \ell_2)$, we add to $G'$ two new vertices $v_{\ell_1}^C,v_{\ell_2}^C$ with an edge between them.
In addition, we add the following edges. For each variable $x$ we add the edges $\{v_x^C,v_{\neg x}^{C'}\}$ to $G'$ for all $C,C'$ where $x$ is a literal in $C$ and $\neg x$ is a literal in $C'$. This guarantees that at most of them is added to an independent set.

$G'$ has a bounded degree because each vertex $v_{\ell}^C$ is connected to at most one additional vertex from the clause $C$, and also to vertices of the form $v_{\neg {\ell}}^{C'}$. However, each literal appears in $\phi'$ at most 4 times, which means that the degree of $v_{\ell}^C$ is at most 5.

We next prove the following claim.

\begin{claim} \label{phi'_G'}
$\alpha(G') = f(\phi')$.
\end{claim}

\begin{proof}
Let $\pi'$ be a maximal assignment for $\phi'$, we build an independent set of size $f(\phi')$ as follows. For each satisfied clause $C$ in $\phi'$, we choose one of the satisfied literals in $C$, $\ell$, and add the corresponding vertex $v_{\ell}^C$ to the set $I$. From the construction, $|I|=f(\phi')$, we next show that $I$ is independent. In all the edges that correspond to clauses, we add at most one vertex to $I$ as needed. For an edge of the form $\{v_x^C,v_{\neg x}^{C'}\}$, we add at most one vertex to $I$, since exactly one of $x,\neg x$ is satisfied by $\pi'$. This shows that $\alpha(G') \geq f(\phi')$.

In the other direction, let $I$ be an independent set of size $\alpha(G')$ in $G'$. Since $G'$ has all the edges of the form $\{v_x^C,v_{\neg x}^{C'}\}$, there are no two vertices of the form $v_x^C,v_{\neg x}^{C'}$ in $I$. We construct an assignment $\pi'$ as follows. We set $\pi'(x)=T$ if there is a vertex of the form $v_x^C$ in $I$, and we set it $\pi'(x)=F$ otherwise. The number of satisfied clauses in $\pi'$ is at least $|I|$: let $v_{\ell}^C \in I$, then $\pi'(\ell)=T$ by the definition of $\pi'$ which means that $C$ is satisfied. Also, from each clause $C=(\ell_1 \vee \ell_2)$ there is at most one vertex in $I$ because we have the edge $\{v_{\ell_1}^C,v_{\ell_2}^C\}$ in $G'$. This means that for any vertex in $I$ we have a different satisfied clause. Also the number of satisfied clauses is at most $f(\phi')$, which shows $\alpha(G') = |I| \leq f(\phi')$. This completes the proof.
\end{proof}

\subsection{A lower bound for MaxIS}

We next explain how we use the reductions described in Section \ref{bounded_reductions} to get a lower bound for MaxIS in bounded degree graphs. We use the construction from \cite{DBLP:conf/wdag/Censor-HillelKP17}. In this construction, Alice and Bob get inputs $x,y$ of size $k^2$, and construct a graph $G_{x,y}$ with $n_G = \Theta(k)$ vertices, constant diameter and a cut of size $\Theta(\log{k})$ such that $G_{x,y}$ has a minimum vertex cover of a certain size $M$ if and only if $DISJ(x,y)=False$. Since the complement of a minimum vertex cover is a MaxIS it follows that $G_{x,y}$ has a MaxIS of size $Z = n_G -M$ if and only if $DISJ(x,y)=False$. In this construction the number of edges is $\Theta(n_G^2)$ and in particular the degrees are not bounded.

\subsubsection*{Applying the reductions}

In order to get a bounded degree graph construction we let Alice and Bob apply the reductions described in Section \ref{bounded_reductions} on the graph $G_{x,y}$, such that each of them applies the reduction on its part, and both of them apply it on parts related to the edges of the cut. We next describe this in more detail.

\paragraph{Building $\phi_{x,y}$:} When building the formula $\phi_{x,y}$ from the graph $G_{x,y}$ there is a variable for every vertex in $G_{x,y}$ and a clause for every vertex and every edge in $G_{x,y}$, each of them knows all the variables corresponding to vertices in its side and all the clauses that contain at least one vertex from its side. In particular, from the construction, for each edge in the cut there is exactly one clause in $\phi_{x,y}$ and both of them know about it.

\paragraph{Building $\phi'_{x,y}$:} Next, they build $\phi'_{x,y}$ from $\phi_{x,y}$. Now each variable $x_v$ is replaced by $\Theta(d_v)$ variables and there are new clauses between these new variables according to the graphs $G_{d_v}$. However, these new clauses are only of the form $(\neg x_v^i \vee x_v^j)$ for the same vertex $v$. This means that it is enough that only one of Alice and Bob would know about any new clause. In particular, if $v$ is in Alice's side, Alice knows $d_v$ since she knows all the clauses that contain $v$, and she can replace $v$ by $\Theta(d_v)$ new variables and add the new clauses between them, and similarly for Bob. For each original clause in $\phi_{x,y}$ there is one clause in $\phi'_{x,y}$ such that a clause of the form $(\neg x_u \vee \neg x_v) \in \phi$ is replaced by a corresponding clause $(\neg x^i_u \vee \neg x^j_v) \in \phi'$ for some $i,j$. If $u$ is in Alice's side, she knows $i$, and otherwise Bob knows it and similarly for $v$.

\paragraph{Building $G'_{x,y}$:} Finally, they build $G'_{x,y}$ from $\phi'_{x,y}$. In $G'_{x,y}$ for each clause $C$ of the form $\ell$, there is a new vertex $v_{\ell}^C$, and for each clause $C=(\ell_1 \vee \ell_2)$, there are two new vertices $v_{\ell_1}^C,v_{\ell_2}^C$ with an edge between them. In addition, for each variable $x$ the edges $\{v_x^C,v_{\neg x}^{C'}\}$ are added to $G'_{x,y}$. Each one of Alice and Bob adds the vertices and edges that correspond to variables in its side. In particular, the edges $\{v_x^C,v_{\neg x}^{C'}\}$ are contained entirely in one of the sides, depending if $x$ is a variable in Alice's or Bob's side. The only edges in the graph that are in the cut between Alice and Bob correspond to clauses that contain one variable from each of the sides. From the construction, the number of these edges is equal to the number of edges in the original cut between Alice and Bob in $G_{x,y}$. This follows since each one of the cut edges is replaced by one clause in $\phi_{x,y}$ that is replaced by one clause in $\phi'_{x,y}$, that is replaced by one edge in $G'_{x,y}$. Alice and Bob have all the information to build the part of $G'_{x,y}$ that contains vertices and edges in their side. Each cut edge in $G'_{x,y}$ corresponds to an original cut edge in $G_{x,y}$ and Alice and Bob know the correspondence between the edges. In addition, from the construction $G'_{x,y}$ has maximum degree 5.

\paragraph{The size of $G'_{x,y}$:} We next analyze the size of $G'_{x,y}$. Let $n_G$ be the number of vertices in $G_{x,y}$. The number $n_G$ depends only on the size of the inputs and not on the specific inputs. From the construction of $\phi_{x,y}$, the number of variables in $\phi_{x,y}$ is $n_G$. In $\phi'_{x,y}$ each variable is replaced by $\Theta(d_v)$ variables where $d_v = deg_G(v)+1$. In $G_{x,y}$ the degree of all the vertices is $\Theta(n_G)$ (this follows from the construction of \cite{DBLP:conf/wdag/Censor-HillelKP17}). Hence, the number of variables in $\phi'_{x,y}$ is $\Theta(n_G^2)$.
Each variable appears in a constant number of clauses, hence the number of clauses is $\Theta(n_G^2)$ as well. In $G'_{x,y}$, for each clause there are one or two new vertices, hence the number of vertices is $\Theta(n_G^2)$. From the construction of $G_{x,y}$, it holds that $n_G=\Theta(k)$, where the size of the inputs of Alice and Bob is $k^2$, and the size of the cut is $\Theta(\log{k})$. In $G'_{x,y}$ the size of the cut is still $\Theta(\log{k})$ as explained above, however the number of vertices is $\Theta(n_G^2)=\Theta(k^2)$.

\paragraph{The diameter of $G'_{x,y}$:} We next show that the diameter of $G'_{x,y}$ is logarithmic.

\begin{claim}
The diameter of $G'_{x,y}$ is logarithmic.
\end{claim}

\begin{proof}
Each vertex in $G'_{x,y}$ is of the form $v_{\ell}^C$ where $\ell$ is a literal that contains a variable of the form $x_u^i$ in $\phi'_{x,y}$. We say that $u$ is the corresponding vertex to $v_{\ell}^C$ in $G_{x,y}$. We show that for every edge $\{u,w\} \in G_{x,y}$, there is a path of logarithmic length between $u',w' \in G'_{x,y}$ for every pair of vertices $u',w'$ that $u,w$ are their corresponding vertices in $G_{x,y}$. Since the diameter of $G_{x,y}$ is constant and each vertex in $G_{x,y}$ has at least one corresponding vertex in $G'_{x,y}$, this implies that the diameter in $G'_{x,y}$ is logarithmic.

Let $\{u,w\} \in G_{x,y}$. In $\phi_{x,y}$, there is a clause $(\neg x_u \vee \neg x_w)$. In $\phi'_{x,y}$ there is a corresponding clause $(\neg x_u^i \vee \neg x_w^j)$ for some $i,j$. Since the graphs $G_{d_u},G_{d_w}$ have logarithmic diameter, for each variable $x_u^{i'}$ there is a path of logarithmic length in $G_{d_u}$ between $x_u^i$ and $x_u^{i'}$ and the same holds for $x_w^j, x_w^{j'}$. Since each edge in the graphs $G_{d_u},G_{d_w}$ corresponds to two clauses in $\phi'_{x,y}$, the path between $x_u^i$ and $x_u^{i'}$ in $G_{d_u}$ corresponds to a series of clauses in $\phi'_{x,y}$, where the first has the variable $x_u^i$, the last has the variable $x_u^{i'}$ and every two consecutive clauses share a variable. The same holds for the path between $x_w^j$ and $x_w^{j'}$.

Now, a pair of vertices $u',w' \in G'_{x,y}$ that correspond to $u,w \in G_{x,y}$ is of the form $u'=v_{\ell_u'}^{C_u'},w'=v_{\ell_w'}^{C_w'}$ such that $x_u^{i'}, x_w^{j'}$ are the variables in the literals $\ell_u', \ell_w'$ for some $i',j'$. For each clause in $\phi'_{x,y}$, there is an edge in $G'_{x,y}$. Hence, the path of clauses between $x_u^{i'}$ and $x_u^i$ in $\phi'_{x,y}$ corresponds to a series of (not necessarily adjacent) edges in $G'_{x,y}$, where the first edge has a vertex of the form $v_{\ell'}^{C'}$ where $\ell'$ contains $x_u^{i'}$, the last edge has a vertex of the form $v_{\ell}^{C}$ where $\ell$ contains $x_u^{i}$, and for every two consecutive edges, there is a variable $z$ that appears in both of the corresponding clauses. We can connect all these edges by a path in $G'_{x,y}$ with the same length up to a constant factor, as follows. Let $C,C'$ be two consecutive clauses, and let $z$ be the variable they share. The corresponding edges in $G'_{x,y}$ have vertices of the form $v_{\ell}^{C}, v_{\ell'}^{C'}$ where both $\ell,\ell'$ contain $z$. If $z$ appears once in a positive form and once in a negative form, there is an edge between $v_{\ell}^{C}, v_{\ell'}^{C'}$ from the construction. Otherwise, $\ell = \ell'$, and both the vertices are connected to a vertex of the form $v_{\neg \ell}^{C''}$ for some $C''$ that contains $\neg \ell$ (it follows from the construction that each variable in $\phi'_{x,y}$ appears at least once in a positive form and once in a negative form in the expander clauses). So we can connect the edges that correspond to $C,C'$ by adding one or two additional edges to the path.

To conclude, this gives a path of logarithmic length between two vertices that contain the variables $x_u^{i'}$ and $x_u^i$. Similarly, there is a path of logarithmic length between two vertices that contain the variables $x_w^{j'}$ and $x_w^j$, and there is an edge in $G'_{x,y}$ that corresponds to the clause $(\neg x_u^i \vee \neg x_w^j)$. As explained above, every two vertices that contain some variable $z$ are connected by a path of length at most 2 in $G'_{x,y}$ (depending if it appears in a positive or negative form in each of them), which implies that for all $u',w' \in G'_{x,y}$ that correspond to $u,w \in G_{x,y}$ there is a path of logarithmic length between them. This completes the proof.
\end{proof}

\subsubsection*{The lower bound}

Next we show that if $A$ is an algorithm for MaxIS in bounded degree graphs that takes $T(n)$ rounds for a graph with $n$ vertices, Alice and Bob can solve disjointness on inputs of size $k^2$ by exchanging $O(T(k^2) \cdot |CUT| \cdot \log{k})$ bits, where $|CUT|$ is the size of the cut in $G'_{x,y}$. Then, the lower bound for set disjointness implies a lower bound for MaxIS.
We assume that $T(n)$ is a monotonic function, and that if $n'=\Theta(n)$, then $T(n') = \Theta(T(n'))$.

\begin{claim} \label{simulation_bounded}
Let $A$ be an algorithm for computing a MaxIS or the size of a MaxIS in bounded degree graphs that takes $T(n)$ rounds, then Alice and Bob can solve disjointness on inputs of size $k^2$ by exchanging $O(T(k^2) \cdot |CUT| \cdot \log{k})$ bits.
\end{claim}

\begin{proof}
Alice and Bob build the graph $G'_{x,y}$ with $\Theta(k^2)$ vertices as described above.
From Claim \ref{G_phi}, Corollary \ref{phi'_phi} and Claim \ref{phi'_G'}, the size of MaxIS in $G'_{x,y}$ is exactly $\alpha(G_{x,y})+m_G+m_{exp}$, where $m_G$ is the number of edges in $G_{x,y}$ and $m_{exp}$ is the number of expander clauses in $\phi'_{x,y}$.
Now, Alice and Bob can compute $m_G$ and $m_{exp}$ as follows. $m_G$ is the number of edges in the graph $G_{x,y}$ that corresponds to the inputs of Alice and Bob, each of them knows the number of edges in $G_{x,y}$ in its side based on its input, hence by exchanging one message they can compute $m_G$. The number $m_{exp}$ is the number of expander clauses in the formula $\phi'_{x,y}$. Again, each one of Alice and Bob knows the number of expander clauses corresponding to variables in its side. Hence, by exchanging one message they can compute $m_{exp}$.
The graph $G_{x,y}$ is taken from \cite{DBLP:conf/wdag/Censor-HillelKP17}, and satisfies $\alpha(G_{x,y}) = Z$ if and only if $DISJ(x,y)=False$, for a certain value $Z$ that depends only on $k$. Specifically, $Z = n_G - 4(k-1) - 4\log{k}$, where $n_G$ is the number of vertices in $G_{x,y}$ that is fixed given $k$, and known to Alice and Bob.

Alice and Bob can simulate the algorithm $A$ on the graph $G'_{x,y}$ as follows. Alice simulates the vertices in her side and Bob simulates the vertices in his side, each message that is sent on an edge in one of the sides is simulated locally by the relevant player, and for each edge in the cut they exchange one message in each direction for each round of the simulation.
At the end, both of them know the output of the algorithm on their sides. If the output is the size of MaxIS they know $\alpha(G'_{x,y}).$ Otherwise, each of them knows the number of vertices taken to the MaxIS on its side. Then, by exchanging one message they compute $\alpha(G'_{x,y})$. Since $A$ takes $\Theta(T(k^2))$ rounds, computing $m_G,m_{exp}$ requires exchanging two messages, and the size of each message is $\Theta(\log(k^2))=\Theta(\log{k})$ bits, after exchanging $O(T(k^2) \cdot |CUT| \cdot \log{k})$ bits, Alice and Bob know the values $\alpha(G'_{x,y}),m_G,m_{exp}$. Since $\alpha(G'_{x,y})=\alpha(G_{x,y})+m_G+m_{exp}$, this allows them to compute also $\alpha(G_{x,y})$. However, $\alpha(G_{x,y}) = Z$ if and only if $DISJ(x,y)=False$, which means that after exchanging $O(T(k^2) \cdot |CUT| \cdot \log{k})$ bits Alice and Bob solved disjointness.
\end{proof}

\begin{theorem}
Any distributed algorithm for computing a MaxIS or the size of a MaxIS in bounded-degree graphs in the \cgst{} model requires $\Omega(\frac{n}{\log^2{n}})$ rounds, this holds even for graphs with maximum degree 5 and logarithmic diameter.
\end{theorem}

\begin{proof}
From Claim \ref{simulation_bounded}, applying a distributed algorithm for MaxIS on the graphs $G'_{x,y}$ allows Alice and Bob to solve disjointness on inputs of size $k^2$ by exchanging $O(T(k^2) \cdot |CUT| \cdot \log{k})$ bits. From the $\Omega(k^2)$ lower bound for disjointness, and since $|CUT| = \Theta(\log{k})$ we get $T(k^2) = \Omega(\frac{k^2}{\log{k} \log{k}})$ rounds. Writing $n=k^2$, we get $T(n) = \Omega(\frac{n}{\log^2{n}})$ rounds. The graphs $G'_{x,y}$ have logarithmic diameter and maximum degree 5 from the construction.
\end{proof}

\subsection{Implications for MVC, MDS and 2-spanner} \label{sec:implications}

We next use known reductions to show $\tilde{\Omega}(n)$ lower bounds for MVC, MDS and minimum 2-spanner in bounded-degree graphs.

\subsubsection*{MVC}
Since the complement of a MaxIS is an MVC, finding an optimal MVC is equivalent to finding a MaxIS. Hence, the same construction implies an $\tilde{\Omega}(n)$ lower bound for exact computation of MVC in bounded-degree graphs, summarized as follows.

\begin{theorem} \label{MVC_bounded}
Any distributed algorithm for computing an MVC or the size of an MVC in bounded-degree graphs in the \cgst{} model requires $\Omega(\frac{n}{\log^2{n}})$ rounds, this holds even for graphs with maximum degree 5 and logarithmic diameter.
\end{theorem}

\subsubsection*{MDS}

\begin{theorem}
Any distributed algorithm for computing an MDS or the size of an MDS in bounded-degree graphs with logarithmic diameter in the \cgst{} model requires $\Omega(\frac{n}{\log^2{n}})$ rounds.
\end{theorem}

\begin{proof}
We use a standard reduction from MVC to MDS (see, for example, \cite{papadimitriou1991optimization}).
Let $G_{VC}$ be an input graph to MVC, we convert it to a new input graph $G$ for MDS as follows. In $G$, there is a vertex $v$ for each vertex in $G_{VC}$, and a vertex $v_e$ for every edge in $G_{VC}$. $G$ has all the edges of $G_{VC}$. In addition, for each edge $e=\{u,v\} \in G_{VC}$ we have the edges $\{v_e,v\},\{v_e,u\}$ in $G$.
Note that if $G_{VC}$ has a bounded-degree, the number of edges in $G_{VC}$ is $O(n)$, which means that the number of vertices in $G$ is $O(n)$. In addition, $G$ has a bounded-degree since the degrees of original vertices only duplicate, and the new vertices have degree 2. The reduction also preserves the diameter up to an additive factor.

Now, the size of MDS in $G$ is the size of MVC in $G_{VC}$. First, any dominating set in $G$ can be converted to one without vertices of the form $v_e$ for $e=\{u,v\} \in G_{VC}$, since both $u$ and $v$ cover $\{v,u,v_e\}$ and possibly additional vertices, so we can replace $v_e$ by one of $v$ or $u$. Next, each dominating set in $G$ that contains only original vertices is a vertex cover in $G_{VC}$ because all the vertices $v_e$ are covered in $G$ which implies that all the edges are covered in $G_{VC}$. Finally, any vertex cover in $G_{VC}$ is a set of vertices that covers all the edges in $G_{VC}$ and hence all the edge-vertices $v_e$ in $G$. However, if $v_e$ is covered at least one of its endpoints is in the vertex cover, and hence the other endpoint is also covered by this endpoint, which means that this is a dominating set in $G$.

The reduction can be simulated locally in the \cgst{} model, by having one of the endpoints of $e \in G_{VC}$ simulate the vertex $v_e \in G$. Since the degrees of new vertices are 2 and the original degree only increased by a factor of 2, any round of an algorithm in $G$ can be simulated in a constant number of rounds in $G_{VC}$. Hence, a distributed $T$-round algorithm for MDS in bounded-degree graphs gives an $O(T)$-round algorithm for MVC in bounded-degree graphs. Hence, Theorem \ref{MVC_bounded} implies an $\Omega(\frac{n}{\log^2{n}})$ lower bound for MDS in bounded-degree graphs.
\end{proof}

\subsubsection*{Minimum 2-spanner}

In the minimum 2-spanner problem, the input is a weighted graph $G$ and the goal is to find the minimum cost subgraph $H$ of $G$ such that for edge $e=\{u,v\} \in G$ there is a path of length at most 2 in $H$ between $u$ and $v$. In \cite{DBLP:conf/podc/Censor-HillelD18} there is a distributed reduction from unweighted MVC to weighted 2-spanner in the \cgst{} model. This reduction preserves the number of vertices, the degrees in the graph and the diameter up to a constant factor. Hence, Theorem \ref{MVC_bounded} implies the following.

\begin{theorem}
Any distributed algorithm for computing a minimum 2-spanner or the cost of the minimum 2-spanner in bounded-degree graphs with logarithmic diameter in the \cgst{} model requires $\Omega(\frac{n}{\log^2{n}})$ rounds.
\end{theorem}
\section{Hardness of approximation}
\label{section:approx(b)}

In this section, we show hardness of approximation results for non-trivial approximation ratios of fundamental problems. 
The main intuition behind our constructions is that we \emph{enhance the bit-gadget}.
The bit-gadget is a key component in obtaining our near-quadratic lower bounds for \emph{exact} computation. However, in all these constructions there is only a slight difference in the size of the optimal solutions according to the disjointness of the inputs. One of the reasons for this is that our row vertices have a vertex for any possible binary representation. Hence, in particular, we have vertices with almost identical binary representations, which limits the potential lower bounds we can show. To overcome this, one approach could be to leave only a subset of the row vertices that satisfy additional suitable properties, for example, a set of vertices in which any two have very different binary representations. Both of our constructions, one for MaxIS and one for $k$-MDS and additional problems, rely on and develop this idea.
\newline

We obtain lower bounds for computing a maximum independent set (MaxIS), for 2-minimum dominating set (2-MDS) and several related problems. In the MaxIS problem, the goal is to find an independent set of maximum weight or cardinality.
In the $k$-MDS (minimum $k$-dominating set) problem, the goal is to find a minimum weight set of vertices such that all the vertices in the graph have a vertex within the set at distance at most $k$.

In order to work with the same infrastructure as Theorem~\ref{generallowerboundtheorem}, we need to slightly extend the definition of a family of lower bound graphs. Instead of considering a Yes/No predicate $P$ when a graph parameter is considered (the max/min size of a certain set), we extend $P$ such that not only its No instances do not have the required size, but they have a bound on their size that separates them from the Yes instances. As a concrete example, for MaxIS we show a family of graphs which either have a maximum independent set of size $8\ell+4t$ for some appropriate values of $\ell, t$ that will be fully explained later, or have any independent set being of size at most $7\ell+4t$. Thus, in order to decide whether $P$ holds, an algorithm does not have to give an exact answer for the size of the maximum independent set, but rather an approximation suffices. This allows us to obtain lower bounds for approximation algorithms.

\subsection{MaxIS}\label{maxis lower bound}
The work of~\cite{DBLP:conf/wdag/Censor-HillelKP17} constructs a family of lower bound graphs for the problem of finding an exact solution for MaxIS, which obtains a lower bound of $\Omega(n^2/\log ^2 n)$. However, because the values of the size of the independent set in Yes/No instances are very close to each other, the construction in~\cite{DBLP:conf/wdag/Censor-HillelKP17} cannot show hardness of approximation, apart from very small additive approximation ratios.

Our main results in this section are the following.
\begin{theorem}\label{theorem:MaxIS}
Given $\epsilon >0$, any distributed algorithm in the \cgst{} model for computing a $(7/8 +\epsilon)$-approximation to the maximum independent set requires $\Omega(n^2/\log^7 n)$ rounds of communication.
\end{theorem}

We first prove that it is hard to find a $(7/8 +\epsilon)$-approximation for the weighted case, and then explain how this extends to the unweighted case as well.
We then further show that even if one is happy with a $(5/6 +\epsilon)$-approximation, at least a linear number of rounds must be spent in order to find it.

\begin{theorem}\label{theorem:linearMaxIS}
Given $\epsilon >0$, any distributed algorithm in the \cgst{} model for computing a $(5/6 +\epsilon)$-approximation to the maximum independent set requires $\Omega(n/\log^6 n)$ rounds of communication.
\end{theorem}

A key component of our lower bounds for MaxIS is a way to employ error correcting codes in our graph constructions, because of reasons which we later give some intuition for. In what follows, we provide the required background. Given some finite field $\mathbb{F}$ of size $q$, denote by $\mathbb{F} ^N$ the linear space of $N$-tuples over $F$. A \emph{linear code} $\C$ of \emph{length} $N$ and \emph{dimension} $\kappa$ is a linear subspace of dimension $\kappa$ of $\mathbb{F} ^N$. The Hamming distance  $d (v_1,v_2)$ between two elements $v_1,v_2\in \mathbb{F} ^N$ is the number of coordinates in which they differ. The \emph{distance} of a linear code $\C$ is $d=\min_{c_1,c_2 \in \C } d (c_1,c_2)$. Thus, a linear code $\C$ has parameters $(N,\kappa,d,q)$. We will make use of Reed-Solomon codes (see, e.g.,~\cite[Proposition 4.2]{DBLP:books/daglib/0015526}), which are codes with parameters $(N,\kappa,N-\kappa+1,q)$, where $q> N$. The value of $q$ is any prime power that is larger $N$.

To prove Theorem~\ref{theorem:MaxIS}, we first prove it for the weighted case.
\begin{theorem}\label{theorem:wMaxIS}
Given $\epsilon >0$, any distributed algorithm in the \cgst{} model for computing a $(7/8 +\epsilon)$-approximation to the maximum weight independent set requires $\Omega(n^2/\log^5 n)$ rounds of communication.
\end{theorem}

\begin{figure}[h]
\label{fig:maxis}
    \centering
    \includegraphics[width=170mm]{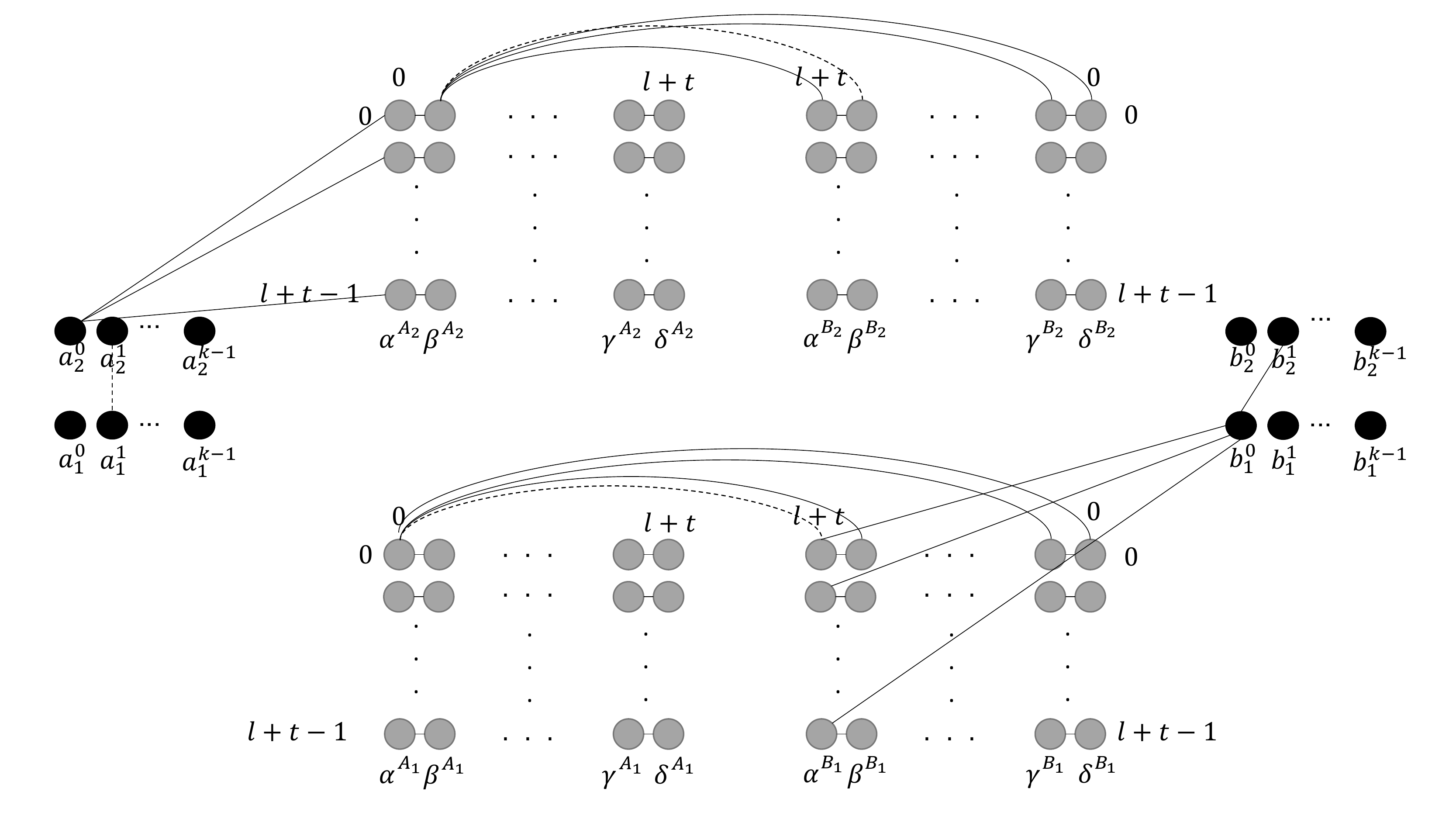}
    \caption{The family of lower bound graphs for the $(7/8+\epsilon)$-approximation of the weighted maximum independent set, with most of the edges omitted for clarity. Note that  $a_2 ^{0}$ is  connected to all vertices associated with $A_2$  except for all the  vertices in $\alpha _0 ^{A_2}$,  Similar examples are given for Bob's side. There are also examples of input edges corresponding to the case where $x_{(1,1)} =1, y_{(0,1)} =0$.}
    \label{fig:codes}
\end{figure}

\textbf{The fixed graph construction:} See Figure~\ref{fig:maxis} for an illustration of the base graph $G$. We define 4 sets of vertices of size $k$:
$A_1=\{a_1 ^i ~|~ 0\leq i\leq k-1\},A_2=\{a_2 ^i ~|~ 0\leq i\leq k-1\},B_1=\{b_1 ^i ~|~ 0\leq i\leq k-1\}$ and $B_2=\{b_2 ^i ~|~ 0\leq i\leq k-1\}$, which we call \emph{row vertices}. Each $S\in \{A_1,A_2,B_1,B_2\}$ is connected as a clique. Denote $\ell =c\log^2 k$, where the choice of the constant $c$ will be explained shortly, and assign all row vertices, i.e., all vertices in $A_1,A_2,B_1,B_2$, with weight $\ell$.

Consider a linear code $\C$ with parameters $(\ell+t,t,\ell+1,q)$ where $t=\log k$. It holds that $q> \ell+t$, and we can choose the constant $c$ so that equality holds, i.e., $q=\ell+t+1$. In words, $\ell+t$ is the length of the code, $t$ is its dimension, and $\ell+1$ is its distance. We leverage this large distance by representing each row vertex with a different codeword of $\C$ using a subgraph which we call a \emph{code gadget}, as follows.

For each element $\alpha \in \mathbb{F}_q$ and each $S\in \{A_1, A_2,B_1,B_2\}$ we add a \emph{column} of $\ell+t$ vertices (these are $4q(\ell+t)$ vertices in total). We call these vertices \emph{code gadget vertices}. All code gadget vertices are assigned a  weight of $1$. For each $0 \leq j \leq \ell+t-1$, we denote by $row(j,S)$ the set of vertices $\{\alpha_j^S ~|~ \alpha \in \mathbb{F}_q\}$. For each $S\in \{A_1,A_2,B_1,B_2\}$ and each $0 \leq j \leq \ell+t-1$, we connect all vertices of $row(j,S)$ as a clique. For each $z\in \{1,2\}$ and each $0 \leq j \leq \ell+t-1$, we connect $row(j,A_z)$ and $row(j,B_z)$ by a complete bipartite graph without a perfect matching. More specifically, for each $z\in \{1,2\}$, $0 \leq j \leq \ell+t-1$, and $\alpha\neq\alpha' \in \mathbb{F}_q$, we add an edge $(\alpha_j^{A_z}, \alpha_{j}^{B_z})$.

Now, we describe how to connect the row vertices to the code gadget vertices. Let $g: \{0,\dots, k-1\} \to \C $ be an injection. We abuse notation by referring to $g$ also as a function from any $S\in \{A_1, A_2,B_1,B_2\}$ to $\C$, in the natural way.
For an element $c \in \C$, and $0\leq j \leq \ell+t-1$ we denote by $c_j$ the element of $F$ which is in the $j$-th coordinate of $c$. For each $S\in \{A_1, A_2,B_1,B_2\}$, $0 \leq i \leq k-1$, we define $code(S^i)=\{v\in \alpha^S ~|~ \alpha \in\mathbb{F}_q \} \setminus \{\alpha^S_j ~|~ g(S^i)_j = \alpha\}$. We connect $S^i$ to all of the code gadget vertices in $code(S^i)$.

\paragraph{Intuition:} The construction of~\cite{DBLP:conf/wdag/Censor-HillelKP17} does not handle approximations well, since the representations of two distinct row vertices may be too similar, e.g., when their binary representation differs in a single bit. Error correcting codes give us the ability to represent each row vertex with a codeword while making sure that the representations of each two distinct row vertices are different in at least $\ell$ edges.
Thus, any small change in the row vertices of a given independent set forces a significant change in the choice of its code-gadget vertices.

\begin{claim} \label{claim:ISfixed}
Let $I\subseteq V$ be an independent set in $G$. Then $I$ contains at most $4(\ell+t)$ code gadget vertices. Specifically, it contains at most 2 vertices from each code gadget row, i.e., for all $0\leq j\leq \ell+t-1$ and each $z \in \{1,2\}$, it holds that $| I\cap \{\alpha_j^S ~|~ \alpha\in \mathbb{F}_q,\ S\in \{A_z,B_z\}\}| \leq 2$.  Furthermore, if equality holds for some $0\leq j\leq \ell+t-1$ and $z \in\{1,2\}$ then there is (a single) $\alpha\in \mathbb{F}_q$ such that  $\alpha _{j} ^{A_z}$ and $\alpha _{j} ^{B_z}$ are the vertices in $I$.
\end{claim}

\begin{proof}
For all $0\leq j\leq \ell+t-1$ and each  $S\in \{A_1, A_2,B_1,B_2\}$, the vertices in $row(j,S)$ are connected by a clique, which implies that no two of them can be in $I$. This gives $| I\cap \{\alpha_j^S ~|~ \alpha\in \mathbb{F}_q, S\in \{A_z,B_z\}\}| \leq 2$. Further, recall that for each $z\in \{1,2\}$, $0 \leq j \leq \ell+t-1$, and $\alpha\neq\alpha' \in \mathbb{F}_q$, we add an edge $(\alpha_j^{A_z}, \alpha_{j}^{B_z})$. This means that if $\alpha_j^{A_z}, \beta_j^{B_z}$ are in $I$, then $\alpha=\beta$. That is, if $| I\cap \{\alpha_j^S ~|~ \alpha\in \mathbb{F}_q, S\in \{A_z,B_z\}\}| = 2$, then there is an element $\alpha\in \mathbb{F}_q$ so that $\alpha _{j} ^{A_z}$ and $\alpha _{j} ^{B_z}$ are the vertices in $I$.
\end{proof}

An immediate corollary of Claim~\ref{claim:ISfixed} is that any independent set in $G$ has a weight of at most $4(\ell+t+k)$ since we can take at most a single vertex from each set $A_1,A_2,B_1,B_2$, as each such set is connected as a clique.

\textbf{Constructing $G_{x,y}$ from $G$ given $x,y\in \{0,1\}^{k^2}$:} We index the strings $x,y$ by pairs of the form $(i,i')$, such that $0 \leq i,i'\leq k-1$. Now we augment $G$ in the following way: For all pairs $(i,i')$, we add the edge $(a_1^i, a_2^{i'})$ if and only if $x_{i,i'} =0$, and we add the edge $(b_1^i, b_2^{i'})$ if and only if $y_{i,i'} =0$.

We now prove the following lemma.
\begin{lemma}\label{lemma:MaxISgap}
For every $x,y\in \{0,1\}^{k^2}$, it holds that $G_{x,y}$ has an independent set of weight $8\ell+4t$ if and only if $\disj_{k^2} (x,y)=\false$. If $\disj_{k^2} (x,y)=\true$ then the weight of the maximum independent set in $G_{x,y}$ is $7\ell+4t$.
\end{lemma}

\begin{proof}
Suppose $\disj_{k^2} (x,y)=\false$ and let $(i,i')$ be the pair for which $x_{i,i'} =y_{i,i'}=1$. This means that  the edges $(a_1 ^i,a_2 ^{i'}),(b_1 ^i,b_2 ^{i'})$ do not exist in $G_{x,y}$. Define $$I=\{a_1 ^i,a_2^{i'},b_1^i,b_2 ^{i'}\}\cup code(a_1 ^i)\cup code(b_1 ^i)\cup code(a_2 ^{i'})\cup code(b_2 ^{i'}).$$ It is easy to verify that $I$ is an independent set and that its weight is $8\ell+4t$.

For the other direction, suppose $G_{x,y}$ has an independent set $I$ of weight $8\ell+4t$. By Claim~\ref{claim:ISfixed} (and its following corollary) it must contain 4 row vertices, one vertex from each of the sets $A_1,B_1,A_2,B_2$. Denote these vertices $a_1 ^i, b_1^{i'} ,a_2 ^{\tilde{i}},b_2 ^{\tilde{i}'}$.
We claim that it must hold that $i=i'$ and $\tilde{i}=\tilde{i}'$, which implies that $\disj_{k^2} (x,y)=\false$. We prove this by showing that otherwise, there do not exist an additional $4(\ell+t)$ vertices in $I$ from the code gadget. Assume that $i\neq i'$. Then there exists $0 \leq j \leq \ell+t-1$ such that $g(i)=\alpha$, $g(i')=\beta$, such that $\alpha\neq\beta$. By the construction, we have that the only vertex in $row(j,A_1)$ that is not connected to $a_1 ^i$ is $\alpha_j^{A_1}$, and the only vertex in $row(j,B_1)$ that is not connected to $b_1 ^i$ is $\beta_j^{B_1}$. But $\alpha_j^{A_1}$ and  $\beta_j^{B_1}$ cannot both be in $I$ because they are connected by an edge, since $\alpha\neq\beta$.  Hence, $I$ contains at most one vertex from $row(j,A_1)\cup row(j,B_1)$, so by Claim~\ref{claim:ISfixed} we have that $I$ cannot contain $4(\ell+t)$ vertices from the code gadget. Hence, it must be that $i=i'$, and a similar argument shows that $\tilde{i}=\tilde{i}'$.

Moreover, if $a_1 ^i, b_1^{i'}$ are in an independent set $I$ but it does not hold that $i=i'$, then not only there exists $0 \leq j \leq \ell+t-1$ such that there can be at most a single vertex from $row(j,A_1)\cup row(j,B_1)$ in $I$, but there are in fact at least $\ell$ such indexes $j_1,\dots, j_{\ell}$ such that for every $1\leq r \leq \ell$ there can be at most a single vertex from $row(j_r,A_1)\cup row(j_r,B_1)$ in $I$. This follows directly from the distance $\ell$ of the code $\C$. Therefore, if $a_1 ^i, b_1^{i'} ,a_2 ^{\tilde{i}},b_2 ^{\tilde{i}'}$ are in $I$ but it does not hold that $i=i'$ and $\tilde{i}=\tilde{i}'$, then there are at most $4(\ell+t)-\ell$ vertices from the code gadget that appear in $I$, implying that its weight is at most $7\ell+4t$. Alternatively, if $\disj_{k^2} (x,y)=\true$, it may be possible to take $4(\ell+t)$ vertices from the code gadget into $I$ while sacrificing one vertex from $A_1,B_1,A_2$ or $B_2$. This also gives that the weight of $I$ is at most $7\ell+4t$, which completes the proof.
\end{proof}

We can now wrap up the proof of the lower bound.
\begin{proofof}{Theorem \ref{theorem:wMaxIS}}
Partition the vertices into $V_A,V_B$, such that $V_A =\bigcup_{j=0} ^{\ell+t-1} \alpha _j^{A_1} \cup \bigcup_{j=0} ^{\ell +t-1} \alpha _j^{A_2} \cup A_1 \cup A_2$ and $V_B=V\setminus V_A$.
Let $P$ be the property of having an independent set of weight $8\ell +4t$. Lemma~\ref{lemma:MaxISgap} implies that $G_{x,y}$ is a family of lower bound graphs with respect to $\disj$ and $P$, which allows us to complete the proof of the lower bound, as follows.

It holds that  $E(V_A,V_B)=E_{cut}=O((\ell+t)^2 )$. By Theorem~\ref{generallowerboundtheorem}, any algorithm for deciding $P$ requires $\Omega(k^2/(\ell+t)^2\log n)$ rounds. Recalling that $\ell=O(\log^2 k), t=O(\log k)$, gives that $n=\Theta(k)$, which implies a lower bound of $\Omega(n^2/\log^5 n)$ rounds, for distinguishing between the case of a maximum independent set being of weight $8\ell+4t$ and that of a maximum independent set being of weight $7\ell+4t$. This implies a lower bound for an approximation ratio of $(7\ell+4t)/(8\ell+4t)$, which is $7/8+\epsilon$.
\end{proofof}

To prove the same lower bound for the unweighted case, we essentially use the same approach, but slightly modify it to handle the fact that we cannot use weights.

\begin{proofof}{Theorem \ref{theorem:MaxIS}}
We modify the family of lower bound graphs as follows. For all $0\leq j\leq \ell+t-1$ and each  $S\in \{A_1, A_2,B_1,B_2\}$, we replace each row vertex $S^i$ with an independent set of $\ell$ vertices, denoted by $batch(S^i)=\{S^i(\xi) ~|~ 0 \leq \xi\leq\ell-1\}$. Any edge touching $S^i$ is now replaced by $\ell$ edges, each one touching a different $S^i(\xi)$. Note that, if one imagines modifying the edges for each batch in any arbitrary order, gives that an edge whose both endpoints were row vertices, is now replaced by $\ell^2$ edges.

The key observation is that if $I$ is a maximum independent set, then for each $batch(S^i)$, either all of its vertices are in $I$ or none are. This is because all the vertices of a batch have exactly the same neighbors. The proof then follows the exact same arguments as those for the weighted case, where now instead of $4k$ row vertices there are $4k\ell$ batch vertices. However, since $\ell=O(\log^2 k)$, this does not change the bound too much. We now have $n=\Theta(k\ell)=\Theta(k\log k)$, which implies that $k=\Theta(n/\log n)$, giving a bound of $\Omega(n^2/\log^7 n)$ rounds for a $(7/8+\epsilon)$-approximation of unweighted MaxIS.
\end{proofof}

A fundamental question is now to map the trade-offs between approximation ratios for MaxIS and the number of rounds required for them in the \cgst{} model.

\begin{proofof}{Theorem~\ref{theorem:linearMaxIS}}
We modify the construction of the family of lower bound graph that we used for proving Theorem~\ref{theorem:MaxIS} as follows. For the fixed graph construction, we completely remove all vertices in $A_1,B_1$ and all the code gadget vertices that touch them. We replace these removed vertices with two batches of $\ell$ vertices, $batch(v_A),batch(v_B)$.

Given strings $x,y$ of length $k$, we construct $G_{x,y}$ by connecting an edge $(v_A,a_2^i)$ if and only if $x_i=0$, and connecting an edge $(v_B,b_2^i)$ if and only if $y_i=0$, for all $i\in \{0,...,k-1\}$.

It now holds that if $\disj_k(x,y)=\false$ then the size of a maximum independent set in $G_{x,y}$ is $6\ell+2t$, and if $\disj_k(x,y)=\true$ then the size of a maximum independent set in $G_{x,y}$ is $5\ell+2t$. The reason is that if $\disj_k(x,y)=\false$, then there exists $0\leq i\leq k-1$ such that the two batches of row vertices $batch(a_2^i), batch(b_2^i)$ are not connected to $batch(v_A), batch(v_B)$. This implies that we can take all four batches, as well as two code gadget vertices from each row $0\leq j\leq \ell+t-1$ into an independent set. Otherwise, if $\disj_k(x,y)=\true$, then for any $0\leq i\leq k-1$, at least one of the two batches of row vertices $batch(a_2^i), batch(b_2^i)$ is connected to $batch(v_A), batch(v_B)$, respectively, implying that we cannot all 4 batches into an independent set, for any such $i$, giving that any maximum independent set has size at most $5\ell+2t$.

To distinguish these cases, Theorem~\ref{generallowerboundtheorem} says that $\Omega(k/(\ell+t)^2 \log n)$ rounds are required. The number of vertices is $n=\Theta(k\ell)$ and so we obtain a lower bound of $\Omega(n/ \log^6 n)$ for obtaining an approximation ratio of $(5\ell+2t)/(6\ell+2t)$, which is $(5/6+\epsilon)$.
\end{proofof}


\subsection{2-MDS}

Our lower bound for approximating the weight of a minimum 2-dominating set is as follows.
\begin{theorem} \label{2mds_thm}
Let $0 < \epsilon < 1$ be a constant. There is a constant $\beta$ such that obtaining a $(\beta \epsilon \log{n})$-approximation for weighted 2-MDS requires $\Omega(n^{1-\epsilon}/\log{n})$ rounds in the \cgst{} model. In addition, there is a constant $\beta$ such that obtaining a $(\beta \log{\log{n}})$-approximation for weighted 2-MDS requires $\widetilde{\Omega}(n)$ rounds.
\end{theorem}

For parameters $T$ and $\ell$ that will be chosen later, we use a collection $C$ of sets $S_1,S_2,...,S_T$ of the universe $U=\{1,...,\ell\}$, with the following property. For any $r$ sets from $S_1,S_2,...,S_T, \overline{S_1}, \overline{S_2},...,\overline{S_T}$, such that for all index $1\leq i\leq T$, the set $S_i$ and its complement $\overline{S_i}$ are not included in the $r$ sets, there is at least one element in $U$ that is not covered, i.e., does not appear in any of the $r$ sets. We call this property of $C$ the \emph{$r$-covering property}. Collections with this property are used in \cite{lund1994hardness, nisan2002communication} to show the inaproximability of set cover in different models. We use different parameters $T,\ell$ to get different time-approximation tradeoffs. In general, construction of such sets exist where $T$ is exponential in $\ell$ and $r$ is logarithmic in $\ell$. We use the following Lemma from \cite{nisan2002communication}.

\begin{lemma}(\cite{nisan2002communication}) \label{cc_2mds}
For any given $r \leq \log{\ell} - O(\log{\log{\ell}})$, there exists a collection $C$ of size $T$ satisfying the $r$-covering property, such that $T = e^{{\ell}/{r2^r}}$.
\end{lemma}

Given a collection $C$ as described above, we present a construction of a family of lower bound graphs for 2-MDS, see Figure~\ref{2mds} for an illustration. We choose $r = c \log{\ell}$ for a constant $c<1$.
\paragraph{The fixed graph construction:} For every $1 \leq j \leq \ell$, we have two vertices $a_j$ and $b_j$, connected by an edge, both representing the element $j\in U$. For every $1 \leq i \leq T$ we have two vertices $S_i$ and $\overline{S_i}$ representing the corresponding sets. The vertex $S_i$ is connected to a vertex $a_j$ if and only if $j$ is in the set $S_i$. Similarly, the vertex $\overline{S_i}$ is connected to a vertex $b_j$ if and only if $j$ is in the set $\overline{S_i}$, or equivalently if and only if $j$ is not in the set $S_i$. In addition, we have the vertices $\{a,b,R\}$ where $a$ is connected to all the vertices $S_i$, $b$ is connected to all the vertices $\overline{S_i}$, and there are edges $\{R,a\},\{R,b\}$.
We define $V_A=\{a_j\}_{j=1}^{\ell}\cup\{S_i\}_{i=1}^{T}\cup\{a\}$ and $V_B=\{b_j\}_{j=1}^{\ell}\cup \{\overline{S_i}\}_{i=1}^{T}\cup\{b,R\}$. The vertices have weights as follows. Let $\alpha$ be an integer greater than $r=c\log{\ell}$. All the vertices $\{a_j\}_{j=1}^{\ell}, \{b_j\}_{j=1}^{\ell}$ and $a,b$ have weight $\alpha$. The vertex $R$ has weight $0$.

\paragraph{Constructing $G_{x,y}$ from $G$ given $x,y \in \{0,1\}^T$:}  The vertex $S_i$ has weight 1 if $x_i=1$, and it has weight $\alpha$ otherwise. Similarly, the vertex $\overline{S_i}$ has weight 1 if $y_i=1$, and it has weight $\alpha$ otherwise.

\setlength{\intextsep}{0pt}
\begin{figure}[h]
\centering
\setlength{\abovecaptionskip}{-2pt}
\setlength{\belowcaptionskip}{6pt}
\includegraphics[scale=0.5]{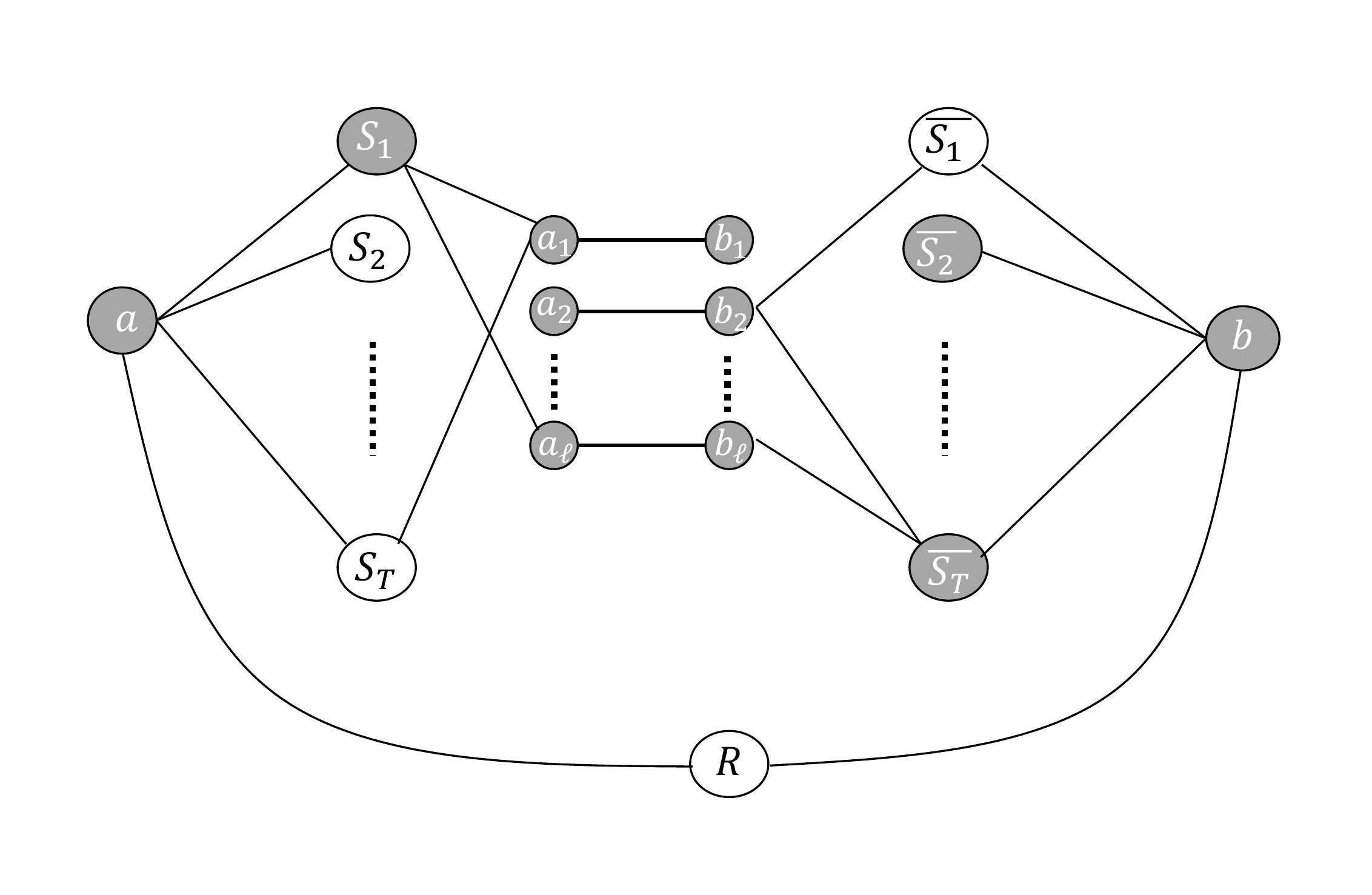}
 \caption{An illustration of the construction for the lower bound for 2-MDS, some of the edges are omitted for clarity. The grey vertices are vertices of weight $\alpha$. The vertices $S_i, \overline{S_i}$ may have weight 1 or $\alpha$ depending on the strings $x,y$.}
\label{2mds}
\end{figure}

The construction satisfies the following.

\begin{lemma} \label{2mds_disj}
If $\disj_T(x,y)=\false$ then there is a 2-MDS of weight 2, and otherwise any 2-MDS has weight greater than $c \log{\ell}$.
\end{lemma}

\begin{proof}
First, note that all the vertices except $\{a_j\}_{j=1}^{\ell}, \{b_j\}_{j=1}^{\ell}$ are at distance at most 2 from the vertex $R$. Since $R$ has weight 0, we can cover all these vertices at cost 0.

If $\disj_T(x,y)=\false$, there is an index $i$ where $x_i=y_i=1$. This means that both the vertices $S_i, \overline{S_i}$ have weight 1. Adding them to the cover results in a 2-MDS of weight 2: for every $j$, either $j \in S_i$ or $j \in \overline{S_i}$. In the former case, both the vertices $a_j,b_j$ are covered by $S_i$ since there is an edge $\{S_i,a_j\}$ and a path of length 2 $(S_i,a_j,b_j)$. In the latter case, $a_j,b_j$ are covered by $\overline{S_i}$.

Assume now that $\disj_T(x,y)=\true$. If a 2-MDS has a vertex of weight $\alpha$ then it has weight greater than $c \log{\ell}$. A 2-MDS that does not contain a vertex of weight $\alpha$ includes a collection $C' \subseteq C$ of sets of weight 1 from $S_1,...,S_T,\overline{S_1},...,\overline{S_T}$. Since $\disj_T(x,y)=\true$, $C'$ does not have a set $S_i$ and its complement $\overline{S_i}$ since at least one of them has weight $\alpha$. We next show that $C'$ is a set cover of $\{1,...,\ell\}$. Let $j \in U$, the vertex $a_j$ is covered by the 2-MDS, which implies that either there is a set $S_i \in C'$ where $j \in S_i$ and then $a_j$ is covered by the vertex $S_i$, or there is a set $\overline{S_i} \in C'$ where $j \in \overline{S_i}$, and then $a_j$ is covered by the vertex $\overline{S_i}$ of distance 2. These are the only possible options to cover the vertex $a_j$ by a vertex of weight 1, which shows that $C'$ is a set cover. Since $C$ satisfies the $r$-covering property, the size of $C'$ is greater than $r$, which completes the proof.
\end{proof}

We next choose different values for $T,\ell$ to show different time-approximation tradeoffs.

\begin{proofof}{Theorem~\ref{2mds_thm}}
We choose $T,\ell$ such that $T>\ell$. The number of vertices in the graph is $n=\Theta(T)$, the size of the inputs is $T$ and the size of $E_{cut}$ is $\Theta(\ell)$.
Lemma \ref{2mds_disj} implies that the graphs $\{G_{x,y}\}$ are a family of lower bound graphs for obtaining a $((c \log{\ell})/2)$-approximation for 2-MDS. Using Theorem \ref{generallowerboundtheorem},
we get a lower bound of $\Omega(n/(\ell \log{n}))$ rounds for a $((c \log{\ell})/2)$-approximation. We next choose different values for $\ell$. From Lemma \ref{cc_2mds}, we can choose $T$ to be at most exponential in $\ell$. However, to show a lower bound for an $O(\log{n})$-approximation, we start by choosing $\ell = T^{\epsilon}$ for some constant $\epsilon$. Now $\log{\ell}=\epsilon{\log{T}}=O(\epsilon \log {n})$, and we get a lower bound of $\Omega(n/(\ell \log{n})) = \Omega(n^{1-\epsilon}/ \log{n})$ for an $O(\epsilon \log{n})$-approximation. We can also get a nearly-linear lower bound, at a price of a smaller approximation ratio. From Lemma \ref{cc_2mds}, we can choose $T$ to be at most $e^{\ell/r2^r} = e^{\ell/c \log{\ell} \cdot \ell^c} = e^{\ell^{1-c}/c \log{\ell}}$. In particular, we can choose $T = e^{\ell^{1-2c}}$, which gives $\ell = (\ln{T})^{1/(1-2c)}$. If we choose $c$ to be a small constant, we get $\log{\ell} = \Theta(\log{\log{T}})$ and $\ell$ is polylogarithmic in $T$. This shows a lower bound of $\widetilde{\Omega}(n)$ for a $(\beta \log{\log{n}})$-approximation for some constant $\beta$.
\end{proofof}


\subsection{$k$-MDS}
A similar construction works also for $k$-MDS for $k>2$. The difference is that any edge of the form $\{S_i,a_j\}$ or $\{\overline{S_i},b_j\}$ is replaced by a path of length $k-1$ that has $(k-2)$ new internal vertices of weight $\alpha$. The vertices on paths corresponding to the original edges $\{S_i,a_j\}$ are in $V_A$, and the vertices on paths corresponding to the original edges $\{\overline{S_i},b_j\}$ are in $V_B$. The crucial observations are that all the vertices except $\{a_j\}_{j=1}^{\ell},\{b_j\}_{j=1}^{\ell}$ are at distance at most $k$ from the vertex $R$ of weight 0, and that the only vertices of weight $1$ at distance at most $k$ from a vertex $a_j$ or $b_j$ are vertices $S_i$ where $j \in S_i$ or vertices $\overline{S_i}$ where $j \in \overline{S_i}$. Using this, and following the proof of Lemma \ref{2mds_disj}, we get the following.

\begin{lemma} \label{kmds_disj}
If $\disj_T(x,y)=\false$ then there is a $k$-MDS of weight 2, and otherwise any $k$-MDS has weight greater than $c \log{\ell}$.
\end{lemma}

Since we add new vertices to the graph, the proof of the following is slightly different from the proof of Theorem~\ref{2mds_thm}.
\begin{theorem}
Let $0 < \epsilon < 1$ and $k$ be constants. There is a constant $\beta$ such that obtaining a $(\beta \epsilon \log{n})$-approximation for weighted $k$-MDS requires $\Omega(n^{1-\epsilon}/\log{n})$ rounds. In addition, there is a constant $\beta$ such that obtaining a $(\beta \log{\log{n}})$-approximation for weighted $k$-MDS requires $\widetilde{\Omega}(n)$ rounds.
\end{theorem}

\begin{proof}
We choose $T,\ell$ such that $T>\ell$ and $\ell = o(n)$. We add $\Theta(k)$ new vertices for each edge of the form $\{S_i,a_j\}$ or $\{\overline{S_i},b_j\}$, and since the number of such edges is at most $\Theta(T\ell)$, this gives that the number of vertices in the graph is $n=\Theta(kT\ell)$. As in the proof of Theorem \ref{2mds_thm}, the size of the inputs to $\disj$ is $T$, the size of the cut $E_{cut}$ is $\Theta(\ell)$, and hence we obtain a lower bound of $\Omega(T/(\ell \log{n}))$ rounds for a $((c \log{\ell})/2)$-approximation.

We next choose different values for $\ell$. We start by choosing $\ell = T^{\epsilon}$ for some constant $\epsilon$. Since $T=\Theta(n/k\ell)$ where $k$ is constant and $\ell = o(n)$, we get $\log{\ell}=\epsilon{\log{T}}=O(\epsilon \log {n})$. The lower bound is $\Omega(T/(\ell \log{n})) = \Omega(T^{1-\epsilon}/\log{n})$, where $n=\Theta(T^{1+\epsilon})$. This gives a lower bound of $\Omega(n^{(1-\epsilon)/(1+\epsilon)}/\log{n})=\Omega(n^{1-2\epsilon/(1+\epsilon)}/\log{n})=\Omega(n^{1-2\epsilon}/\log{n})$. Choosing $\epsilon'=2\epsilon$, we get a lower bound of $\Omega(n^{1-\epsilon'}/\log{n})$ for an $O(\epsilon' \log{n})$-approximation, as needed.

The second case follows the proof of Theorem \ref{2mds_thm}. We choose $T = e^{\ell^{1-2c}}$, which gives $\ell = (\ln{T})^{1/1-2c}$ where $c$ is a small constant. Since $\ell$ is polylogarithmic in $T$ and $k$ is constant, we get $n=\widetilde{\Theta}(T)$. As in the proof of Theorem \ref{2mds_thm}, we get a lower bound of $\widetilde{\Omega}(n)$ rounds for a $(\beta \log{\log{n}})$-approximation for some constant $\beta$.
\end{proof}


\subsection{Steiner tree problems}

A similar construction provides similar hardness results for the Steiner tree problem in directed graphs and for the Steiner tree problem with weights on vertices instead of edges.
\paragraph{Node weighted Steiner tree:} In the node weighted Steiner tree problem, the input is a graph $G=(V,E)$ with weights on the vertices and a set of terminals $S \subseteq V$, and the goal is to find a tree of minimum cost that spans all the terminals.

The exact same graph construction as the one for 2-MDS works also for the node weighted Steiner tree problem with the following changes. Now the weights of all the vertices in $\{a,b,R\}$ as well as in $A=\{a_j\}_{j=1}^{\ell},B=\{b_j\}_{j=1}^{\ell}$ are 0. The set of terminals is $A \cup B$. The weights of the vertices $S_i,\overline{S_i}$ is still $1$ or $\alpha$ depending on the strings $x,y$, exactly as in the 2-MDS construction. We show the following.

\begin{lemma} \label{steiner_disj}
If $\disj_T(x,y)=\false$ then there is a Steiner tree of weight 2, and otherwise any Steiner tree has weight greater than $c \log{\ell}$.
\end{lemma}

\begin{proof}
If $\disj_T(x,y)=\false$ then there is an index $i$ where both the vertices $S_i,\overline{S_i}$ have weight 1. We add them to the tree with all the edges adjacent to them, and we add all the edges $\{a_j,b_j\}$. In addition, we add to the tree the vertices $\{a,b,R\}$ with the edges $\{R,a\},\{R,b\},\{a,S_i\},\{b,\overline{S_i}\}$. The cost of the tree is 2, and it spans all the terminals, as follows. Let $1 \leq j \leq \ell$. If $j \in S_i$, we have the edges $\{S_i,a_j\},\{a_j,b_j\}$ in the tree, which shows that $a_j,b_j$ are spanned. Otherwise, $j \in \overline{S_i}$, and we have the edges $\{\overline{S_i},b_j\},\{b_j,a_j\}$ in the tree which shows that $a_j,b_j$ are spanned.

Assume now that $\disj_T(x,y)=\true$. Let $1 \leq j \leq \ell$. Let $C'$ be the collection of all the vertices from $S_1,...,S_T,\overline{S_1},...,\overline{S_T}$ that are in the Steiner tree.
To connect the vertices $a_j,b_j$ to the tree, there must be an edge in the tree of the form $\{S_i,a_j\}$ or $\{\overline{S_i},b_j\}$. This implies that $C'$, when we identify vertices with the corresponding sets, is a set cover of $\{1,...,\ell\}$. Now if $C'$ has a vertex of weight $\alpha$, we are done. Otherwise, since $\disj_T(x,y)=\true$, for each $i$, at least one of $S_i,\overline{S_i}$ has weight $\alpha$. Since $C$ has the $r$-covering property, we get that $C'$ is of size at least $r=c \log{\ell}$, which completes the proof.
\end{proof}

Using Lemma \ref{steiner_disj}, and following the proof of Theorem \ref{2mds_thm}, we get the following.

\begin{theorem} \label{steiner_thm}
Let $0 < \epsilon < 1$ be a constant. There is a constant $\beta$ such that obtaining a $(\beta \epsilon \log{n})$-approximation for the node weighted Steiner tree problem requires $\Omega(n^{1-\epsilon}/\log{n})$ rounds. In addition, there is a constant $\beta$ such that obtaining a $(\beta \log{\log{n}})$-approximation for the node weighted Steiner tree problem requires $\widetilde{\Omega}(n)$ rounds.
\end{theorem}
\paragraph{Directed Steiner tree:} In the directed Steiner tree problem, the input is a graph $G=(V,E)$ with weights on the edges, a root $R$, and a set of terminals $S \subseteq V$, and the goal is to find a directed tree with root $R$ of minimum cost that spans all the terminals.

We can adapt the 2-MDS construction for the directed Steiner tree problem, as follows. The vertex $R$ is the root, and $A \cup B$ are the terminals. Now all the edges are directed, as follows. We have the directed edges $(R,a),(R,b)$. For all $i$, we have the directed edges $(a,S_i),(b,\overline{S_i})$ and for all $j$ we have the following edges $(a_j,b_j),(b_j,a_j)$. If there is an edge between a vertex $S_i$ and a vertex $a_j$, the direction is $(S_i,a_j)$, and similarly for $\overline{S_i}$ and $b_j$. See Figure \ref{steiner_pic} for an illustration.
Now instead of weights on the vertices, there are weights on the edges, as follows. All the edges of the form $(a,S_i)$ or $(b,\overline{S_i})$ have weight 1, and the rest of the edges have weight 0. Given two strings $x=(x_1,...,x_t),y=(y_1,...,y_t)$, if $x_i = 1$ the graph has the edges $(S_i,a_j)$ for all $j \in S_i$, and if $x_i = 0$, none of these edges is in the graph. Similarly, the edges $(\overline{S_i},b_j)$ for $j \in \overline{S_i}$ appear in the graph if and only if $y_i = 1$.

\setlength{\intextsep}{0pt}
\begin{figure}[h]
\centering
\setlength{\abovecaptionskip}{-2pt}
\setlength{\belowcaptionskip}{6pt}
\includegraphics[scale=0.5]{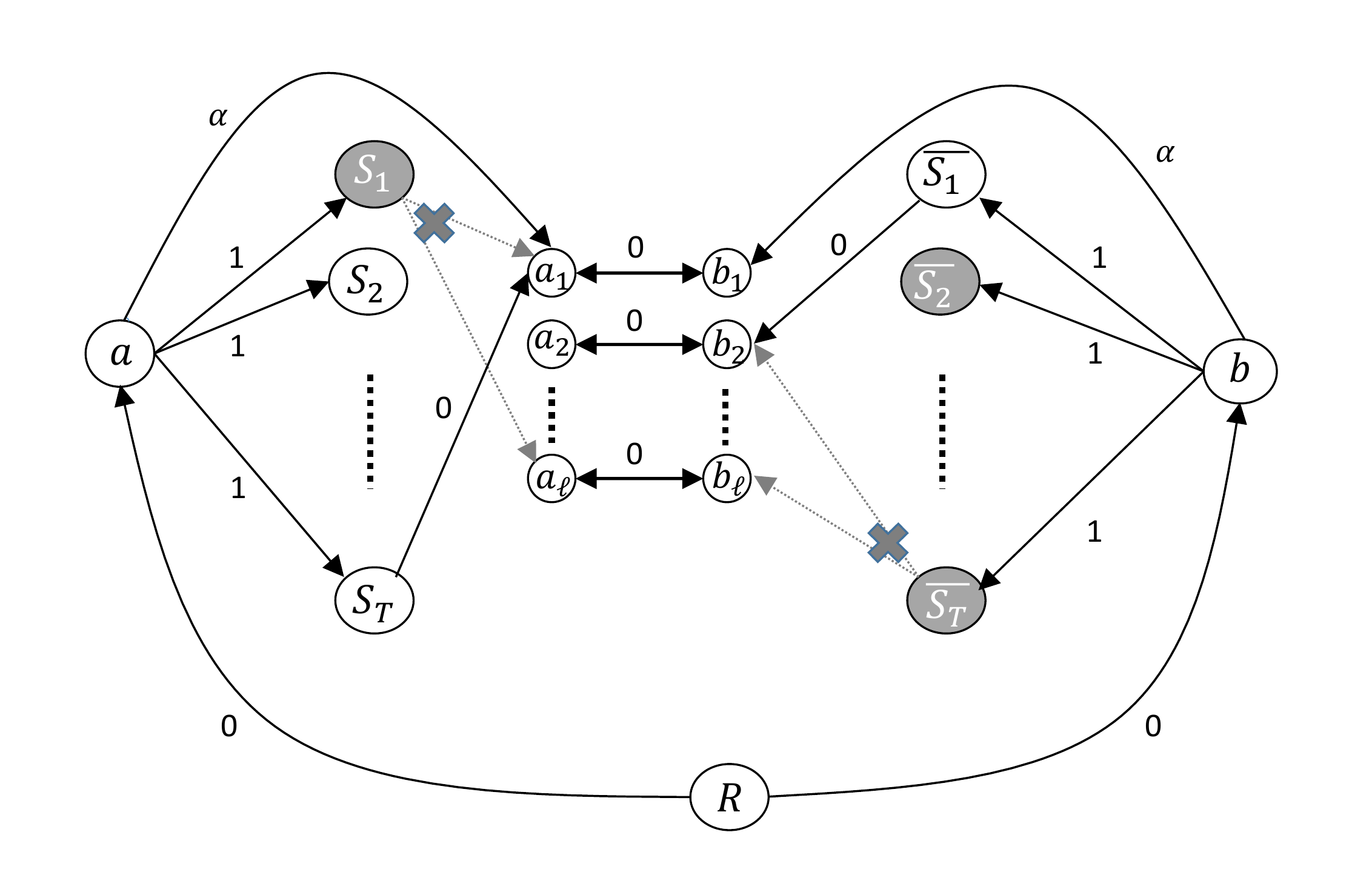}
 \caption{An illustration of the construction for the direct Steiner tree problem, where some of the edges are omitted for clarity. The edges from $a$ to $S_i$ have weight 1, and the edges from $a$ to $a_j$ have weight $\alpha$. Similarly, in the other side with respect to $b,\overline{S_i},b_j$. The rest of the edges have weight 0.
The edges $(S_i,a_j)$ appear in the graph depending on the inputs. For example, if $x_i=0$, the edges $(S_1,a_j)$ are not in the graph.}
\label{steiner_pic}
\end{figure}

To guarantee that there is always a solution (regardless of $x,y$), we add to the construction additional edges of weight $\alpha$, from the vertex $a$ to all the vertices $a_j \in A$ and from the vertex $b$ to all the vertices $b_j \in B$. The construction satisfies the following.

\begin{lemma} \label{dsteiner_disj}
If $\disj_T(x,y)=\false$ then there is a directed Steiner tree of weight 2, and otherwise any directed Steiner tree has weight greater than $c \log{\ell}$.
\end{lemma}

\begin{proof}
If $\disj_T(x,y)=\false$ then there is an index $i$ where the edges $(S_i,a_j)$ for $j \in S_i$ and the edges $(\overline{S_i},b_j)$ for $j \in \overline{S_i}$ are in the graph. Adding these edges to the tree, as well as the edges $(a_j,b_j)$ for $j \in S_i$, the edges $(b_j,a_j)$ for $j \in \overline{S_i}$, and the paths from $R$ to $S_i$ and $\overline{S_i}$ results in a directed Steiner tree of weight 2, since the only edges with weight greater than $0$ in the tree are the two edges $(a,S_i),(b,\overline{S_i})$.

On the other hand, to connect the vertices $a_j,b_j$ to the tree, the tree must include an edge incoming to $a_j$ or $b_j$. The only possible such edges with weight smaller than $\alpha$ are edges of the form $(S_i,a_j)$ for $j \in S_i$ where $x_i=1$ or $(\overline{S_i},b_j)$ for $j \in \overline{S_i}$ where $y_i=1$. In addition, if $S_i$ or $\overline{S_i}$ are in the tree, the edge $(a,S_i)$ or $(b,\overline{S_i})$ must be in the tree, and these edges have weight $1$. If $\disj_T(x,y)=\true$, since $C$ satisfies the $r$-covering property, following the proof of Lemma \ref{steiner_disj}, we get that we must add at least $r > c \log{\ell}$ such edges, which completes the proof.
\end{proof}

Using Lemma \ref{dsteiner_disj}, and following the proof of Theorem \ref{2mds_thm}, we get the following.

\begin{theorem} \label{dsteiner_thm}
Let $0 < \epsilon < 1$ be a constant. There is a constant $\beta$ such that obtaining a $(\beta \epsilon \log{n})$-approximation for the directed Steiner tree problem requires $\Omega(n^{1-\epsilon}/\log{n})$ rounds. In addition, there is a constant $\beta$ such that obtaining a $(\beta \log{\log{n}})$-approximation for the directed Steiner tree problem requires $\widetilde{\Omega}(n)$ rounds.
\end{theorem}

\subsection{Restricted hardness of approximation for MDS}

The MDS problem is a well-studied problem with several efficient $O(\log{\Delta})$-approximation algorithms in the \cgst{} model \cite{jia2002efficient, kuhn2005constant, DBLP:journals/jacm/KuhnMW16}. In addition, in the LOCAL model there are efficient $(1+\epsilon)$-approximations \cite{ghaffari2017complexity}. A natural question is whether is is possible to get efficient \cgst{} algorithms also for better than $O(\log{\Delta})$-approximations. If we can extend the construction for $2$-MDS, also for MDS, this would rule out such algorithms. However, to do so, we need to replace the sets of vertices $\{a_j\}_{j=1}^{\ell},\{b_j\}_{j=1}^{\ell}$ by one set, which would increase significantly the size of the cut. Moreover, as we explain in more detail in Section \ref{sec:limitations_MDS}, we cannot use the same lower bound technique to show such hardness result, since Alice and Bob can always get a 2-approximation to the MDS problem by solving the problem optimally on their sides and combining the solution. However, we can show such hardness result for a more restricted class of algorithms we refer to as \emph{local aggregate algorithms}. This includes algorithms that are based on simple operations such as computing a minimum, a sum, etc. We note that there are many \cgst{} algorithms that are based on such simple operations, and in particular there are efficient $O(\log{\Delta})$-approximations for MDS that are local aggregate algorithms.\footnote{See for example \cite{kuhn2005constant}, \cite{jia2002efficient} and Section 4.2 in \cite{DBLP:conf/podc/Censor-HillelD18}. The first algorithm presented in \cite{jia2002efficient} uses computation of the median, but it has variants that replace this with computation of the average or maximum (which fit our definition of local aggregate algorithms), with a slight increase in the approximation obtained or the time complexity \cite{jia2002efficient,lecture_mds}.} We show that obtaining a better approximation requires different algorithms.

Our construction is based on our 2-MDS construction, with the following changes. Instead of having the two vertices $a_j,b_j$ represent the element $j$, we have only one vertex named $j$ for each $1 \leq j \leq \ell$. The vertex $j$ is connected to all the vertices $S_i$ such that $j \in S_i$, and to all the vertices $\overline{S_i}$ such that $j \in \overline{S_i}$. The rest of the graph remains the same.
The vertices $\{S_i\}_{i=1}^{T}\cup\{a\}$ are in $V_A$ and the vertices $\{\overline{S_i}\}_{i=1}^{T}\cup\{b,R\}$ are in $V_B$, as before. The difference is that the new vertices $\{1,...,\ell\}$ are not in any of the sides and will be simulated together by Alice and Bob as we explain later. The vertices have weights, similarly to before. All the new vertices $\{1,...,\ell\}$ have weight $\alpha$ where $\alpha$ is an integer greater than $r=c\log{\ell}$. The vertices $\{R,a,b\}$ have weight 0, and the vertices $S_i, \overline{S_i}$ have weights according to the strings $x,y$, each of size $T$, denoted by $x=(x_1,...,x_T),y=(y_1,...,y_T)$, as before. The vertex $S_i$ has weight 1 if $x_i=1$, and it has weight $\alpha$ otherwise. Similarly, the vertex $\overline{S_i}$ has weight 1 if $y_i=1$, and it has weight $\alpha$ otherwise. See Figure \ref{mds_pic} for an illustration.

\setlength{\intextsep}{0pt}
\begin{figure}[h]
\centering
\setlength{\abovecaptionskip}{-2pt}
\setlength{\belowcaptionskip}{6pt}
\includegraphics[scale=0.5]{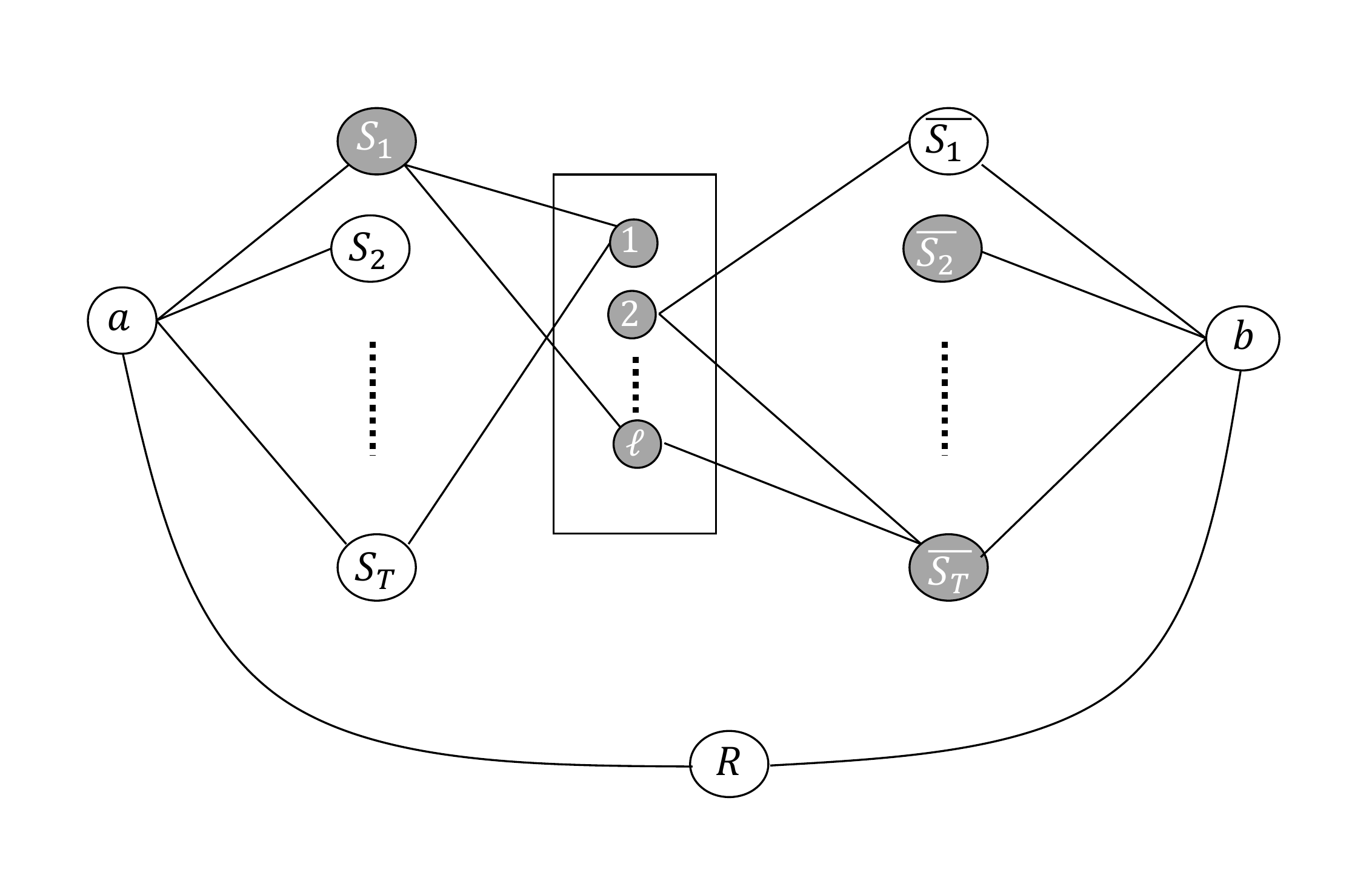}
 \caption{An illustration of the construction for the restricted hardness of approximating MDS. The new vertices $\{1,...,\ell\}$ replace the vertices $\{a_j\}_{j=1}^{\ell}, \{b_j\}_{j=1}^{\ell}$. The grey vertices have weight $\alpha$.}
\label{mds_pic}
\end{figure}

The construction satisfies the following.

\begin{lemma} \label{mds_hard_disj}
If $\disj_T(x,y)=\false$ then there is an MDS of weight 2, and otherwise any MDS has weight greater than $c \log{\ell}$.
\end{lemma}

The proof follows the proof of Lemma \ref{2mds_disj}, and uses the following. All the vertices except $\{1,...,\ell\}$ are covered by a vertex of weight 0, since they are neighbours of $a$ or $b$. In addition, the only vertices of weight smaller than $\alpha$ that cover a vertex $j \in \{1,...,\ell\}$, are vertices $S_i$ such that $j \in S_i$ and $x_i=1$, or vertices $\overline{S_i}$ such that $j \in \overline{S_i}$ and $y_i=1$.

\paragraph{Lower bounds for local aggregate algorithms:}
We next use Lemma \ref{mds_hard_disj} to show lower bounds for MDS with restrictions on the algorithm. We start by defining the notion of a \emph{local aggregate} algorithm, following the definition in \cite{bar2017distributed} with slight changes.
A function $f$ is called \emph{order-invariant} if $f(x_1,...,x_n)=f(x_{\pi(1)},...,x_{\pi(n)})$ for any permutation $\pi$. We use the notation $f(\{x_i\})$ for $f(x_1,...,x_n)$. We consider functions $f$ that receive at most $n$ inputs, each of size $O(\log{n})$, and their output is of size $O(\log{n})$. If $f$ has $k<n$ inputs, we write it as $f(x_1,...,x_k,\bot,...,\bot)$, where $\bot$ is an empty input. Simple examples of $f$ we use are summation of at most $n$ elements, minimum of at most $n$ elements, etc.

\begin{definition}
A function $f$ is an \emph{aggregate} function if it is order-invariant and there exists a function $\phi$ such that for each partition of the inputs $X$ into two disjoint sets $X_1,X_2$ it holds that $f(X)=\phi(f(X_1),f(X_2))$.
\end{definition}

In many simple examples, such as minimum or sum, $\phi=f$.
Now, we define the notion of a \emph{local aggregate} algorithm.
First, we assume that in each round each vertex $v$ has some initial input of $O(\log{n})$ bits. At the beginning this input contains the input of $v$ and its id, and later it depends on the input of $v$ in the previous round and an aggregate function of the messages it received in the previous round. We also assume that all the vertices have access to a shared random string of arbitrary length. Showing a lower bound in such model clearly gives also a lower bound also in the more severe model where each vertex has only a private random string. We say that an algorithm is a \emph{local aggregate} algorithm if in addition to the above, the message that vertex $v$ sends to vertex $u$ in round $i$ depends only on the input of $v$ for round $i$, the $ID$ of $u$, the shared randomness, and an aggregate function $f$ of the messages $v$ received in round $i-1$, where $f$ is a function with inputs and output of size $O(\log{n})$ bits.
Note that since the message from $v$ to $u$ depends on the $ID$ of $u$, it may be that $v$ sends different messages to different neighbours. We next show the following.

\begin{theorem} \label{mds_hard_thm}
Let $0 < \epsilon < 1$ be a constant. There is a constant $\beta$ such that any local aggregate $(\beta \epsilon \log{n})$-approximation algorithm for weighted MDS requires $\Omega(n^{1-\epsilon}/\log{n})$ rounds. In addition, there is a constant $\beta$ such that any local aggregate $(\beta \log{\log{n}})$-approximation for weighted MDS requires $\widetilde{\Omega}(n)$ rounds.
\end{theorem}

\begin{proof}
We prove a variant of Theorem~\ref{generallowerboundtheorem} for local aggregate algorithms. Alice and Bob need to simulate a distributed algorithm in the constructed graph, but now there are vertices that need to be simulated together by both of the players, as follows. From her input, Alice knows all the edges in the graph that are adjacent to the vertices $\{S_i\}_{i=1}^{T}\cup\{ a\}$, and Bob knows all the edges adjacent to the vertices $\{\overline{S_i}\}_{i=1}^{T}\cup\{b,R\}$. Both Alice and Bob know the $ID$s of the new vertices $\{1,...,\ell\}$, but each of them knows only about part of the edges adjacent to each of these vertices. However, for each vertex $j \in \{1,...,\ell\}$, all the edges adjacent to it are of the form $\{j,S_i\}$ or $\{j, \overline{S_i}\}$, and hence for each such edge $e$, one of Alice and Bob knows about $e$.

To simulate a local aggregate algorithm on this graph, Alice and Bob work as follows. Alice simulates the vertices in $V_A$, and Bob simulates the vertices in $V_B$. To simulate a vertex $j \in \{1,...,\ell\}$, notice that both Alice and Bob know the initial input and $ID$ of $j$ (in our case the input is the weight $\alpha$ of the vertex), and the shared randomness. They also know the input of $v$ for any later round as we show next. To decide which message $j$ sends to a vertex $v$ in $V_A$, Alice should know the $ID$ of $v$ and an aggregate function $f$ of the messages $j$ receives in round $i-1$. Since Alice and Bob together know all the edges adjacent to $j$, they can compute the output of $f$ as follows. Alice computes $f$ on the messages sent by the neighbours of $j$ in $V_A$, and Bob does the same for the neighbours in $V_B$. Now the output of $f$ is of $O(\log{n})$ bits, Alice and Bob exchange the outputs they received, and compute $\phi$ on these two outputs to know the final output of $f$. For example, if $f$ is the minimum function, each of them computes the minimum from the values it knows, and then they just exchange the values computed and take the minimum of them. Alice and Bob do this for each one of the $\ell$ vertices in $\{1,...,\ell\}$, which requires exchanging $\ell$ messages of $O(\log{n})$ bits per round. From the values computed, Alice knows all the messages sent from vertices in $\{1,...,\ell\}$ to vertices in $V_A$, and Bob knows the same with respect to $V_B$. They also know the input of each one of the vertices $\{1,...,\ell\}$ for the next round because it depends on the previous input and an aggregate function of the messages received in the current round.

After simulating the whole algorithm, Alice and Bob know which vertices in $V_A,V_B$, respectively, and in $\{1,...,\ell\}$ are added to the MDS. By exchanging one additional message they learn the size of the MDS and can deduce if $\disj_T(x,y)=\false$ according to Lemma \ref{mds_hard_disj}. Given a distributed algorithm that takes $t$ rounds, the cost of the whole simulation is $O(t \ell \log{n})$ bits, since Alice and Bob exchange $O(\ell \log{n})$ bits per round. The size of the inputs is $T=\Theta(n)$, hence using the lower bound for set disjointness (that applies also for algorithms that use shared randomness), we get a lower bound of $\Omega(T/(\ell \log{n}))$ rounds for approximating MDS within a $((c \log{\ell})/2)$ factor. These are exactly the values obtained in the proof of Theorem \ref{2mds_thm}, and we choose values for the parameters $T,\ell$ following the same analysis, which completes the proof.
\end{proof}

\section{Limitations of Theorem \ref{generallowerboundtheorem}} \label{section:limitations(b)}
In this section, we show that while Theorem \ref{generallowerboundtheorem} is powerful for many near-quadratic lower bounds for exact and approximate computations, for some problems, such as maximum matching and max-flow that are in P, it can be somewhat limited. To this end, we equip the previously used approach for showing limitation with new machinery.
\newline

For proving that the framework of Theorem~\ref{generallowerboundtheorem} cannot prove a lower bound greater than $\Omega(t(n))$ for a distributed algorithm for deciding some predicate $P$ and some function $t(n)$, we assume towards a contradiction that there exists a function $f$ and a family $\{G_{x,y}\}$ of lower bound graphs with respect to $f$ and $P$, and then we show that Alice and Bob can solve $P$ on this family of graphs while exchanging only few bits, whose number we denote by $t'$. Notice that in addition to depending on $n$, $t'$ may depend on the parts of the family $\{G_{x,y}\}$ which do not depend on the inputs $x,y$. Specifically, for our proofs we will have $t'$ depend on $|E_{cut}|$, that is, $t'=t'(n,E_{cut})$. Since the assumption is that the answer for $P$ on $G_{x,y}$ determines the output of $f$ on $(x,y)$, this implies that $CC(f)$ is at most $t'(n, E_{cut})$, which in turn gives that the lower bound obtained by Theorem~\ref{generallowerboundtheorem} is at most $t(n)=t'(n, E_{cut})/|E_{cut}|\log{n}$. Formally, for a function $f: \{0,1\}^K\times \{0,1\}^K \rightarrow \{\true,\false\}$, a predicate $P$, and a family $\{G_{x,y}\}$ of lower bound graphs with respect to $f$ and $P$, we denote by $CC^{\{G_{x,y}\}}(P)$ the number of bits that need to be exchanged between Alice and Bob for answering $P$ on graphs in $\{G_{x,y}\}$. The above discussion implies that $CC(f) \leq CC^{\{G_{x,y}\}}(P)$, which directly gives the following tool, which was implicitly used for showing a limitation of this framework for obtaining any lower bound for triangle detection~\cite{DruckerKO13}, any super-linear lower bound for weighted APSP~\cite{DBLP:conf/wdag/Censor-HillelKP17} (recently proven to have a linear solution~\cite{BernsteinN18}), and any lower bound larger than $\Omega(\sqrt{n})$ for detecting 4-cliques~\cite{DBLP:conf/wdag/CzumajK18}.

\begin{corollary}
\label{theorem:alice-bob-cut}
Fix a function $f: \{0,1\}^K\times \{0,1\}^K \rightarrow \{\true,\false\}$ and a predicate $P$. If $\{G_{x,y}\}$ is a family of lower bound graphs with respect to $f$ and $P$, then Theorem~\ref{generallowerboundtheorem} cannot give a lower bound for a distributed algorithm that is larger than $\Omega(CC^{\{G_{x,y}\}}(P)/|E_{cut}|\log{n})$.
\end{corollary}
Corollary~\ref{theorem:alice-bob-cut} also holds for randomized distributed algorithms, when considering $CC^{\{G_{x,y}\},R}(P)$ as the respective randomized complexity of a 2-party protocol.
In many cases we will have $CC^{\{G_{x,y}\}}(P)=t(n)|E_{cut}|\log{n}$, for which Theorem~\ref{theorem:alice-bob-cut} implies that Theorem~\ref{generallowerboundtheorem} cannot give a lower bound that is larger than $t(n)$.

\subsection{Limitations by determining the predicate}
\label{section:alice-bob}
\subsubsection{Bounded-degree graphs}
\label{section:alice-bob-bounded}

In Section \ref{sec:exact-bounded(b)}, we showed that solving \emph{unweighted} MaxIS, MVC and MDS \emph{exactly} in bounded-degree graphs requires $\tilde{\Omega}(n)$ rounds. A natural question is whether these problems are also difficult to approximate. We do not answer this question, but we prove that this framework is incapable of showing such lower bounds.
Specifically, we show that if $\{G_{x,y}\}$ are bounded degree graphs, then $CC^{\{G_{x,y}\}}(P)=O(|E_{cut}|\log(n)/\epsilon)$, for $P$ that implies a $(1+\epsilon)$-approximation for \emph{unweighted} MVC and MDS, or a $(1-\epsilon)$-approximation for MaxIS. This gives that for these problems, Corollary~\ref{theorem:alice-bob-cut} implies that Theorem~\ref{generallowerboundtheorem} cannot provide lower bounds that are larger than $\Omega(1/\epsilon)$.

\paragraph{MVC:}
\begin{claim}
\label{claim:limit-bounded-MVC}
Let $P$ be a predicate that implies a $(1+\epsilon)$-approximation for unweighted MVC. If $\{G_{x,y}\}$ is a family of lower bound graphs with bounded degree with respect to a function $f$ and the predicate $P$, then 
Theorem~\ref{generallowerboundtheorem} cannot give a lower bound for a distributed algorithm for deciding $P$ that is larger than $\Omega(1/\epsilon)$.
\end{claim}

\begin{proof}
Alice and Bob construct a graph from the family according to their inputs, and compute the number of edges in the graph, $m$, by exchanging $\log(n)$ bits, which is possible since each of them knows the number of edges touching $V_A,V_B$, respectively. Similarly, they learn the maximum degree $\Delta$.

There are now two cases. If $|E_{cut}|\leq\frac{\epsilon m}{2 \Delta}$, then each of Alice and Bob chooses an optimal vertex cover for the edges in $V_A,V_B$, respectively, and they add also all the vertices that touch the cut (at most $\frac{\epsilon m}{\Delta}$ vertices) to cover the edges of $E_{cut}$. This gives a vertex cover of size at most $OPT+\frac{\epsilon m}{\Delta}$ since the union of the disjoint optimal covers in $V_A,V_B$ is of size at most $OPT$. Now, since $OPT$ covers all the edges and each vertex in $OPT$ covers at most $\Delta$ edges, we have $m \leq \Delta OPT$. Hence, the size of the cover that Alice and Bob get is at most $OPT+\frac{\epsilon \Delta OPT}{\Delta}=(1+\epsilon)OPT$.

Otherwise, $|E_{cut}|\geq\frac{\epsilon m}{2 \Delta}$. In this case, Alice and Bob exchange information about all $m$ edges in the graph and solve the problem optimally. This requires a number of bits that is $m\log(n) \leq |E_{cut}|2\Delta\log(n)/\epsilon$. Since $\Delta$ is a constant by the assumption on $\{G_{x,y}\}$, we have that $CC^{\{G_{x,y}\}}(P)=O(|E_{cut}|\log(n)/\epsilon)$, which implies that Theorem~\ref{generallowerboundtheorem} cannot give a lower bound for a distributed algorithm for deciding $P$ that is larger than $\Omega(1/\epsilon)$.
\end{proof}

\paragraph{MDS:}

\begin{claim}
\label{claim:limit-bounded-MDS}
Let $P$ be a predicate that implies a $(1+\epsilon)$-approximation for unweighted MDS. If $\{G_{x,y}\}$ is a family of lower bound graphs with bounded degree with respect to a function $f$ and the predicate $P$, then Theorem~\ref{generallowerboundtheorem} cannot give a lower bound for a distributed algorithm for deciding $P$ that is larger than $\Omega(1/\epsilon)$.
\end{claim}

\begin{proof}
Similarly to the proof of Claim~\ref{claim:limit-bounded-MVC}, there are two cases. If $|E_{cut}|\leq\frac{\epsilon m}{2(\Delta+1)\Delta}$, Alice and Bob cover optimally the internal vertices in their sides (the vertices in $V_A,V_B$, respectively, that do not touch the cut) by adding vertices only from $V_A,V_B$ (they can use vertices touching the cut in the solution). In addition, they add to the solution all the vertices of the cut. The size of the solution is at most $OPT+\frac{\epsilon m}{(\Delta+1)\Delta}$ since the union of the disjoint optimal solutions they add is at most $OPT$. Since each vertex in $OPT$ covers at most $\Delta+1$ vertices, we get $n \leq (\Delta+1)OPT$. In addition, $m \leq \Delta n$, which gives $OPT+\frac{\epsilon m}{(\Delta+1)\Delta} \leq OPT+\frac{\epsilon (\Delta+1)\Delta OPT}{(\Delta+1)\Delta} = (1+\epsilon)OPT$.

If $|E_{cut}|\geq\frac{\epsilon m}{2(\Delta+1)\Delta}$, Alice and Bob learn the whole graph and find an optimal solution. This requires exchanging a number of bits that is at most $m\log(n) \leq  |E_{cut}|2(\Delta+1)\Delta\log(n)/\epsilon$. Since $\Delta$ is a constant by the assumption on $\{G_{x,y}\}$, we have that $CC^{\{G_{x,y}\}}(P)=O(|E_{cut}|\log(n)/\epsilon)$, which implies that Theorem~\ref{generallowerboundtheorem} cannot give a lower bound for a distributed algorithm for deciding $P$ that is larger than $\Omega(1/\epsilon)$.
\end{proof}

\paragraph{MaxIS:}

\begin{claim}
\label{claim:limit-bounded-MaxIS}
Let $P$ be a predicate that implies a $(1-\epsilon)$-approximation for unweighted MaxIS. If $\{G_{x,y}\}$ is a family of lower bound graphs with bounded degree with respect to a function $f$ and the predicate $P$, then Theorem~\ref{generallowerboundtheorem} cannot give a lower bound for a distributed algorithm for deciding $P$ that is larger than $\Omega(1/\epsilon)$.
\end{claim}

\begin{proof}
Similarly to the proof of Claim~\ref{claim:limit-bounded-MVC}, there are two cases.
If $|E_{cut}|\leq\frac{\epsilon m}{(\Delta+1)\Delta}$, Alice and Bob compute optimal solutions in their sides of the graph not including vertices of the cut (i.e., each of them takes maximum number of independent internal vertices), and they take the union of the solutions. Since all the vertices are internal vertices in $V_A,V_B$, this gives an independent set. An optimal solution has at most $|E_{cut}|\leq\frac{\epsilon m}{(\Delta+1)\Delta}$ vertices in the cut, and the number of its internal vertices in $V_A,V_B$ is at most as in the solution that Alice and Bob construct. This means that the size of their solution is at least $OPT-\frac{\epsilon m}{(\Delta+1)\Delta}$. In addition, $OPT \geq \frac{n}{\Delta+1}$, since if the size of the maximum independent set is smaller, there is a vertex that is not in the set and is not a neighbour of a vertex in the set, which contradicts maximality. Since $m \leq \Delta n$ this gives $\frac{\epsilon m}{(\Delta+1)\Delta} \leq \frac{\epsilon OPT (\Delta+1) \Delta}{(\Delta+1)\Delta} = \epsilon OPT$, which shows that Alice and Bob get a $(1-\epsilon)$-approximation.

If the cut size it at least $\frac{\epsilon m}{(\Delta+1)\Delta}$, Alice and Bob solve the problem optimally by learning the graph using $O(|E_{cut}|\log(n)/\epsilon)$ bits, as before.
\end{proof}

\subsubsection{General graphs}
\label{sec:limitations_MDS}
Here we leverage Corollary~\ref{theorem:alice-bob-cut} for showing limitations of using Theorem~\ref{generallowerboundtheorem} with general graphs.

\paragraph{Max-cut:}
We show limitations for obtaining a $(1-\epsilon)$-approximation for unweighted max-cut, and for obtaining a $2/3$-approximation for weighted max-cut.

\begin{claim}
\label{claim:limit-general-MaxCut}
Let $P$ be a predicate that implies a $(1-\epsilon)$-approximation for unweighted max-cut. If $\{G_{x,y}\}$ is a family of lower bound graphs with respect to a function $f$ and the predicate $P$, then Theorem~\ref{generallowerboundtheorem} cannot give a lower bound for a distributed algorithm for deciding $P$ that is larger than $\Omega(1/\epsilon)$.
\end{claim}

\begin{proof}
We consider two possible cases. If $|E_{cut}|\leq\frac{\epsilon m}{2}$, let $E_A$ be the set of edges between vertices of $V_A$, and $E_B$ be the set of edges between vertices of $V_B$. Alice provides a solution for $V_A$ that is optimal for $E_A$, and Bob provides a solution for $V_B$ that is optimal for $E_B$.
In the optimal solution, the number of cut edges in $E_A$ or $E_B$ can only be smaller than the solution that Alice and Bob provide, but the optimal solution may include $\frac{\epsilon m}{2}$ additional edges from $E_{cut}$. Hence, the solution of Alice and Bob is of size at least $OPT-\frac{\epsilon m}{2}$. Since It is well known that in any graph, the size of an optimal solution is at least $\frac{m}{2}$ (consider a random assignment of vertices to different sides), the solution obtained has size at least $OPT - \epsilon OPT = (1-\epsilon)OPT$.

If $|E_{cut}|\leq\frac{\epsilon m}{2}$, Alice and Bob solve the problem optimally by learning the graph using $m\log(n)=O(|E_{cut}|\log(n)/\epsilon)$ bits.
\end{proof}

In \emph{weighted} graphs, Alice and Bob can get a $\frac{2}{3}$-approximation for max-cut. This follows from~\cite[Section 2.3]{kogan2015sketching}, where it is shown that the following simple protocol gives a $\frac{2}{3}$-approximation (this is shown even in a more restricted model with one-way communication). Divide the edges of $G=(V,E)$ to two disjoint sets $E_A,E_B$ where $E_A$ is the set of edges in $V_A$, and $E_B$ is the union of the set of edges in $V_B$ and the set of edges $E_{cut}$. Alice finds an optimal solution $C_A$ of $(V,E_A)$ (vertices that do not touch $E_A$ are added to one of the sides of the cut arbitrarily), and Bob finds an optimal solution $C_B$ of $(V,E_B)$. A cut can be seen as a function $C: V \rightarrow \{0,1\}$, where each vertex gets a side in the cut. In~\cite{kogan2015sketching}, it is shown that at least one of $\{C_A,C_B,C_A \oplus C_B\}$ gives a $\frac{2}{3}$-approximation for the max-cut. Alice and Bob can compute the values of all of these three cuts and get a $\frac{2}{3}$-approximation, by exchanging $O(|E_{cut}|\log(n))$ bits, which gives the following.

\begin{claim}
\label{claim:limit-general-wMaxCut}
Let $P$ be a predicate that implies a $2/3$-approximation for weighted max-cut. If $\{G_{x,y}\}$ is a family of lower bound graphs with respect to a function $f$ and the predicate $P$, then Theorem~\ref{generallowerboundtheorem} cannot give a lower bound for a distributed algorithm for deciding $P$ that is larger than $\Omega(1)$.
\end{claim}

\paragraph{MVC:}
We show the following two claims. First, Alice and Bob can get a $\frac{3}{2}$-approximation to weighted MVC by exchanging $O(|E_{cut}|\log(n))$ bits. Second, by exchanging $O(OPT |E_{cut}|\log(n))$ bits they can get a $(1+\epsilon)$-approximation for unweighted MVC, where $OPT$ is the size of optimal solution. In particular, since $OPT \leq n$ we cannot use this framework to show any super-linear lower bound for approximating unweighted MVC. Interestingly, for the MaxIS problem we do show a nearly-quadratic lower bound for a constant approximation (Section~\ref{section:approx(b)}). While a lower bound for obtaining an exact solution for one of these problems implies directly a lower bound for the other since the complement of a MaxIS is an MVC, an approximation for one does not imply an approximation for the other in a general graph.

\begin{claim}
\label{claim:limit-general-wMVC3/2}
Let $P$ be a predicate that implies a $3/2$-approximation for weighted MVC. If $\{G_{x,y}\}$ is a family of lower bound graphs with respect to a function $f$ and the predicate $P$, then Theorem~\ref{generallowerboundtheorem} cannot give a lower bound for a distributed algorithm for deciding $P$ that is larger than $\Omega(1)$.
\end{claim}

\begin{proof}
Let $OPT_A$ be the cost of an optimal vertex cover to all the internal edges in $V_A$, and define $OPT_B$ accordingly with respect to $V_B$. Note that since these optimal partial solutions are disjoint, we have $OPT_A+OPT_B \leq OPT$, or equivalently $\min\{OPT_A,OPT_B\} \leq \frac{OPT}{2}$. Alice and Bob can compute $OPT_A$ and $OPT_B$ respectively, and exchange $\log(n)$ bits for computing $\min\{OPT_A,OPT_B\}$. Assume w.l.o.g that $OPT_A = \min\{OPT_A,OPT_B\}$, then Alice adds to the solution the respective optimal cover of cost at most $OPT_A$. Bob adds to the solution an optimal cover for all the edges touching $V_B$ (including edges in $E_{cut}$, and tells Alice which vertices in $V_B$ are in it. This requires at most $|E_{cut}|\log(n)$ bits. The size of the solution Bob adds is clearly at most $OPT$. Hence, the total cost of their solution is at most $OPT_A + OPT \leq \frac{3}{2} OPT$, and is obtained by exchanging $O(|E_{cut}|\log(n))$ bits.
\end{proof}

We next show that Alice and Bob can always get a $(1+\epsilon)$-approximation to unweighted MVC in $O(\frac{OPT}{\epsilon})$ rounds. This is inspired by~\cite{DBLP:journals/corr/abs-1807-04900} that shows an $O(OPT^2)$ algorithm for MVC in the \cgst{} model. We show that similar ideas lead to a protocol for Alice and Bob which exchanges $O(|E_{cut}|\log(n)OPT/\epsilon)$ bits.

\begin{claim}
\label{claim:limit-general-MVCeps}
Let $P$ be a predicate that implies a $(1+\epsilon)$-approximation for unweighted MVC. If $\{G_{x,y}\}$ is a family of lower bound graphs with respect to a function $f$ and the predicate $P$, then 
Theorem~\ref{generallowerboundtheorem} cannot give a lower bound for a distributed algorithm for deciding $P$ that is larger than $\Omega(n/\epsilon)$.
\end{claim}

\begin{proof}
First, Alice and Bob run the protocol from Claim~\ref{claim:limit-general-wMVC3/2} and get an estimate $OPT \leq k \leq \frac{3}{2}OPT$ on the size of $OPT$.
If $|E_{cut}|<\frac{\epsilon k}{3}$, each of Alice and Bob covers optimally the edges of $V_A,V_B$, respectively, and they add also all the vertices touching $E_{cut}$. This gives a solution of size at most $OPT+2|E_{cut}|\leq OPT + \frac{2 \epsilon k}{3} \leq (1+\epsilon)OPT$.

Otherwise, $|E_{cut}|\geq\frac{\epsilon k}{3}$. Alice and Bob take each vertex of degree greater than $k$ into the cover, which requires no communication. Note that each of these vertices must be in any optimal solution, as otherwise there is a vertex with at least $k+1$ neighbors not in the cover and so the optimal solution must consist of all of these neighbors, which would contradict $OPT \leq k$. Then, Alice and Bob tell each other which cut vertices in $V_A,V_B$ are taken to the cover. They remove from the graph all the edges that are covered by vertices taken into the solution. In the remaining graph there are at most $k^2$ edges: the degree of each vertex is at most $k$, and there is an optimal cover of size at most $k$ that covers all the edges. Since $|E_{cut}|\geq\frac{\epsilon k}{3}$, Alice and Bob can learn the whole remaining graph and solve the problem optimally. The number of bits required for learning this remaining graph is $k^2\log(n) \leq k \cdot k \cdot \log(n) \leq \frac{3|E_{cut}|}{\epsilon}\cdot \frac{3OPT}{2} \cdot \log(n) = O(|E_{cut}|\log(n)OPT/\epsilon)$ bits. The proof is completed by bounding $OPT$ with $n$.
\end{proof}

\paragraph{MDS and MaxIS:}
Finally, simple arguments show that Alice and Bob can always get a 2-approximation for weighted MDS and a $\frac{1}{2}$-approximation for weighted MaxIS in general graphs. These approximations are much better than the currently best approximations known in the \cgst{} model.

\begin{claim}
\label{claim:limit-general-MDS}
Let $P$ be a predicate that implies a $2$-approximation for weighted MDS. If $\{G_{x,y}\}$ is a family of lower bound graphs with respect to a function $f$ and the predicate $P$, then  Theorem~\ref{generallowerboundtheorem} cannot give a lower bound for a distributed algorithm for deciding $P$ that is larger than $\Omega(1)$.
\end{claim}

\begin{proof}
To get a 2-approximation for MDS, each of Alice and Bob covers optimally all the vertices in $V_A,V_B$, respectively, by using possibly vertices in the cut. Each of these solutions is of cost at most $OPT$, hence their union gives a 2-approximation. Alice and Bob tell each other which vertices on the other side they take into the solution, which requires at most $O(|E_{cut}|\log(n))$ bits.
\end{proof}

\begin{claim}
\label{claim:limit-general-MaxIS}
Let $P$ be a predicate that implies a $\frac{1}{2}$-approximation for weighted MaxIS. If $\{G_{x,y}\}$ is a family of lower bound graphs with respect to a function $f$ and the predicate $P$, then 
Theorem~\ref{generallowerboundtheorem} cannot give a lower bound for a distributed algorithm for deciding $P$ that is larger than $\Omega(1)$.
\end{claim}

\begin{proof}
For MaxIS, each of Alice and Bob finds an optimal solution for $V_A,V_B$, respectively, and they take to the solution the maximum of them, which requires at most $O(\log(n))$ bits to compute. At least one of their solutions is of size at least $\frac{1}{2}OPT$, which shows a $\frac{1}{2}$-approximation.
\end{proof}

\subsection{Limitations using non-determinism}
\label{subsec:limit-nondet}

In this section, we prove limitations of the framework in providing lower bounds for several central problems in $P$ as maximum cardinality matching (MCM), max flow, min $s$-$t$ cut and weighted $(s,t)$-distance. We also consider several verification problems. The main tools we introduce are connections to \emph{non-deterministic} complexity, and \emph{proof labeling schemes} (PLS).

Corollary~\ref{theorem:alice-bob-cut} allows us to obtain a wide range of limitations, by providing protocols for Alice and Bob that decide $P$ on the given families of lower bound graphs, and hence provide upper bounds for $CC^{G_{x,y}}(P)$ for the different predicates $P$ in question. Another way to obtain such an upper bound is to resort to non-determinism. We prove that having an efficient nondeterministic protocol that allows Alice and Bob to verify a graph property
and an efficient nondeterministic protocol for the complementary property, proves a limitation of this framework in giving certain lower bounds on the communication complexity of deterministically (or randomly) deciding this property.

In a nondeterministic protocol, Alice and Bob use auxiliary bit strings, which we think of as nondeterministic strings. These strings are assumed to be optimal with respect to the protocol and input, in the sense that they minimize the number of bits exchanged when $f(x,y)=\true$. For the case of $f(x,y)=\false$, instead of discussing the co-nondeterministic communication complexity, we discuss the nondeterministic communication complexity of the complement function, $\notf$, defined as $\notf(x,y)=\true$ iff $f(x,y)=\false$. The \emph{nondeterministic communication complexity}, $CC^N(f)$, of computing $f$ is defined analogously to the deterministic/randomized complexities: it is the minimum $CC(\pi)$, taken over all non-deterministic protocols $\pi$ for computing $f$.

A central theorem in communication complexity asserts that for any Boolean $f$, it holds that $CC(f)=O(CC^N(f)\cdot CC^N(\notf))$ (See, e.g.,~\cite[Theorem 2.11]{KushilevitzN:book96} and the references therein). This means that we can bound $CC(f)$ by bounding $CC^N(f)$ and $CC^N(\notf)$, and thus if Theorem~\ref{generallowerboundtheorem} is used with any family of lower bound graphs $G_{x,y}$ then it cannot obtain a lower bound that is larger than $\Omega(CC^{N}(f)\cdot CC^{N}(\notf) /|E_{\cut}|\log{n})$. However, we can sometimes limit the hardness results even more.

For a function $f: \{0,1\}^K\times \{0,1\}^K \rightarrow \{\true,\false\}$, we denote $$\Gamma(f)= CC(f)/\max\{CC^{N}(f), CC^{N}(\notf)\}.$$ Since $CC(f)=O(CC^{N}(f)\cdot CC^{N}(\notf))$, we can always say that $\Gamma(f)=O(\sqrt{CC(f)})$ and in particular $\Gamma(f)=O(\sqrt{K})$. While there are functions for which $\Gamma(f)=\Theta(\sqrt{CC(f)})$, and in particular $\Gamma(f)=\Theta(\sqrt{K})$, for some functions $\Gamma(f)$ is known to be much smaller. One example is the widely used (also in this paper) $f=\disj_K$, for which $CC^N(\disj_K)=\Theta(K)$ (this follows, e.g., from ~\cite[Example 1.23 and Definition 2.3]{KushilevitzN:book96}), which implies $\Gamma(\disj_K)=O(1)$. The nondeterministic communication complexity of $\lnot\disj_K$ is only $O(\log K)$. Another example is the equality function over $K$ bits, $\eq_K$, whose output is $\true$ if and only if $x=y$. It is known that $CC(\eq_K)=\Theta(K)$ and that $CC^R(\eq_K)=O(\log K)$. Moreover, $CC^N(\eq_K)=\Theta(K)$, and $CC^N(\lnot\eq_K)=O(\log K)$.

\begin{claim}
	\label{theorem:alice-bob-max}
	Fix a graph predicate $P$.
	No Boolean function $f$
	and a family $\{G_{x,y}\}$ of lower bound graphs with respect to $f$ and $P$
	can give a lower bound for distributed algorithms for $P$ that is larger than $\Omega(\max\{CC^{N}(f), CC^{N}(\notf)\}\Gamma(f) /|E_{\cut}|\log{n})$
	using Theorem~\ref{generallowerboundtheorem}.
\end{claim}

\begin{proof}
	Fix a function $f: \{0,1\}^K\times \{0,1\}^K \rightarrow \{\true,\false\}$ and a predicate $P$. Let $\{G_{x,y}\}$ be a family of lower bound graphs with respect to $f$ and $P$. Denote by $t=\alpha CC(f)/|E_{\cut}|\log n$ the lower bound obtained by applying Theorem~\ref{generallowerboundtheorem} to $\{G_{x,y}\}$, where $\alpha$ is the constant that hides in the asymptotic notation.
	
	We can write
	\begin{align*}
	t &=	\alpha\frac{CC(f)}{\size{E_{\cut}}\log n}\\
	&=\alpha\frac{\max\{CC^{N}(f), CC^{N}(\notf)\}}{\size{E_{\cut}}\log n}\cdot\frac{CC(f)}{\max\set{CC^N(f),CC^N(\notf)}}\\
	&=\alpha\frac{\max\{CC^{N}(f), CC^{N}(\notf)\}}{\size{E_{\cut}}\log n}\cdot\Gamma(f),
	\end{align*}
	as needed.
\end{proof}

When using Claim~\ref{theorem:alice-bob-max}, we will mainly bound $CC^{N}(f)$ and $CC^{N}(\notf)$ by bounding $CC^{\{G_{x,y}\},N}(P)$ and $CC^{\{G_{x,y}\},N}(\notP)$, respectively, which we define similarly to their deterministic counterpart. Fix a function $f: \{0,1\}^K\times \{0,1\}^K \rightarrow \{\true,\false\}$ and a predicate $P$. If $\{G_{x,y}\}$ is a family of lower bound graphs with respect to $f$ and $P$, then $CC^{\{G_{x,y}\},N}(P)$ is the non-deterministic complexity for deciding $P$ on graphs in $\{G_{x,y}\}$. As in the deterministic/randomized cases, since the answer for $P$ on $\{G_{x,y}\}$ determines the output of $f$ on $(x, y)$, this implies that $CC^N(f) \leq CC^{\{G_{x,y}\},N}(P)$. Similarly, $CC^N(\notf) \leq CC^{\{G_{x,y}\},N}(\notP)$. This directly gives the following (which in general may be coarser than Claim~\ref{theorem:alice-bob-max}).

\begin{corollary}
	\label{theorem:alice-bob-cut-nondet}
	Fix a graph predicate $P$.
	No function $f: \{0,1\}^K\times \{0,1\}^K \rightarrow \{\true,\false\}$
	and a family $\{G_{x,y}\}$ of lower bound graphs with respect to $f$ and $P$
	can give a lower bound for distributed algorithms for $P$ that is larger than $\Omega(\max\{CC^{\{G_{x,y}\},N}(P), CC^{\{G_{x,y}\},N}(\notP)\}\Gamma(f) /|E_{\cut}|\log{n})$
	using Theorem~\ref{generallowerboundtheorem}.
\end{corollary}

\subsubsection{Impossibility using direct non-deterministic protocols}
We begin by applying Claim~\ref{theorem:alice-bob-max} for showing that the framework of Theorem~\ref{generallowerboundtheorem} cannot give any super-constant lower bound for Max $(s,t)$-flow.
\paragraph{Max $(s,t)$-flow and Min $(s,t)$-cut:}
\begin{claim}
	\label{claim:limit-nondet-MaxFlow}
	Let $k$ be some parameter, and let $P$ be a predicate that says that the Max $(s,t)$-flow solution is of size $k$.
Using a Boolean function $f$, Theorem~\ref{generallowerboundtheorem} cannot give a lower bound for a distributed algorithm for deciding $P$ that is larger than $\Omega(\Gamma(f))$.
\end{claim}

As the size of the Max $(s,t)$-flow is the same as the size of the
Min $(s,t)$-cut, the same limitations obviously apply for Min $(s,t)$-cut.

\begin{proof}
Given a Boolean function $f$, let $G_{x,y}$ be a family of lower bound graphs with respect to $f$ and $P$.

	Let $MF=MF(G)$ denote the size of a maximum $(s,t)$-flow in a graph $G$. We show a nondeterministic protocol for the predicate $MF\geq k$ and a similar protocol for $MF<k$, each communicating $O(|E_{\cut}|\log n)$ bits on graphs in $\{G_{x,y}\}$.
	
	For $MF\geq k$, Alice gets a nondeterministic string describing the flow on the edges touching $V_A$, and Bob gets the flow on the edges touching $V_B$. Alice sends to Bob the flow on each of the edges in $E_{\cut}$, and then they both locally verify that the flow is preserved for each vertex excluding $s$ and $t$, and that the total $(s,t)$-flow is at least $k$. Note that $s$ and $t$ may belong to the same player or to different players.
	
	For $MF<k$, the players receive nondeterministic strings containing an encoding of an $(s,t)$-cut $E_C=E(C,\bar{C})$, where $s \in C$: This consists of one bit for each vertex, indicating whether it is in $C$ or not. Alice sends to Bob all the vertices in $V_A\cap C$ that touch $E_{\cut}$, and also the total weight of the $(s,t)$-cut edges within $V_A$. Bob sums this weight with the weight of $(s,t)$-cut edges within $V_B$ and in $E_{\cut}$, and verifies that it is smaller than $k$.
	
	The communication complexities of both protocols are $O(|E_{\cut}|\log n)$ bits. Corollary~\ref{theorem:alice-bob-cut-nondet} now implies that Theorem~\ref{generallowerboundtheorem} cannot give a lower bound that is larger than $\Omega(\Gamma(f))$ for a distributed algorithm for deciding $P$.
\end{proof}

In particular, this means that for Max $(s,t)$-flow, Theorem~\ref{generallowerboundtheorem} cannot give any lower bound that is super-constant using $\disj_K$ or $\eq_K$. For any other function $f$, the lower bound cannot be super-linear in $\sqrt{K}$, so for example, if $K=\Theta(n^2)$ then it cannot be super-linear in $n$, and if $K=\Theta(n)$ then it cannot be greater than $\Omega(\sqrt{n})$.

\subsubsection{Impossibility using proof labeling schemes}
\label{subsub:limit-pls}
Rather than directly giving protocols for bounding $CC^{\{G_{x,y}\},N}(P)$ and $CC^{\{G_{x,y}\},N}(\notP)$, we can also leverage the concept of \emph{proof labeling schemes} (PLS)~\cite{KormanKP10}.
A proof labeling scheme for a graph predicate $\calP$ is a distributed mechanism for verifying that $P$ holds on the communication graph. Each vertex is given a label and in the distributed verification procedure, each vertex executes an algorithm which receives as input its state ($ID$, list of neighbors and node/edge weights), its label, and the labels of its neighbors, and outputs accept or reject. A correct PLS must have the following property: if $\calP$ holds then there exists a label assignment such that all vertices accept, while if $\calP$ does not hold, in every label assignment at least one vertex rejects. The proof-size of a PLS is the maximal size of a label assigned in it.

\begin{theorem}
	\label{theorem:PLSlowerboundtheorem}
	Fix a function $f:\{0,1\}^K \times \{0,1\}^K \to \{\true,\false\}$ and a predicate $P$. If there exists a family of lower bound graphs $\{G_{x,y} \}$ w.r.t $f$ and $P$ then
    $CC^{\{G_{x,y}\},N}(P)=O(\plssize(P)|E_{\cut}|)$.	
\end{theorem}

We note that the above directly implies that lower bounds can be carried over from Theorem~\ref{generallowerboundtheorem} to lower bounds on the proof-size of distributed predicates, which generalizes several previous lower bounds for PLS~\cite{GoosS16,Censor-HillelPP17}.

\begin{proof}
	We show how Alice and Bob can use a given PLS to devise a nondeterministic protocol for the $\true$ instances of $P$ in which they communicate $O(\plssize(P)|E_{\cut}|)$ bits. Let $V_{\cut}=\cup E_{\cut}$ be the set of vertices appearing in edges that are in $E_{\cut}$.
	
	Given two inputs $x,y\in \{0,1\}^K$ and two nondeterministic strings to Alice and Bob, they simulate their corresponding vertices of the graph $G_{x,y}$: Alice simulates $V_A$ and Bob simulates $V_B$. They interpret their nondeterministic strings as the labels of the PLS. In order to simulate the verification process, Alice sends to Bob the labels of the vertices in $V_A\cap V_{\cut}$ in a predefined order, and Bob similarly sends to Alice the labels of $V_B\cap V_{\cut}$. Both players simulate the local verification process of each node, without any further communication, and send a $\false$ bit to the other player if there exists a vertex that rejects in the PLS. If at least one of the players detected a rejecting node, they output $\false$, and otherwise they output $\true$.
	
	The correctness follows from the definitions of a family of lower bound graph and of PLS, as follows. If $G_{x,y}$ satisfies $P$, then the PLS implies that there exists an assignment of labels to vertices that causes all of them accept. Thus, there exist nondeterministic strings for Alice and Bob that causes all the vertices to accept and thus Alice and Bob correctly answer $\true$. If $G_{x,y}$ does not satisfy $P$, any assignment of labels causes at least one vertex to reject, and thus any assignment of nondeterministic strings to Alice and Bob results in at least one rejecting vertex in the simulation, and Alice and Bob correctly answering $\false$.
	
	For the communication complexity, note that $\size{V_A\cap V_{\cut}}\leq \size{E_{\cut}}$, and similarly $\size{V_B\cap V_{\cut}}\leq \size{E_{\cut}}$. Alice and Bob thus communicate at most $2\size{E_{\cut}}$ labels, of $O(\plssize(P))$ bits each, and thus they correctly answer on the $\true$ instances for $P$ while communicating $O(\plssize(P)|E_{\cut}|)$ bits. 	
%
%
\end{proof}

We can now plug Theorem~\ref{theorem:PLSlowerboundtheorem} into Corollary~\ref{theorem:alice-bob-cut-nondet} and obtain the following.
\begin{corollary}
	\label{theorem:pls-to-alice-bob}
	Fix a predicate $P$.
	Using a Boolean function $f$,
	Theorem~\ref{generallowerboundtheorem} cannot give a lower bound for a distributed algorithm that is larger than $\Omega(\max\{\plssize(P), \plssize(\notP)\}\Gamma(f) /\log{n})$.
\end{corollary}

Corollary~\ref{theorem:pls-to-alice-bob}, along with known PLS constructions or new ones, imply the limitations of Theorem~\ref{generallowerboundtheorem} for many fundamental problems, which we overview next.

\paragraph{Maximum matching:} We consider the maximum cardinality matching (MCM) problem, and a predicate $P$ which asserts that $\size{MCM}\geq k$, for which we show that no super-constant lower bound can be proven using Theorem~\ref{generallowerboundtheorem} with $\disj$ or $\eq$, and the corresponding generalization for $f$ and $K$ as captured by $\Gamma(f)$.

\begin{claim}
	\label{claim:limit-mcm}
	Let $k$ be some parameter, and let $P$ be a predicate that says that $\size{MCM}\geq k$.
	Using a Boolean function $f$,
	Theorem~\ref{generallowerboundtheorem} cannot give a lower bound for a distributed algorithm for deciding $P$ that is larger than $\Omega(\Gamma(f))$.
\end{claim}

\begin{proof}
	By Corollary~\ref{theorem:pls-to-alice-bob}, it is enough to show a PLS for $\size{MCM}\geq k$ and a PLS for $\size{MCM}<k$, both of size $O(\log n)$.
	
	A PLS for $\size{MCM}\geq k$ can be designed using standard techniques: the labels mark the edges of a matching of the desired size, and are also used to count the number of matched vertices over a spanning tree. In the verification procedure, each vertex verifies that it is matched at most once, and that the count of matched vertices in its label equals the sum of counters of its tree children, or equals this sum plus one if it is also matched. The root also verifies that its counter is at least twice the desired value, $k$ (the two factor comes from that fact that each edge is counted by both its endpoints).
	
	A PLS for $\size{MCM}<k$ with logarithmic labels was presented in~\cite{Censor-HillelPP17}, using a duality theorem of Tutte and Berge.
\end{proof}

\paragraph{Weighted $(s,t)$-distance:} Consider the weighted $(s,t)$-distance problem, for which a predicate $P$ says that $\wdist(s,t)\geq k$, for two given vertices $s$ and $t$ and an integer $k$. We show that no super-constant lower bound can be proven using Theorem~\ref{generallowerboundtheorem} with $\disj$ or $\eq$, and the corresponding generalization for $f$ and $K$ as captured by $\Gamma(f)$.

\begin{claim}
	\label{claim:limit-wDist}
	Let $k$ be some parameter, and let $P$ be a predicate that says that $\wdist(s,t)\geq k$.
	Using a Boolean function $f$,
	Theorem~\ref{generallowerboundtheorem} cannot give a lower bound for a distributed algorithm for deciding $P$ that is larger than $\Omega(\Gamma(f))$.
\end{claim}

\begin{proof}
	By Corollary~\ref{theorem:pls-to-alice-bob}, it is enough to show a PLS for $\wdist(s,t)\geq k$ and a PLS for $\wdist(s,t)<k$, both of size $O(\log n)$. A simple PLS works for both problems: give each vertex a label containing its weighted distance from $s$. In both cases, each vertex $v\neq s$ verifies that its distance is the minimum, among all its neighbors $u$, of the label of $u$ plus the weight of the edge $(u,v)$, and $s$ verifies that its label is $0$. For $\wdist(s,t)\geq k$, the vertex $t$ accepts if and only if its label is at least $k$, and for $\wdist(s,t)<k$ it accepts if and only if its label is smaller than $k$.
\end{proof}

\subsubsection{Limitations of proving lower bounds for verification problems}
In this section, we discuss several subgraph-verification problems. For all these problems, a tight bound of $\tilde{\Theta}(\sqrt{n}+D)$ rounds was proven~\cite{Dassarmaetal12}, and we show that the method presented in Theorem~\ref{generallowerboundtheorem} is incapable of proving even a super-constant lower bound for them. Roughly speaking, one can consider the method of~\cite{Dassarmaetal12} as a generalization of Theorem~\ref{generallowerboundtheorem}, which is allowed to use a family of lower bound graphs in which rather than having a fixed partition of $V$ into $V_A$ and $V_B$ that are simulated by Alice and Bob, the players simulate subsets of $V$ that can vary from one simulated round to the other. Moreover, the subsets simulated for a certain round do not have to be a partition: there can be vertices that are simulated by both players, and vertices whose simulation stops before they produce an output. We refer the reader to~\cite{Dassarmaetal12} for a thorough exposition. In this sense, the framework that is given in our paper is a special case of the framework of~\cite{Dassarmaetal12}, and the implication of the limitations that we give below is a \emph{separation} between the abilities of the special \emph{fixed-simulation} case and the more general framework.

In all problems below, the graph $G=(V,E_G)$ has a marked subgraph $H=(V,E_H)$ on the same set of vertices, and  sometimes also marked vertices $s$ and $t$ or an edge $e$. Each vertex knows which of its edges is in $H$, if one of its edges is $e$, and if it is $s$ or $t$. The goal is to decide if $H,e,s,t$ have the given property. We use $G\setminus H$ as a shorthand for $(V,E_G\setminus E_H)$, and $H\setminus \set{e}$ for $(V,E_H\setminus \set{e})$.

\begin{lemma}
	\label{lemma:limit-verification}
	Let $P$ be a predicate for any of the problems listed below.
Using a Boolean function $f$,
	Theorem~\ref{generallowerboundtheorem} cannot give a lower bound for a distributed algorithm for deciding $P$ that is larger than $\Omega(\Gamma(f))$.
	\begin{enumerate}
		\item
		Connected spanning subgraph: $H$ is connected, and each vertex in $G$ has non-zero degree in $H$.
		\item
		Cycle containment: $H$ contains a cycle.
		\item
		$e$-cycle containment: $H$ contains a cycle through $e$.
		\item
		Bipartiteness: $H$ is bipartite.
		\item
		$(s,t)$-connectivity: $s$ and $t$ are connected by a path in $H$.
		\item
		Connectivity: $H$ is connected.
		\item
		Cut: $H$ is a cut of $G$, i.e., the graph $G\setminus H$ is not connected
		\item
		Edge on all paths: $s$ and $t$ are not in the same connected component of $H\setminus \set{e}$.
		\item
		$(s,t)$-cut: $s$ and $t$ are not in the same connected component of $G\setminus H$.
		\item
		Hamiltonian cycle: $H$ is a Hamiltonian cycle in $G$.
		\item
		Spanning tree: $H$ is a spanning tree in $G$.
		\item
		Simple path: $H$ is a simple path.
	\end{enumerate}
\end{lemma}

We emphasis that these are all verification problems. For example, in Hamiltonian cycle verification, the algorithm is given a subgraph as an input, and needs to verify this subgraph is a Hamiltonian cycle, as opposed to the computation problem discussed in Section~\ref{subsec:Hamiltonian}, where the algorithm needs to decide if such a cycle exists.

\begin{proof}
	The proof follows the lines of the previous proofs: By Corollary~\ref{theorem:pls-to-alice-bob}, it is enough to show a PLS for $P$ and a PLS for $\notP$, both of size $O(\log n)$. For this, we employ known schemes (see, e.g.,~\cite{KormanKP10,BaruchFP15,Censor-HillelPP17}) for acyclicity, spanning tree, and for pointing to a desired vertex or subgraph, without reproving them. All schemes have proof size $O(\log n)$.
	
	\begin{enumerate}
		\item
		A PLS for a connected spanning subgraph is obtained by having the labels form a spanning tree of $H$, i.e., an id of the root and the distance from it.
		The vertices compare their claimed root id for equality
		(in $G$, so they all must have the same root id),
		and each vertex verifies all its neighbors have distance differing from its distance by at most 1, and at least one neighbor with distance 1 less.
		The root verifies its distance is 0.
		
		A PLS for proving that $H$ is not a connected spanning subgraph, it is enough to prove it is not connected, because if it is not spanning then at least one vertex will detect this, without any labels.
		To show it is not connected, give the vertices of one connected component of$h$ a $0$ label, and a $1$ label to the rest, and describe two spanning trees in $G$ pointing to vertices of different labels.
		The vertices verify the two spanning trees as before, which guarantees the existence of a $0$ labels and of a $1$ label.
		In addition, they make sure no edge in $H$ connects a $0$ label and a $1$ label.
		
		\item
		A PLS for cycle containment gives each vertex of $G$ its distance from the cycle. This promises that a cycle exists, and the vertices with a $0$ label verify they have exactly two neighbors in $H$ that have a $0$ label as well.
		A PLS for the non existence of a cycle uses a standard scheme for acyclicity~\cite{BaruchFP15} on $H$ alone.
		
		\item
		A PLS for $e$-cycle containment in $H$ uses the previous scheme works, where the vertices of $e$ also verify they have distance $0$.
		For a PLS for the negation of the above, note that if the predicate does not hold, then $H\setminus\set{e}$ is disconnected. Mark vertices in one side by $0$, the other by $1$, and the vertices verify that the only edge in $H$ that connects both labels is $e$.
		
		\item
		A PLS for bipartiteness provides the vertices with a $2$-coloring.
		A PLS for non-bipartiteness point to a cycle as before, and enumerates the cycle vertices consecutively. A cycle vertex (distance $0$ form the cycle) enumerated $i$ verifies it has two cycle neighbors enumerated $i-1$ and $i+1$, except for two cycle vertices: the cycle vertex enumerated $1$ which has neighbors enumerated $2$ and $x$, for some odd $x$, and the cycle vertex enumerated $x$
		which has neighbors enumerated $x-1$ and $1$.
		
		\item
		A PLS for $(s,t)$-connectivity gives each vertex its distance from $s$ in $H$.
		A PLS for non-$(s,t)$-connectivity marks all the vertices in the connected component of $s$ in $H$ by $0$, and marks the rest by $1$.
		
		\item
		A PLS for connectivity gives all the vertices their distance in $H$ from some vertex $s$, and its id.
		A PLS for non-connectivity marks the vertices of one connected component in $H$ by $0$ and the rest by $1$, and creates two spanning trees rooted at vertices of opposite marks.
		The vertices verify the spanning trees, and that no edge in $H$ has both $0$ and $1$ endpoints.
		
		\item
		A PLS for cut verification marks the sides by $0$ and $1$, and the vertices verify that all non-$H$ edges are monochromatic. A PLS for its negation shows a spanning tree of $G\setminus H$.
		
		\item
		A PLS for the predicate of having the edge $e$ on all paths between $s$ and $t$ marks the vertices of the connected component of $s$ in $H\setminus\set{e}$ by $0$ and the rest by $1$. All vertices verify that edges in $E_H\setminus\set{e}$ are monochromatic, and $s$ and $t$ verify they have labels $0$ and $1$ respectively. A PLS for its negation marks an $(s,t)$-path in $H\setminus\set{e}$ by giving each vertex its distance from the path (in $G$). Each path ($0$-distance) vertex verifies it has exactly $2$ neighboring path vertices, except for $s$ and $t$ who have only one.
		
		\item
		A PLS for $(s,t)$-cut marks the connected components in $G\setminus H$ as in the scheme for cut, and has $s$ and $t$ verify they have labels $0$ and $1$ respectively. For a PLS for its negation, give each vertex its distance from $s$ in $G\setminus H$ (which could be $\infty$ for some vertices, but not for $t$).
		
		\item
		A PLS for Hamiltonian cycle enumerates the cycle vertices consecutively.
		For a PLS proving that $H$ is not a Hamiltonian cycle, either there is a vertex of degree different than $2$in $H$, or $H$ is split into several cycles.
		For the first case, have a spanning tree rooted at a vertex with a degree different than $2$ in $H$.
		For the second, have all vertices point to a single cycle in $H$, and enumerate the vertices of this cycle; the vertex with distance $0$ from the cycle and number $1$ in it should have neighbors $2$ and $x$, for some $x<n$,
		proving that the cycle is smaller than $n$.
		
		\item
		A PLS for a spanning tree uses the known spanning-tree scheme with the edges of $H$. A PLS for $H$ not being a spanning tree, gives all the vertices a category (1),(2), or (3), and according to it add the following:
		(1) If there is a vertex not spanned by $H$, give all vertices their distances from it.
		(2) If there is a cycle in $H$, give all vertices their distances from the cycle.
		(3) If $H$ is not connected, use the scheme for non-connectivity described above.
		
		\item
		A PLS for $H$ being a simple path enumerates the path vertices consecutively, and gives all the graph vertices the id of the vertex marked $1$.
		If $H$ is nonempty, every vertex in $H$ either has two neighbors in $H$ with $\pm1$ labels, one neighbor enumerated one less (for the endpoints not enumerated $1$),
		or label $1$ and the id given to graph all vertices.
		For a PLS for the negation use categories, as before: (1) If $H$ contains a cycle, give all vertices their distance from it. (2) If $H$ is not connected, use the scheme for non-connectivity described above.\qedhere
	\end{enumerate}
\end{proof}

\section{Discussion}

In this paper, we provide many near-quadratic lower bounds for exact computation, as well as some hardness of approximation results. 
However, in Section \ref{section:limitations(b)}, we also show the limitations of these lower bound techniques when it comes to approximations and to exact computation of some central problems in P. For example, while we show a near-quadratic lower bound for a $(7/8 +\epsilon)$ approximation to MaxIS, we show that the same approach cannot say anything about a $1/2$-approximation, although this problem may be much more difficult to approximate. 
When it comes to exact computation of problems in P, such as maximum matching and max flow, we show that the lower bounds techniques are limited, at least if we use reductions from standard problems in communication complexity, such as disjointness or equality.
It would be interesting to study whether different techniques or reductions from non-standard communication complexity problems can provide such lower bounds.

It is intriguing
to note that apart from the $(1-\epsilon)$-approximation for max-cut,
the other optimization problems addressed in this paper do not
admit any non-trivial algorithms that obtain an approximation
factor that is NP-hard to obtain in the sequential setting. There is
currently no common ground to blame for the big gaps in terms
of the approximation factors that one can obtain in the \cgst{} 
model within any sub-quadratic complexity. Understanding this
hardness is a major open research direction in this model.

\paragraph{Acknowledgments:} We would like to thank Seri Khoury and Eyal Kushilevitz for many valuable discussions. This project has received funding from the European Union’s Horizon 2020 Research And Innovation Program under grant agreement no.755839. Supported in part by the Binational Science Foundation (grant 2015803). Supported
in part by the Israel Science Foundation (grant 1696/14).
Ami Paz is supported by the Fondation Sciences Mathématiques
de Paris (FSMP).

\bibliographystyle{abbrv}
\bibliography{References}


\end{document}